\documentclass[journal]{IEEEtran} % letter double
\usepackage{amsmath, amsfonts, amssymb, amsthm}
\usepackage[dvips]{graphicx}
\usepackage{verbatim}
\usepackage{setspace}
\usepackage{bbm}
\usepackage{algorithmic} %format of the algorithm
\usepackage[ruled,vlined]{algorithm2e}
\usepackage{cite}
\usepackage{url}

\usepackage{breakurl}
\usepackage{changepage}
\usepackage{pdfpages}
\usepackage{color}
\usepackage{blkarray, bigstrut}
\usepackage{caption}
\usepackage{subcaption}

\usepackage{mathtools, cuted}
\usepackage{stfloats}

\usepackage[labelformat=simple]{subcaption}

\newtheorem{lemma}{Lemma}
\newtheorem{proposition}{Proposition}

\newcommand{\updatecolor}{black}

\IEEEoverridecommandlockouts

\begin{document}

\title{
    {Movable Antenna Aided Multiuser Communications: Antenna Position Optimization Based on Statistical Channel Information}
    % \vspace{-5pt}
}

\author{Ge Yan, 
        Lipeng Zhu,~\IEEEmembership{Member,~IEEE,}
        Rui Zhang,~\IEEEmembership{Fellow,~IEEE}
        % \vspace{-10pt}
%\thanks{This work was supported in part by the National Natural Science Foundation of China (NSFC) (Grant Nos. 91538204, 61571025, 61827901 and 91738301.), the Open Research Fund of Key Laboratory of Space Utilization, Chinese Academy of Scpiences (No. LSU-DZXX-2017-02).}
% \thanks{Part of this paper has been submitted to 2019 IEEE Global Communications Conference Workshops \cite{zhu2019FDUAVconf}.}

%%%%%%%% newest fund
% \thanks{This work is supported in part by Advanced Research and Technology Innovation Centre (ARTIC) of National University of Singapore under Research Grant R-261-518-005-720,  The Guangdong Provincial Key Laboratory of Big Data Computing, the National Natural Science Foundation of China (No. 62331022), Shenzhen Research Institute of Big Data (No. J00120230006), the 2022 Stable Research Program of Higher Education of China (No. 20220817144726001), and the Guangdong Major Project of  Basic and Applied Basic Research (No. 2023B0303000001). }
% \thanks{Part of this work was presented at the IEEE Global Communications Conference Workshops 2023, Kuala Lumpur, Malaysia~\cite{ref:my-conf-ver}.}
\thanks{G. Yan is with the NUS Graduate School, National University of Singapore, Singapore 119077, and also with the Department of Electrical and Computer Engineering, National University of Singapore, Singapore 117583 (e-mail: geyan@u.nus.edu). }
\thanks{L. Zhu is with the Department of Electrical and Computer Engineering, National University of Singapore, Singapore 117583 (zhulp@nus.edu.sg). }
\thanks{R. Zhang is with School of Science and Engineering, Shenzhen Research Institute of Big Data, The Chinese University of Hong Kong, Shenzhen, Guangdong 518172, China (e-mail: rzhang@cuhk.edu.cn). 
He is also with the Department of Electrical and Computer Engineering, National University of Singapore, Singapore 117583 (e-mail: elezhang@nus.edu.sg). }
}

% \thanks{This work was supported in part by the National Key Research and Development Program (Grant Nos. 2016YFB1200100), the National Natural Science Foundation of China (NSFC) (Grant Nos. 61571025 and 91538204), the Open Research Fund of Key Laboratory of Space Utilization, Chinese Academy of Sciences (No. LSU-DZXX-2017-02).}

\maketitle

\IEEEpeerreviewmaketitle

\begin{abstract}

    % 1st part: big background, 1 sentence. 
    The movable antenna (MA) technology has attracted great attention recently due to its promising capability in improving wireless channel conditions by flexibly adjusting antenna positions. 
    % 2nd part: common method in literature and main challenges, 1-2 sentences. 
    To reap maximal performance gains of MA systems, existing works mainly focus on MA position optimization to cater to the instantaneous channel state information (CSI). 
    However, the resulting real-time antenna movement may face challenges in practical implementation due to the additional time overhead and energy consumption required, especially in fast time-varying channel scenarios. 
    % 3rd part: summary of what we have done, 1-2 sentences. 
    To address this issue, we propose in this paper a new approach to optimize the MA positions based on the users' statistical CSI over a large timescale. 
    In particular, we propose a general field response based statistical channel model to characterize the random channel variations caused by the local movement of users. 
    Based on this model, a two-timescale optimization problem is formulated to maximize the ergodic sum rate of multiple users, where the precoding matrix and the positions of MAs at the base station (BS) are optimized based on the instantaneous and statistical CSI, respectively. 
    % 4th part: detailed explanations for our work, 2-3 sentences. 
    To solve this non-convex optimization problem, a log-barrier penalized gradient ascent algorithm is developed to optimize the MA positions, where two methods are proposed to approximate the ergodic sum rate and its gradients with different complexities. 
    % 5th part: insights and implications, 1-2 sentences. 
    Finally, we present simulation results to evaluate the performance of the proposed design and algorithms based on practical channels generated by ray-tracing. 
    The results verify the performance advantages of MA systems compared to their fixed-position antenna (FPA) counterparts in terms of long-term rate improvement, especially for scenarios with more diverse channel power distributions in the angular domain. 
    % The effectiveness of the proposed algorithm in exploiting the spatial degree of freedom of statistical channels is verified in simulation results conducted with urban ray-tracing data of the physical world, which is further proved to be more advantageous for more diverse angular power distributions of users. 
\end{abstract}
\vspace{-4pt}

\vspace{-6pt}
\section{INTRODUCTION}\label{sec:introduction}
    % 1st part: advances of MIMO and promising app of MA 
    \IEEEPARstart{O}{ver} the past few decades, multi-antenna or so-called multiple-input multiple-output (MIMO) technology has brought substantial improvements to the transmission rate and reliability of wireless communication systems. % as one of the most important technologies in modern mobile networks~\cite{ref:mimo-overview}. 
    By exploiting the spatial degrees of freedom (DoFs), significant beamforming and multiplexing gains can be obtained in MIMO systems, which enables the overall performance boost with more antennas deployed. 
    As the demand for high data rates and massive access further increases, MIMO has been extended to various forms, e.g., massive MIMO~\cite{ref:m-mimo-for-next-gen, ref:m-mimo-overview}, extremely large-scale MIMO (XL-MIMO)~\cite{ref:xl-mimo-tutorial-for-6g,ref:tutorial-near-field-xl-mimo}, and holographic MIMO (HMIMO)~\cite{ref:hmimo-tutorial, ref:hmimo-isac}. 
    However, integrating a large number of antennas into one array may incur excessive signal processing overhead and prohibitively high hardware costs due to the increased number of radio frequency (RF) chains and/or phase shifters~\cite{ref:m-mimo-overview, ref:m-mimo-hybrid-bf-survey,ref:large-scale-mimo-hybrid-bf}, posing great challenges to their applications in future wireless systems. 
    
    % 2nd part: recent works on MA to reap spatial flexibility
    In light of the limitations of conventional MIMO technologies, movable antenna (MA)~\cite{ref:ma-opportunities} has been introduced to wireless communication systems for fully exploiting the spatial variations of wireless channels via local antenna movement at transceivers. 
    Different from conventional fixed-position antennas (FPAs), each MA or MA array can flexibly adjust its position for reaping more favourable channels. 
    By leveraging such DoFs in antenna movement, the MA-aided wireless systems can be strategically designed for performance improvements in various scenarios, such as signal-to-noise ratio (SNR) boost~\cite{ref:ma-modeling-and-perf-analysis,ref:zlp-perf-optm-ma-wideband}, interference mitigation~\cite{ref:ma-null-steering}, flexible beamforming~\cite{ref:mwy-multi-beam-ma,ref:zlp-dynamic-beam-coverage-satellite-ma}, and spatial multiplexing enhancement~\cite{ref:mwy-ma-mimo-capacity-characterization,ref:zlp-multiuser-commun-aided-by-ma}, using the same or even smaller number of antennas compared to conventional FPA systems. 
    % thus outperforming conventional FPA systems with even more antennas~\cite{ref:ma-modeling-and-perf-analysis,ref:mwy-ma-mimo-capacity-characterization}. 
    Note that a conceptually similar technology named fluid antenna system (FAS) has also been previously investigated~\cite{ref:fas-isac}, which shares the same idea with MA in exploiting the flexibility in antenna positioning~\cite{ref:zlp-historical-review}. 
    From the perspective of wireless communication, existing studies on FAS mainly focus on antenna port/positioning adaptation to fading channels regardless of the specific antenna implementation~\cite{ref:fas-perf-mrc, ref:fas-multi-access-capacity-max}. 
    In comparison, MA emphasizes the practical implementation via mechanical movement to adjust its position and/or orientation for improving either small or large scale channel conditions~\cite{ref:ma-opportunities,ref:6dma-modeling-and-optm-statistical}.
    Despite the potentially different implmentations, the channel models tailored for MA systems are also applicable to FAS with flexible antenna positions, and a variety of relevant technologies, such as flexible-position MIMO~\cite{ref:flexible-position-mimo} and flexible antenna array~\cite{ref:flexible-array-model-and-perf-eval}, have been inspired in this context. 

    Given their promising advantages, MAs have attracted increasing attention and extensive studies have been conducted to reap the substantial performance gain via antenna movement optimization~\cite{ref:mwy-ma-mimo-capacity-characterization,ref:wu-yifei-ma-multiuser-discrete-optm,ref:near-field-multiuser-ma,ref:zlp-multiuser-commun-aided-by-ma, ref:ma-multicasting-dl}. 
    For example, the channel capacity of point-to-point MIMO communication systems aided by MAs was characterized in~\cite{ref:mwy-ma-mimo-capacity-characterization}, where alternating optimization (AO) and successive convex approximation (SCA) were employed to optimize transmit and receive MAs' positions as well as the transmit signal covariance matrix. 
    Then, the studies on MAs were extended to multiuser systems, where various algorithms were proposed for MA position optimization to either maximize the sum rate or minimize transmit power given quality-of-service (QoS) constraints, such as gradient ascent~\cite{ref:zlp-multiuser-commun-aided-by-ma}, mixed integer non-linear programming (MINLP)~\cite{ref:wu-yifei-ma-multiuser-discrete-optm}, particle swarm optimization~\cite{ref:near-field-multiuser-ma}, and deep learning~\cite{ref:ma-multicasting-dl}. 
    % Additionally, a graph-based algorithm was proposed in~\cite{ref:mei-graph-based-ma-optm} to optimally maximize the received signal power at a single user by sampling the MA region at the base station (BS) with candidate positions, while deep learning was employed to solve the MA positions at the BS for improving the minimum beamforming gain among multiple users. 
    % Additionally, a graph-based algorithm was proposed in~\cite{ref:mei-graph-based-ma-optm} to optimally maximize the received signal power at a single user, while deep learning was employed to solve the MA positions at the BS for improving the minimum beamforming gain among multiple users. 
    Besides, MAs were also applied to other wireless systems such as intelligent reflecting surface (IRS)-aided communications~\cite{ref:joint-bf-ma-irs-multi-user,ref:ma-secure-irs-isac}, secure communications~\cite{ref:ma-secure-commun}, covert communications~\cite{ref:sum-rate-max-ma-covert-commun}, wireless-powered communication networks~\cite{ref:ma-wpcn-noma}, wireless sensing~\cite{ref:wenyan-ma-sensing,ref:xiaodan-6dma-sensing}, and integrated sensing and communication (ISAC)~\cite{ref:crb-min-ma-isac}. 
    In addition, the six-dimensional movable antenna (6DMA) was recently proposed to exploit the movement capabilities of antennas in both three-dimensional (3D) position and 3D rotation~\cite{ref:6dma-opportunity-challenge,ref:6dma-modeling-and-optm-statistical,ref:6dma-discrete-optm}. 

    On the other hand, to enable antenna movement optimization, significant efforts have also been devoted to channel acquisition/estimation for the MA-aided wireless communication systems~\cite{ref:succ-bayesian-ce-fas,ref:mwy-cs-based-ma-ce,ref:ma-ce-cs-framework,ref:ma-ce-tensor-decomp,ref:ma-ce-wideband}. 
    Different from the channel estimation for conventional FPA systems, MA position optimization usually requires the knowledge of channel mapping between transmit and receive regions where the MAs can be flexibly located~\cite{ref:ma-opportunities}. 
    The existing studies on this issue can be mainly divided into two categories, i.e., the model-based approaches recovering the channel paths in the angular domain~\cite{ref:mwy-cs-based-ma-ce,ref:ma-ce-cs-framework,ref:ma-ce-tensor-decomp,ref:ma-ce-wideband, ref:ce-fas-multiuser-mmw, ref:6dma-ce-directional-sparsity} and the model-free approaches utilizing the channel inherent correlation in the spatial domain~\cite{ref:succ-bayesian-ce-fas, ref:corr-ml-ce-fas, ref:ce-fas-oversample}. 
    % one employing interpolation approaches and the other exploiting channel sparsity in the angular domain. 
    
    % 3rd part: difficulties for MA applied to inst csi & related efforts
    In most of aforementioned works, both MA positioning and channel acquisition are designed based on instantaneous channel state information (CSI) between the transceivers, which may encounter difficulties in practical wireless communication systems. 
    % instantaneous channel was considered for MA position optimization and channel estimation, which, at the same time of fully capturing the channel dynamics, may also lead to several difficulties in practice. 
    First, the channel coherence time may not be sufficient for antenna movement due to the moving speed limitations, especially for mechanically-driven MAs operating in fast fading channels. 
    % As such, even though the optimal MA positions are obtained, it is still very hard to keep up with the channel variations to achieve performance gain. 
    % In comparison, the electrically-driven FAS is more efficient in time and energy, while it is confined to $1$-dimensional arrays and thus 
    Second, considerable energy is required for frequent adjustment of antenna positions adapting to instantaneous CSI, which undermines the efficiency of MA-aided communication systems. 
    % Moreover, despite the various studies, instantaneous CSI acquisition for MA systems still needs extra overhead for antenna movement and channel measurement compared to conventional FPA systems, which further hinders the real-time adjustment of MA positions. 
    To tackle the above challenges, a promising approach is to design MA systems based on statistical CSI, which can effectively reduce the antenna movement overhead because the MAs' positions are reconfigured over large timescales. 
    Driven by this idea, the antenna positions at transceivers are jointly optimized based on statistical CSI to maximize the channel capacity of MA-aided MIMO systems in~\cite{ref:schober-joint-bf-ma-statistic-csi,ref:fas-statistical-csi}, where conventional Rician fading channel model and field-response based channel model were adopted, respectively. 
    This idea was also applied to multiuser systems in~\cite{ref:two-timescale-ma-uplink-ula-pga,ref:two-timescale-ma-qqw} and a two-timescale design for MA-aided wireless communications was proposed based on Rician fading channel. 
    Additionally, line-of-sight (LoS) channels were considered in~\cite{ref:zlp-ma-near-field-statistical,ref:6dma-opportunity-challenge} by assuming specific user distributions while neglecting the non-LoS (NLoS) channel multi-paths. 
    However, both Rician fading and LoS channel models lack generality in characterizing the statistics of wireless channels for practical systems. 
    For example, they cannot represent wireless channels with the LoS path blocked or those with a few dominant channel paths. 
    Thus, there remains a significant knowledge gap in statistical channel modeling and performance optimization for MA-aided multiuser communication systems based on statistical CSI. 

    % 5th part: algorithm details & observations
    To fill this gap, we propose in this paper a field-response based statistical channel model for MA-aided multiuser MIMO (MU-MIMO) systems, based on which the antenna position optimization is investigated. 
    % It is verified in simulations that the downlink ergodic sum rate of all users can be improved significantly compared to conventional FPA systems. 
    The main contributions of this paper are summarized as follows:
    \begin{itemize}
        \item Based on the field-response channel model~\cite{ref:ma-modeling-and-perf-analysis} tailored for MA systems, a statistical channel model is proposed for MA-aided MU-MIMO systems to account for the random channel variations in both LoS and NLoS paths. 
        Specifically, it is assumed that a few dominant scatterers are present near the BS, while rich local scatterers surround each user. 
        The angle of departure (AoD) and angle of arrival (AoA) for each path from the BS to users are determined by the positions of these dominant and local scatterers, respectively. 
        As users move within their vicinities, the path coefficients vary randomly, whereas the AoDs of all paths from the BS remain relatively constant. 
        This statistical channel model provides a foundation for optimizing antenna positions over a large timescale.

        \item With zero-forcing (ZF) beamforming applied at the BS, a Log-barrier-penalized Approximate Gradient Ascent (LAGA) algorithm is proposed to optimize the antenna positions given the knowledge of statistical CSI. 
        Unlike AO and SCA methods commonly employed in the literature, which iteratively solve convex subproblems to address non-convex constraints, the proposed algorithm achieves lower computational complexity in practice by applying gradient ascent with the constraints incorporated into penalty functions. 
        Moreover, two methods are presented to approximate the ergodic sum rate as well as its gradients with respect to (w.r.t.) antenna positions with high/low computational complexities, utilizing Monte-Carlo simulations and asymptotic analysis, respectively. 
        
        \item To validate the proposed scheme, simulations are conducted using ray-tracing generated channels of an urban area in Singapore~\cite{ref:openstreetmap}. 
        The results demonstrate that the proposed algorithm based on statistical CSI between the BS and users achieves performance closely approaching its upper bound, i.e., the ergodic sum rate obtained by optimizing antenna positions based on instantaneous CSI. 
        Moreover, significant improvements in average ergodic sum rate are observed with optimized antenna positions compared to conventional FPA systems under various user distributions. 
        In particular, higher performance gain can be reaped in scenarios with more diverse channel power distributions in the angular domain. 
    \end{itemize}
    
    % 5.5 part: paper organization. 
    The rest of this paper is organized as follows. 
    Section~\ref{sec:system-channel-model} introduces the system and channel models, based on which the two-timescale optimization problem is formulated. 
    Section~\ref{sec:proposed} details the proposed algorithm for antenna position design with ZF beamforming at the BS. 
    Simulation results are presented in Section~\ref{sec:performance-evaluation} and conclusions are drawn in Section~\ref{sec:conclusion}. 
    
    % 6th part: notations
    \textit{Notations:} 
    Boldface letters refer to vectors (lower case) or matrices (upper case). 
    For square matrix $\boldsymbol{A}$, $\text{tr}(\boldsymbol{A})$ denotes its trace and $\boldsymbol{A}^{-1}$ denotes its inverse matrix. 
    For matrix $\boldsymbol{B}$, let $\boldsymbol{B}^{T}$, $\boldsymbol{B}^{H}$, $\text{rank}(\boldsymbol{B})$, $\|\boldsymbol{B}\|_F$, and $[\boldsymbol{B}]_{nm}$ denote the transpose, conjugate transpose, rank, Frobenius norm, and the element in the $n$-th row and $m$-th column of $\boldsymbol{B}$, respectively. 
    $\boldsymbol{I}_N$ denotes the $N\times N$-dimensional identity matrix. 
    For vector $\boldsymbol{x}$, $\|\boldsymbol{x}\|_{2}$ denotes its Euclidean norm. 
    $\boldsymbol{0}_{N\times M}$ denotes the $N\times M$-dimensional zero matrix. 
    Vector $\boldsymbol{1}_{K}$ denotes the $K$-dimensional column vector with all entries equal to one. 
    For vector $\boldsymbol{x}$, $\text{Diag}(\boldsymbol{x})$ denotes the diagonal matrix whose main diagonal elements are extracted from $\boldsymbol{x}$. 
    For matrix $\boldsymbol{A}$, $\text{diag}(\boldsymbol{A})$ denotes the vector whose elements are extracted from the main diagonal elements of $\boldsymbol{A}$. 
    Sets $\mathbb{C}^{a\times b}$, $\mathbb{R}^{a\times b}$, and $\mathbb{R}_{+}^{a\times b}$ denote the space of $a\times b$-dimensional matrices with complex, real, and non-negative real elements, respectively. 
    $\mathbb{E}[\cdot]$ denotes the statistical expectation. 
    Symbol $\mathcal{CN}$ denotes the circular symmetric complex Gaussian (CSCG) distribution. 
    Symbol $\odot$ represents the Hadamard product for matrices. 
    Symbol $j = \sqrt{-1}$. 
    % In addition, $\log{x}$ is used as an abbreviation for $\log_{2}{x}$. 

\vspace{-6pt}
\section{System Model and Statistical Channel Model}\label{sec:system-channel-model}
    
    \subsection{MA-Aided MU-MIMO System}\label{subsec:system-model}
        The MA-aided downlink MU-MIMO system is illustrated in Fig.~\ref{fig:system-model}. 
        The BS is equipped with a two-dimensional (2D) array with ${N}$ MAs while $K$ users located within the site are assumed to each have a single FPA. 
        By establishing a 3D Cartesian coordinate system centered at the BS such that the $x$-$O$-$y$ plane coincides with the antenna plane, as shown in Fig.~\ref{fig:system-model}, the 2D position of the $n$-th MA in the array plane is denoted as $\boldsymbol{r}_{n} = [x_n, y_n]^T\in\mathbb{R}^{2\times 1}$, $\forall n$, where $x_n$ and $y_n$ are its coordinates along the horizontal and vertical directions, respectively. 
        Moreover, define $\boldsymbol{x} = [{x}_{1}, \ldots, {x}_{N}]^T\in\mathbb{R}^{N\times 1}$ and $\boldsymbol{y} = [{y}_{1}, \ldots, {y}_{N}]^T\in\mathbb{R}^{N\times 1}$ as the $x$ and $y$ coordinates of the $N$ antennas, respectively. 
        The antenna moving region is a rectangle area on the $x$-$O$-$y$ plane centered at the origin and its sizes along the $x$ and $y$ axies are denoted as $S_{x}$ and $S_{y}$, respectively. 

        \begin{figure}[t]
            \begin{center}
                \includegraphics[scale = 0.3]{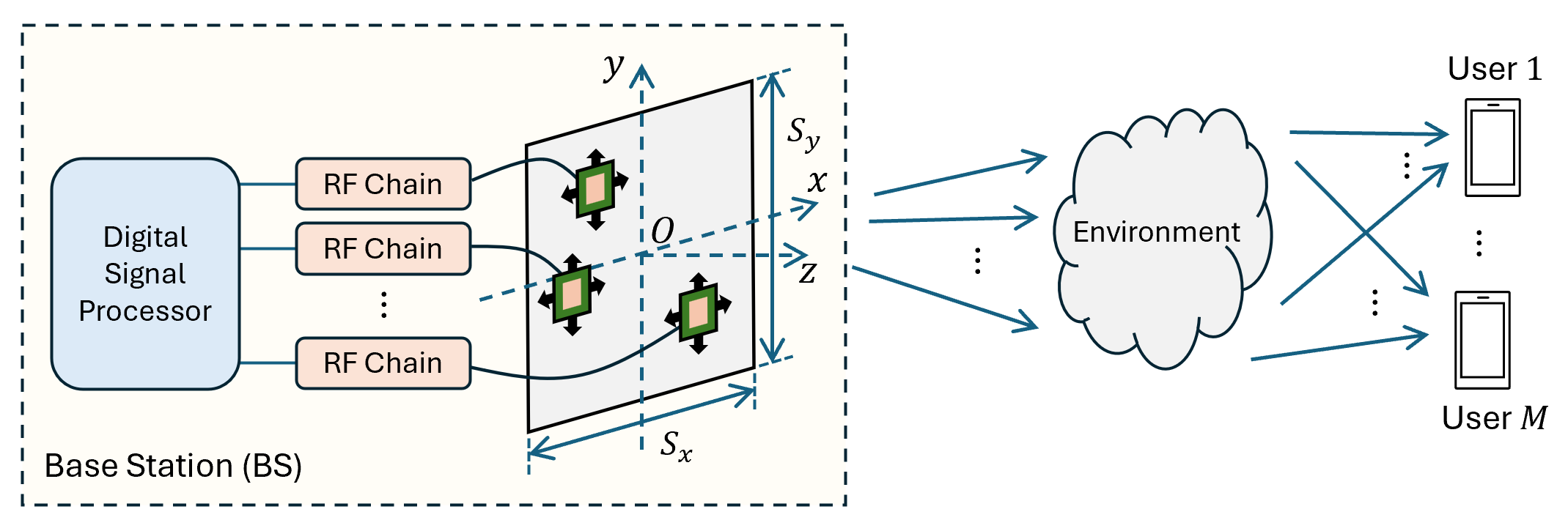}
                \caption{The MA-aided downlink MU-MIMO system. }
                \label{fig:system-model}
            \end{center}
            \vspace{-6pt}
        \end{figure}
        
        Denote $\boldsymbol{W} = [\boldsymbol{w}_{1}, \ldots, \boldsymbol{w}_{K}]\in\mathbb{C}^{{N}\times K}$ as the precoding matrix at the BS and the baseband equivalent channel from the BS to the $k$-th user is denoted as $\boldsymbol{h}_{k}\in\mathbb{C}^{{N}\times{1}}$, $1\le k\le K$. 
        By denoting $\boldsymbol{H} = [\boldsymbol{h}_{1}, \ldots, \boldsymbol{h}_{K}]\in\mathbb{C}^{{{N}}\times{K}}$, the received signal $\boldsymbol{y} = [y_{1}, \ldots, y_{K}]^T\in\mathbb{C}^{K\times 1}$ is given by 
        \begin{equation}\label{def:inst-received-signal}
            \boldsymbol{y} = \boldsymbol{H}^H\boldsymbol{W}\boldsymbol{s} + \boldsymbol{z}, 
        \end{equation}
        where $\boldsymbol{s}\in\mathbb{C}^{K\times 1}$ is the transmitted signal with $\mathbb{E}[\boldsymbol{s}\boldsymbol{s}^H] = \boldsymbol{I}_{K}$, and $\boldsymbol{z}\in\mathbb{C}^{K\times 1}$ is the receiver noise following CSCG distribution $\mathcal{CN}(\boldsymbol{0}_{K\times 1}, \sigma^2\boldsymbol{I}_{K})$, with $\sigma^2$ denoting the average noise power.

    \vspace{-6pt}
    \subsection{Field-Response Based Statistical Channel Model}\label{subsec:statistical-channel-model}
        In this subsection, a field-response-based statistical channel model is introduced. 
        To this end, the instantaneous channel model between the BS and each user is first presented~\cite{ref:ma-modeling-and-perf-analysis,ref:zlp-multiuser-commun-aided-by-ma}. 
        Denote $L_{k}^{t}$ and $L_{k}^{r}$ as the number of transmit and receive channel paths from the BS to the $k$-th user, respectively. 
        The 3D position of the $k$-th user is denoted as $\boldsymbol{u}_{k} = [X_{k}, Y_{k}, Z_{k}]^T\in\mathbb{R}^{3\times 1}$ in its local coordinate system. 
        The azimuth and elevation AoDs for the $l$-th transmit path are denoted as $\theta_{k, l}^{t}$ and $\varphi_{k, l}^{t}$, respectively, $1\le l\le L_{k}^{t}$, while the azimuth and elevation AoAs for the $i$-th receive paths are represented by $\theta_{k, i}^{r}$ and $\varphi_{k, i}^{r}$, respectively, $1\le i\le L_{k}^{r}$. 
        Following the field-response channel model for MA systems~\cite{ref:zlp-multiuser-commun-aided-by-ma}, the transmit and receive field-response vectors (FRVs) for channel between the $k$-th user and the $n$-th antenna are given by
        \begin{subequations}\label{def:field-response-vec-user-k-antenna-n}
            \begin{align}
                \boldsymbol{q}_{k}(\boldsymbol{r}_{n}) & = \bigg[
                    e^{j\rho_{k, 1}^{t}(\boldsymbol{r}_{n})}, \ldots, e^{j\rho_{k, L_{k}^{t}}^{t}(\boldsymbol{r}_{n})}
                \bigg]^{T}\in\mathbb{C}^{L_{k}^{t}\times 1}, \label{def:transmit-FRV} \\
                \boldsymbol{f}_{k}(\boldsymbol{u}_{k}) & = \bigg[
                    e^{j\rho_{k, 1}^{r}(\boldsymbol{u}_{k})}, \ldots, e^{j\rho_{k, L_{k}^{r}}^{r}(\boldsymbol{u}_{k})}
                \bigg]^{T}\in\mathbb{C}^{L_{k}^{r}\times 1}. \label{def:receive-FRV}
            \end{align}
        \end{subequations}
        The phase variations are defined as $\rho_{k, l}^{t}(\boldsymbol{r}_{n}) = \boldsymbol{r}_{n}^{T}\boldsymbol{\kappa}_{k, l}^{t}$, $1\le l\le L_{k}^{t}$, and $\rho_{k, i}^{r}(\boldsymbol{u}_{k}) = \boldsymbol{u}_{k}^{T}\boldsymbol{\kappa}_{k, i}^{r}$, $1\le i\le L_{k}^{r}$, where $\boldsymbol{\kappa}_{k, l}^{t} = \frac{2\pi}{\lambda}[\cos(\theta_{k, l}^{t})\cos(\varphi_{k, l}^{t}), \cos(\theta_{k, l}^{t})\sin(\varphi_{k, l}^{t})]^T\in\mathbb{R}^{2\times 1}$ and $\boldsymbol{\kappa}_{k, i}^{r} = \frac{2\pi}{\lambda}[\cos(\theta_{k, i}^{r})\cos(\varphi_{k, i}^{r}), \cos(\theta_{k, i}^{r})\sin(\varphi_{k, i}^{r}), \sin(\theta_{k, i}^{r})]^T\in\mathbb{R}^{3\times 1}$ are the 2D transmit and 3D receive wavevectors corresponding to the $l$-th transmit channel path and $i$-th received channel path for user $k$, respectively, and $\lambda$ is the carrier wavelength. 
        Then, by defining $\boldsymbol{\Sigma}_{k}\in\mathbb{C}^{{L_{k}^{t}}\times{L_{k}^{r}}}$ as the path-response matrix (PRM)  that represents the response between all the transmit and receive channel paths from the BS to user $k$, the channel vector $\boldsymbol{h}_{k}$ can be expressed as~\cite{ref:zlp-multiuser-commun-aided-by-ma} 
        \begin{equation}\label{def:inst-channel-vec-user-k}
            \boldsymbol{h}_{k} = \boldsymbol{Q}_{k}^{H}\boldsymbol{\Sigma}_{k}\boldsymbol{f}_{k}(\boldsymbol{u}_{k}), 
        \end{equation}
        where $\boldsymbol{Q}_{k} = [\boldsymbol{q}_{k}(\boldsymbol{r}_{1}), \ldots, \boldsymbol{q}_{k}(\boldsymbol{r}_{N})]\in\mathbb{C}^{L_{k}^{t}\times N}$ is the transmit field-response matrix (FRM) for user $k$. 

        Based on the above form of $\boldsymbol{h}_{k}$, a statistical channel model is then developed. 
        Specifically, the BS generally has a high altitude and is mainly surrounded by dominant scatterers of large size, such as tall buildings, which primarily determine the AoDs for the NLoS transmit paths. 
        In contrast, the NLoS receive paths for each user predominantly originate from local scatterers, such as the trees and vehicles, in the vicinity of the user. 
        Each user is assumed to be moving within a local region, the center of which is defined as the reference location for the user, denoted as $\boldsymbol{u}_{k}^{\text{ref}}$. 
        Since the sizes of regions for BS antenna movement and user local movement are generally much smaller than the signal propagation distances between the BS/user and their dominant scatterers, the AoDs and AoAs for the channel paths remain unchanged under the far-field condition. 
        Therefore, the transmit FRM $\boldsymbol{Q}_{k}$ and PRM $\boldsymbol{\Sigma}$ can be considered as approximately constant. 
        In contrast, as each user moves within its local region, the received FRV $\boldsymbol{f}_{k}(\boldsymbol{u}_{k})$ may change rapidly, where phase shifts $\rho_{k, i}^{r}(\boldsymbol{u}_{k})$, $\forall k, i$, can be modeled as indepedent and identically distributed (i.i.d.) uniform random variables within $[0, 2\pi)$. 
        % Moreover, user mobility causes variations in the interaction points on scatterers, introducing random phase shifts in the path-response coefficients, while their amplitudes stay approximately the same. 
        % Specifically, denote $\Sigma_{k, li}$ as the element of the PRM $\boldsymbol{\Sigma}_{k}$ in the $l$-th row and $i$-th column, which is the path-response coefficient corresponding to the $l$-th transmit channel path and the $i$-th receive channel path for user $k$. 
        % Then, it can be rewritten as $\Sigma_{k, li} = \varrho_{k, li}\cdot\exp(j\omega_{k, li})$, where $\varrho_{k, li}$ is the path-response amplitude that stays constant while $\omega_{k, li}$, $\forall k, l, i$, are i.i.d. random phase shifts following uniform distribution within $[0, 2\pi)$. 
        Define $\boldsymbol{\psi}_{k} = \boldsymbol{\Sigma}_{k}\boldsymbol{f}_{k}(\boldsymbol{u}_{k})\in\mathbb{C}^{L_{k}^{t}\times 1}$ as the transmit path-response vector (PRV) for user $k$, where its $l$-th element $\psi_{kl}$ represents the path-response coefficient of the $l$-th transmit channel path for user $k$ and can be written as
        \begin{equation}\label{def:transmit-path-response-coeff}
            \psi_{kl} = \sum_{i = 1}^{L_{k}^{r}}{\Sigma_{k, li}\exp(j\rho_{k, i}^{r}(\boldsymbol{u}_{k}))}, ~\forall k, l, 
        \end{equation}
        where $\Sigma_{k, li}$ as the element of the PRM $\boldsymbol{\Sigma}_{k}$ in the $l$-th row and $i$-th column. 
        % Note that $\Sigma_{k, li}$ is the path-response coefficient corresponding to the $l$-th transmit channel path and the $i$-th receive channel path for user $k$, which 
        Note that $\psi_{kl}$ is expressed as the weighted sum of multiple i.i.d. random variables $\exp(j\rho_{k, i}^{r}(\boldsymbol{u}_{k}))$. 
        Since $L_{k}^{r}$ is generally large in practice due to the rich scattering environment around the user, $\psi_{kl}$ can be approximately modeled as a CSCG random variable according to the Lyapunov Central Limit Theorem~\cite{ref:billingsley-lyapunov-CLT}. 
        Specifically, we have $\psi_{kl}\sim\mathcal{CN}(0, {b}_{kl})$, where ${b}_{kl}$ denotes the expected path-response power of the $l$-th transmit channel path for user $k$. 
        Moreover, the covariance between $\psi_{kl}$ and $\psi_{kl'}$, $\forall l\neq l'$, is given by $\mathbb{E}[\psi_{kl}^{*}\psi_{ki}] = \sum_{i = 1}^{L_{k}^{r}}{\Sigma_{k, li}^{*}\Sigma_{k, l'i}}$. 
        % \begin{equation}\label{def:transmit-coeff-covariance}
        %     \mathbb{E}[\psi_{kl}^{*}\psi_{ki}] = \sum_{i = 1}^{L_{k}^{r}}{\Sigma_{k, li}^{*}\Sigma_{k, l'i}}. 
        % \end{equation}
        Given that $L_{k}^{r}$ is large and the phases of $\Sigma_{k, li}$, $\forall l, i$, are independent and uniformly distributed within $[0, 2\pi)$, we approximate $\mathbb{E}[\psi_{kl}^{*}\psi_{ki}]$ to be $0$. 
        Therefore, the transmit path response vector $\boldsymbol{\psi}_{k} = [{\psi}_{k1}, \ldots, {\psi}_{kL_{k}^{t}}]^T\in\mathbb{C}^{{L_{k}^{t}}\times 1}$ is modeled as $\boldsymbol{\psi}_{k}\sim\mathcal{CN}(\boldsymbol{0}_{{L_{k}^{t}}\times 1}, \text{Diag}(\boldsymbol{b}_{k}))$, where $\boldsymbol{b}_{k} = [{b}_{k1}, \ldots, {b}_{kL_{k}^{t}}]^{T}\in\mathbb{R}_{+}^{{L_{k}^{t}}\times 1}$ is the transmit path-response power vector for user $k$. 
        Note that $\boldsymbol{b}_{k}$ can be regarded as the angular power spectrum for the channel of user $k$ with the BS, which characterizes the average power distribution on the multi-path channel in the angular domain. 
        Hence, the statistical model for channel $\boldsymbol{h}_{k} = \boldsymbol{Q}_{k}^{H}\boldsymbol{\psi}_{k}$ is written as 
        \begin{equation}\label{def:statistical-channel-vec-model}
            \boldsymbol{h}_{k} = \boldsymbol{Q}_{k}^{H}\boldsymbol{\psi}_{k}, ~\boldsymbol{\psi}_{k}\sim\mathcal{CN}\left(
                \boldsymbol{0}_{L_{k}^{t}\times 1}, \text{Diag}(\boldsymbol{b}_{k})
            \right), ~\forall k. 
        \end{equation}
        
        For simplicity of notations, a unified expression is presented for $\boldsymbol{H}$. 
        Specifically, transmit AoDs for all users are put into set $\mathcal{D}$ and are relabeled as $\mathcal{D} = \{\boldsymbol{\kappa}_{1}, \ldots, \boldsymbol{\kappa}_{L}\}$, where where $L = \sum_{k = 1}^{K}{L_{k}^{t}}$ is the total number of transmit channel paths\footnote{Note that the actual number of dominant scatterers surrounding the BS is practically smaller than $L$ as some users may share common scatterers depending on their locations, which causes certain similarities in their statistical CSI, i.e., $\boldsymbol{Q}_{k}$'s. }. 
        Moreover, the extended transmit FRM is given by $\boldsymbol{Q} = [\boldsymbol{Q}_{1}^{H}, \ldots, \boldsymbol{Q}_{K}^{H}]^{H}\in\mathbb{C}^{L\times N}$ and the extended transmit PRM is written as 
        \begin{equation}
            \boldsymbol{\Psi} = \left[
                \begin{array}{cccc}
                    \boldsymbol{\psi}_{1} & \boldsymbol{0}_{{L_{1}^{t}}\times 1} & \ldots & \boldsymbol{0}_{{L_{1}^{t}}\times 1} \\
                    \boldsymbol{0}_{{L_{2}^{t}}\times 1} & \boldsymbol{\psi}_{2} & \ldots & \boldsymbol{0}_{{L_{2}^{t}}\times 1} \\
                    \vdots & \vdots & \ddots & \vdots \\
                    \boldsymbol{0}_{{L_{K}^{t}}\times 1} & \boldsymbol{0}_{{L_{K}^{t}}\times 1} & \ldots & \boldsymbol{\psi}_{K}
                \end{array}
            \right]\in\mathbb{C}^{L\times K}. 
        \end{equation}
        It can be verified that the $n$-th column of matrix $\boldsymbol{Q}$, denoted by $\tilde{\boldsymbol{q}}_{n}$, is given by
        \begin{equation}\label{def:extended-field-response-vec}
            \tilde{\boldsymbol{q}}_{n} = \left[
                \exp(j\boldsymbol{r}_{n}^T\boldsymbol{\kappa}_{1}), \ldots, \exp(j\boldsymbol{r}_{n}^T\boldsymbol{\kappa}_{L})
            \right]^T\in\mathbb{C}^{L\times 1}. 
        \end{equation}
        Therefore, matrix $\boldsymbol{H}$ can be equivalently written as  
        \begin{equation}\label{def:statistical-channel-matrix-model}
            \boldsymbol{H} = \boldsymbol{Q}^H\boldsymbol{\Psi}. 
        \end{equation}

    \vspace{-6pt}
    \subsection{Two-Timescale Design and Problem Formulation}\label{subsec:two-timescale-problem-formulation}
        Based on the signal model in~\eqref{def:inst-received-signal} and the channel model in~\eqref{def:statistical-channel-matrix-model}, the antenna positions and the precoding matrix $\boldsymbol{W}$ at the BS can be jointly designed to improve the system performance. 
        In this paper, we consider maximizing the ergodic sum rate of all users via a two-timescale design approach. 
        Specifically, $\boldsymbol{W}$ is optimized based on the instantaneous channel $\boldsymbol{H}$ to maximize the instantaneous sum rate, while antenna positions are designed over a relatively longer period to improve the ergodic sum rate based on the statistical CSI, which is given by the transmit wavevectors $\mathcal{D}$ and angular power spectrums $\boldsymbol{B} = [\boldsymbol{b}_{1}, \ldots, \boldsymbol{b}_{K}]\in\mathbb{R}_{+}^{L\times K}$ for all users. 
        The antenna positions are denoted by two vectors $\boldsymbol{x} = [{x}_{1}, \ldots, {x}_{N}]^T\in\mathbb{R}^{N\times 1}$ and $\boldsymbol{y} = [{y}_{1}, \ldots, {y}_{N}]^T\in\mathbb{R}^{N\times 1}$, representing the $x$ and $y$ coordinates of the $N$ antennas, respectively. 
        The receive signal-to-interference-and-noise ratio (SINR) of user $k$ is denoted as ${\gamma}_{k}$, which is given by 
        \begin{equation}\label{def:sinr-user-k}
            {\gamma}_{k} = \frac{\big|\boldsymbol{h}_{k}^H\boldsymbol{w}_{k}\big|^2}{\sigma^2 + \sum_{i\neq k}{\big|\boldsymbol{h}_{k}^H\boldsymbol{w}_{i}\big|^2}}, ~\forall k. 
        \end{equation}
        Moreover, the instantaneous sum rate $R$ is defined as $R = \sum_{k = 1}^{K}{\log_{2}(1 + {\gamma}_{k})}$ and the ergodic sum rate $\bar{R}$ is given by the expectation of $R$ w.r.t. the random instantaneous channels, i.e., $\bar{R} = \mathbb{E}_{\boldsymbol{H}}[R]$. 
        Therefore, the two-timescale optimization problem can be formulated as\footnote{The main purpose of the considered problem is to obtain optimal/suboptimal positions of MAs based on statistical CSI. In practice, once the antennas have been moved to the optimized positions, channel estimation and precoding design can be conducted based on instantaneous CSI in a similar way to that in conventional FPA systems. }
        \begin{subequations}\label{prob:two-timescale-optm}
            \allowdisplaybreaks
            \begin{align}
                & \max_{\boldsymbol{x}, \boldsymbol{y}} ~\mathbb{E}_{\boldsymbol{H}}\left[
                    \max_{\boldsymbol{W}} ~\sum_{k = 1}^{K}{\log_{2}(1 + {\gamma}_{k})}
                \right], \tag{\ref{prob:two-timescale-optm}}\\
                & ~\text{s.t.} ~~\text{tr}\left(
                    \boldsymbol{W}^{H}\boldsymbol{W}
                \right)\le {P}_{T}, \label{prob-constraint:tx-power} \\
                & ~~~~~~ |x_n|\le\frac{{S}_{x}}{2}, |y_n|\le\frac{{S}_{y}}{2}, ~\forall n, \label{prob-constraint:ma-region}\\
                & ~~~~~~ \|\boldsymbol{r}_{n} - \boldsymbol{r}_{i}\|_2\ge \Delta, ~\forall n\neq i, \label{prob-constraint:ma-separation}
            \end{align}
        \end{subequations}
        where constraint~\eqref{prob-constraint:tx-power} confines the maximum transmit power ${P}_{T}$, constraint~\eqref{prob-constraint:ma-region} is resulted from the limited antenna moving region, and constraint~\eqref{prob-constraint:ma-separation} specifies the minimum inter-antenna spacing $\Delta$, which is usually set as $\Delta = \lambda/2$. 
        Note that the optimization for $\boldsymbol{W}$ is the conventional transmit precoding problem. 
        However, due to the non-convex constraint~\eqref{prob-constraint:ma-separation} and the expectation over $\boldsymbol{H}$, globally optimal antenna positions are generally difficult to obtain. 
        In the following section, the solution for $\boldsymbol{W}$ is introduced and the LAGA algorithm is proposed to find a suboptimal solution for antenna positions.

\vspace{-6pt}
\section{Proposed Solutions}\label{sec:proposed}
    In this section, the proposed solutions to problem~\eqref{prob:two-timescale-optm} for $\boldsymbol{W}$ and $\boldsymbol{x}, \boldsymbol{y}$ are presented. 
    Specifically, ZF beamforming is employed for $\boldsymbol{W}$, which is asymptotically optimal for maximizing the sum rate of multiple users in the high-SNR region~\cite{ref:emil-optimal-bf}. 
    Based on that, the LAGA algorithm is proposed for antenna position optimization, where the gradient ascent algorithm framework is first presented with the non-convex constraints incorporated into log barrier penalty functions. 
    Then, two methods for approximating the ergodic sum rate as well as its gradient w.r.t. antenna positions are illustrated by leveraging Monte-Carlo simulations and asymptotic analysis, respectively. 
    % The former fully characterizes the statistical distribution of the channel $\boldsymbol{H}$ while the latter grasps the asymptotic forms of the ergodic sum rate without applying Monte-Carlo simulations. 

    \vspace{-6pt}
    \subsection{Zero-Forcing Beamforming}\label{subsec:zero-forcing-bf}
        By applying ZF beamforming at the BS, we have $\boldsymbol{W} = \boldsymbol{H}(\boldsymbol{H}^H\boldsymbol{H})^{-1}\text{Diag}(\boldsymbol{p})^{\frac{1}{2}}$, where $\boldsymbol{C}_{\boldsymbol{H}} = \boldsymbol{H}^H\boldsymbol{H}\in\mathbb{C}^{K\times K}$ is invertible with probability $1$ and $\boldsymbol{p} = [p_1, \ldots, p_{K}]\in\mathbb{R}_{+}^{K\times 1}$ is the power allocation vector. 
        Then, the received signal is written as $\boldsymbol{y} = \text{Diag}(\boldsymbol{p})^{\frac{1}{2}}\boldsymbol{s} + \boldsymbol{z}_{r}$ and thus the instantaneous sum rate is simplified as $R = \sum_{k = 1}^{K}{\log_{2}(1 + \frac{p_k}{\sigma^2})}$. 
        To maximize $R$, $\boldsymbol{p}$ is solved via the following problem:
        \begin{equation}\label{prob:zfbf-power-allocation}
            \max_{\boldsymbol{p}} ~R, ~~\text{s.t.} ~\boldsymbol{p}\ge\boldsymbol{0}_{K\times 1}, ~\text{tr}(\boldsymbol{W}^H\boldsymbol{W})\le P_{T}. 
        \end{equation}
        Note that the second constraint in problem~\eqref{prob:zfbf-power-allocation} can be equivalently written as
        \begin{equation}\label{eq:transmit-power-constraint-simplified}
            \text{tr}(\boldsymbol{W}^H\boldsymbol{W}) = \text{tr}\left(\boldsymbol{C}_{\boldsymbol{H}}^{-1}\text{Diag}(\boldsymbol{p})\right) = \boldsymbol{c}^T\boldsymbol{p}\le P_{T}, 
        \end{equation}
        where $\boldsymbol{c} = \text{diag}(\boldsymbol{C}_{\boldsymbol{H}}^{-1}) \triangleq [c_1, \ldots, c_{K}]^{T}\in\mathbb{R}_{+}^{K\times 1}$ is the diagonal vector of matrix $\boldsymbol{C}_{\boldsymbol{H}}^{-1}$. 
        Substituting~\eqref{eq:transmit-power-constraint-simplified} into the second constraint in problem~\eqref{prob:zfbf-power-allocation}, the vector $\boldsymbol{p}$ can be solved optimally with the water-filling algorithm~\cite{ref:emil-optimal-bf} and is given by ${p}_{k} = (\nu/{c}_{k} - \sigma^2)_{+}, \forall k$, where $(\cdot)_{+} = \max(\cdot, 0)$ and $\nu$ can be solved numerically via bisection search from the following equation: 
        \begin{equation}\label{eq:water-filling-solutions}
            \sum_{k = 1}^{K}{(\nu - \sigma^2 {c}_{k})_{+}} = P_{T}. 
        \end{equation}

    \vspace{-6pt}
    \subsection{LAGA Algorithm for Antenna Position Optimization}\label{subsec:log-barrier-gradient-ascent}
        Given ZF beamforming and the optimal power allocation vector $\boldsymbol{p}$, the resulting instantaneous sum rate is denoted as $R_{\text{ZF}}$. %, which is the optimal value for problem~\eqref{prob:zfbf-power-allocation}. 
        Then, the ergodic sum rate is given by $\bar{R}_{\text{ZF}} = \mathbb{E}_{\boldsymbol{H}}[R_{\text{ZF}}]$ and the optimization problem for antenna position design can be simplified as 
        \begin{equation}\label{prob:antenna-position-optm}
            \max_{\boldsymbol{x}, \boldsymbol{y}} ~\bar{R}_{\text{ZF}}, ~\eqref{prob-constraint:ma-region}, \eqref{prob-constraint:ma-separation}. 
        \end{equation}
        % whose optimal solution is denoted as $(\boldsymbol{x}^{\star}, \boldsymbol{y}^{\star})$. 
        Next, the proposed LAGA algorithm for antenna position optimization is presented to find a suboptimal solution $(\hat{\boldsymbol{x}}, \hat{\boldsymbol{y}})$ for problem~\eqref{prob:antenna-position-optm}. 
        
        To address the intractable expectation over $\boldsymbol{H}$, the ergodic sum rate $\bar{R}_{\text{ZF}}$ is approximated by a surrogate function $\bar{R}_{\text{approx}}$ that allows for efficient computation of gradients w.r.t. antenna positions, which will be specified later. 
        Then, constraints~\eqref{prob-constraint:ma-region} and~\eqref{prob-constraint:ma-separation} are incorporated into log-barrier penalty functions and gradient ascent can be applied. 
        Specifically, define $\mathcal{S}_{\text{MA}}$ as the feasible set for the duplet $(\boldsymbol{x}, \boldsymbol{y})$ that satisfies constraints~\eqref{prob-constraint:ma-region} and~\eqref{prob-constraint:ma-separation}. 
        Then, the log-barrier function $\mathcal{L}(\boldsymbol{x}, \boldsymbol{y})$ is defined based on set $\mathcal{S}_{\text{MA}}$ as follows:
        \begin{equation}\label{def:log-barrier-func}
            \mathcal{L}(\boldsymbol{x}, \boldsymbol{y}) = \left\{
                \begin{array}{ll}
                    \mathcal{L}_{f}(\boldsymbol{x}, \boldsymbol{y}), & (\boldsymbol{x}, \boldsymbol{y})\in\text{int}(\mathcal{S}_{\text{MA}}), \\
                    -\infty, & \text{otherwise}, 
                \end{array}
            \right.
        \end{equation}
        where $\text{int}(\mathcal{S}_{\text{MA}})$ denotes the interior of set $\mathcal{S}_{\text{MA}}$ and $\mathcal{L}_{f}(\boldsymbol{x}, \boldsymbol{y})$ is given by 
        \begin{equation}\label{def:log-barrier-feasible-func}
            \begin{aligned}
                & \mathcal{L}_{f}(\boldsymbol{x}, \boldsymbol{y}) \triangleq \sum_{
                    1 \le n < i \le {N}
                }{\ln{\left(
                    \|\boldsymbol{r}_{n} - \boldsymbol{r}_{i}\|_2^2 - \Delta^2
                \right)}} \\
                & ~~~~~~ + \left[
                    \sum_{n = 1}^{{N}}{\ln{\left(
                        \frac{{S}_{x}^2}{4} - x_n^2
                    \right)} + \ln{\left(
                        \frac{{S}_{y}^2}{4} - y_n^2
                    \right)}}
                \right]. 
            \end{aligned}
        \end{equation}
        Using $\mathcal{L}(\boldsymbol{x}, \boldsymbol{y})$ as a penalty function to replace constraints~\eqref{prob-constraint:ma-region} and~\eqref{prob-constraint:ma-separation}, problem~\eqref{prob:antenna-position-optm} is approximated as 
        \begin{equation}\label{prob:ma-optm-log-barrier-relaxed}
            \max_{\boldsymbol{x}, \boldsymbol{y}} ~f(\boldsymbol{x}, \boldsymbol{y}) = \bar{R}_{\text{approx}} + \mu\mathcal{L}(\boldsymbol{x}, \boldsymbol{y}), 
        \end{equation}
        where $\mu > 0$ is the penalty parameter. 
        It is easy to verify that if the antenna positions are infeasible to constraints~\eqref{prob-constraint:ma-region} and~\eqref{prob-constraint:ma-separation}, we have $f(\boldsymbol{x}, \boldsymbol{y}) = \mathcal{L}(\boldsymbol{x}, \boldsymbol{y}) = -\infty$. 
        Thus, $(\boldsymbol{x}, \boldsymbol{y})$ is not the optimal solution for problem~\eqref{prob:ma-optm-log-barrier-relaxed}, which indicates that the optimal solution for problem~\eqref{prob:ma-optm-log-barrier-relaxed} is always feasible to constraints~\eqref{prob-constraint:ma-region} and~\eqref{prob-constraint:ma-separation}. 
        Moreover, as $\mu\to 0^{+}$, we have $\mathcal{L}_{f}(\boldsymbol{x}, \boldsymbol{y})\to 0$ for any $(\boldsymbol{x}, \boldsymbol{y})\in\text{int}(\mathcal{S}_{\text{MA}})$. 
        Therefore, given a sufficiently small value of $\mu$, the penalized term $\mu\mathcal{L}(\boldsymbol{x}, \boldsymbol{y})$ becomes negligible for any $(\boldsymbol{x}, \boldsymbol{y})\in\text{int}(\mathcal{S}_{\text{MA}})$ and the optimal solution for problem~\eqref{prob:ma-optm-log-barrier-relaxed} approximately maximizes $\bar{R}_{\text{approx}}$ (or $\bar{R}_{\text{ZF}}$), which solves the original problem~\eqref{prob:antenna-position-optm}\footnote{If the optimal solution $(\boldsymbol{x}^{\star}, \boldsymbol{y}^{\star})$ for problem~\eqref{prob:antenna-position-optm} is right on the boundary of set $\mathcal{S}_{\text{MA}}$, it cannot be precisely solved by the proposed algorithms since $f(\boldsymbol{x}^{\star}, \boldsymbol{y}^{\star}) = -\infty$. 
        However, it can be approached by the interior points of $\mathcal{S}_{\text{MA}}$ with an arbitrarily high accuracy by adjusting $\mu$. }. 

        Nevertheless, directly solving problem~\eqref{prob:ma-optm-log-barrier-relaxed} with a fixed and small value of $\mu$ may not always lead to a good solution because it sacrifices the global search capability. 
        Therefore, we start from a relatively large $\mu$ and iteratively shrink its value by factor $\rho$. 
        Given any value of $\mu$, gradient ascent is applied to find a suboptimal solution for problem~\eqref{prob:ma-optm-log-barrier-relaxed}. 
        As such, the antenna positions solved from the previous iteration for $\mu$ serve as a good initialization for the next iteration and gradually approach a locally optimal solution for problem~\eqref{prob:antenna-position-optm}. 
        The algorithm terminates when the displacement of antenna positions become negligible compared with the previous iteration for $\mu$, which indicates that a suitable value for $\mu$ has been reached and the solved antenna positions can be assumed to be close enough to a suboptimal solution to problem~\eqref{prob:antenna-position-optm}. 
        
        {
        % \vspace{-10pt}
        \begin{algorithm}[t]
            \begin{minipage}{0.95\linewidth}
            \centering
            \caption{Proposed LAGA Algorithm. }
            \begin{algorithmic}[1]\label{alg:log-barrier-gradient-ascent-framework}
                \REQUIRE Statistical CSI $\mathcal{D}$ and $\boldsymbol{B}$; threshold $\varepsilon_{r}$, initial penalty parameter $\mu_{0}$, maximum step size ${\alpha}_{0}$, and maximum iteration number $I$ for gradient ascent. 
                \STATE Sparse uniform planar array (UPA) is applied as initialization for $(\boldsymbol{x}^{(0)}, \boldsymbol{y}^{(0)})$, where antennas are formed into rectangular UPA (with antenna inter-spacing specified in Section~\ref{subsec:perf-benchmarks} for UPA-sparse); $\mu\gets\mu_{0}$, $\varepsilon\gets\varepsilon_{r}$. 
                \WHILE{$\varepsilon \ge \varepsilon_{r}$}
                    \STATE Let $i\gets 0$. 
                    \WHILE{$i < I$}
                        \STATE Given $(\boldsymbol{x}^{(i)}, \boldsymbol{y}^{(i)})$, compute gradients $\nabla_{v}{f} = \nabla_{v}{\bar{R}_{\text{approx}}} + \mu\nabla_{v}{\mathcal{L}_{f}}$ for $v\in\{x, y\}$. 
                        \STATE Compute $\boldsymbol{d}^{(i)}\gets[\nabla_{x}{f}^{T}, \nabla_{y}{f}^{T}]^{T}$ and normalized gradients $\boldsymbol{g}_{v}^{(i)}\gets\nabla_{v}{f}/\|\boldsymbol{d}^{(i)}\|_2$, $v\in\{x, y\}$. 
                        \STATE Initial step size ${\alpha}\gets{\alpha}_{0}$. 
                        % \STATE Solve step size via Algorithm~\ref{alg:backtracking}. 
                        \WHILE{conditions~\eqref{def:backtracking-conditions} are not satisfied}
                            \STATE Shrink the step size as $\alpha\gets\alpha/2$. 
                        \ENDWHILE
                        \STATE Let $\boldsymbol{x}^{(i + 1)}\gets\boldsymbol{x}^{(i)} + \alpha\boldsymbol{g}_{x}^{(i)}$, $\boldsymbol{y}^{(i + 1)}\gets\boldsymbol{y}^{(i)} + \alpha\boldsymbol{g}_{y}^{(i)}$, and $i\gets i + 1$. 
                    \ENDWHILE
                    \STATE Compute $\varepsilon\gets(\|\boldsymbol{x}^{(i)} - \boldsymbol{x}^{(0)}\|_{2}^{2} + \|\boldsymbol{y}^{(i)} - \boldsymbol{y}^{(0)}\|_{2}^{2})^{1/2}$. 
                    \STATE Let $\boldsymbol{x}^{(0)}\gets\boldsymbol{x}^{(i)}$, $\boldsymbol{y}^{(0)}\gets\boldsymbol{y}^{(i)}$, and $\mu\gets\rho\mu$. 
                \ENDWHILE
                \RETURN The optimized antenna positions $\hat{\boldsymbol{x}}\gets\boldsymbol{x}^{(0)}$ and $\hat{\boldsymbol{y}}\gets\boldsymbol{y}^{(0)}$. 
            \end{algorithmic}
            \end{minipage}
        \end{algorithm}
        
        }
        
        The proposed LAGA algorithm is summarized in Algorithm~\ref{alg:log-barrier-gradient-ascent-framework}, where $\nabla_{x}{\bar{R}_{\text{approx}}}$, $\nabla_{y}{\bar{R}_{\text{approx}}}$ and $\nabla_{x}{\mathcal{L}_{f}}$, $\nabla_{y}{\mathcal{L}_{f}}$ denote the gradients of $\bar{R}_{\text{approx}}$ and $\mathcal{L}_{f}$ w.r.t. $\boldsymbol{x}$ and $\boldsymbol{y}$, respectively. 
        Correspondingly, the gradients of $f$ w.r.t. antenna positions are given by $\nabla_{v}{f} = \nabla_{v}{\bar{R}_{\text{approx}}} + \mu\nabla_{v}{\mathcal{L}_{f}}$, $v\in\{x, y\}$. 
        Instead of directly employing $\nabla_{x}{f}$ and $\nabla_{y}{f}$, the gradients used for antenna positions' updates are normalized such that the displacement of antennas is determined by the step size $\alpha$. 
        Specifically, the antenna positions in the $i$-th iteration of gradient ascent are updated as $\boldsymbol{x}^{(i + 1)} = \boldsymbol{x}^{(i)} + \alpha\boldsymbol{g}_{x}^{(i)}$ and $\boldsymbol{x}^{(i + 1)} = \boldsymbol{y}^{(i)} + \alpha\boldsymbol{g}_{y}^{(i)}$, where $\boldsymbol{g}_{x}^{(i)}$ and $\boldsymbol{g}_{y}^{(i)}$ are normalized gradients defined as 
        \begin{subequations}\label{def:gradient-normalization}
            \begin{gather}
                \boldsymbol{g}_{v}^{(i)} = {\nabla_{v}{f}}/{\|\boldsymbol{d}^{(i)}\|_2}\in\mathbb{R}^{N\times 1}, ~v\in\{x, y\}, \label{subdef:normalized-gradients} \\
                \boldsymbol{d}^{(i)} = [\nabla_{x}{f}^{T}, \nabla_{y}{f}^{T}]^{T}\in\mathbb{R}^{2N\times 1}. \label{subdef:total-gradient}
            \end{gather}
        \end{subequations}
        Besides, the step size $\alpha$ is obtained by performing the backtracking line search~\cite{ref:Armijo-gradient-backtracking} to ensure that $(\boldsymbol{x}, \boldsymbol{y})$ is always feasible during the iterations and the objective function $f$ increases given $\mu$. 
        Starting from the initial value $\alpha_{0}$, $\alpha$ keeps shrinking by half until it satisfies the following conditions:
        \begin{subequations}\label{def:backtracking-conditions}
            \begin{gather}
                (\boldsymbol{x}^{(i)} + \alpha\boldsymbol{g}_{x}^{(i)}, \boldsymbol{y}^{(i)} + \alpha\boldsymbol{g}_{y}^{(i)})\in\text{int}(\mathcal{S}_{\text{MA}}), \label{subdef:backtracking-feasible} \\
                f(\boldsymbol{x}^{(i)} + \alpha\boldsymbol{g}_{x}^{(i)}, \boldsymbol{y}^{(i)} + \alpha\boldsymbol{g}_{y}^{(i)}) \ge f(\boldsymbol{x}^{(i)}, \boldsymbol{y}^{(i)}) + \eta\alpha\|\boldsymbol{d}^{(i)}\|_2, \label{subdef:backtracking-increasing}
            \end{gather}
        \end{subequations}
        where $\eta\in(0, 1)$ is the control parameter. 
        Moreover, for $v\in\{x, y\}$, it can be verified that the gradient $\nabla_{v}{\mathcal{L}_{f}}$ is written as $\nabla_{v}{\mathcal{L}_{f}} = [\frac{\mathrm{d}\mathcal{L}_{f}}{\mathrm{d}{v}_{1}}, \ldots, \frac{\mathrm{d}\mathcal{L}_{f}}{\mathrm{d}{v}_{N}}]^T\in\mathbb{R}^{{N}\times 1}$, which is given by
        \begin{equation}\label{def:log-barrier-feasible-gradients}
            \frac{\mathrm{d}\mathcal{L}_{f}}{\mathrm{d}{v}_{n}} = \frac{-2{v}_{n}}{
                {S}_{v}^2/4 - {v}_{n}^2
            } + \sum_{\substack{
                1\le i\le {N} \\
                i\neq n
            }}{\frac{
                2({v}_{n} - {v}_{i})
            }{
                \|\boldsymbol{r}_{n} - \boldsymbol{r}_{i}\|_2^2 - \Delta^2
            }}, ~\forall n. 
        \end{equation}
        However, due to the implicit definition of vector $\boldsymbol{p}$ and the expectation over $\boldsymbol{H}$, it is difficult to obtain the ergodic sum rate $\bar{R}_{\text{ZF}}$ and its gradient in closed form. 
        This motivates the consideration of a surrogate function $\bar{R}_{\text{approx}}$ that enables efficient computations for its gradient. 
        In the following sections, two expressions are derived for ${\bar{R}_{\text{approx}}}$ as well as their gradients with high/low computational complexities by exploiting Monte-Carlo simulation method and asymptotic analysis, respectively.

    \vspace{-6pt}
    \subsection{Monte-Carlo Approximation}\label{subsec:monte-carlo-approx}
        Since the ergodic sum rate $\bar{R}_{\text{ZF}}$ is obtained via expectation over $\boldsymbol{H}$, it can be approximated as the average of instantaneous sum rates given a sufficiently large number of independent realizations of $\boldsymbol{H}$ based on its spatial distribution (to be specified in Section~\ref{subsec:simulation-setup}). 
        Specifically, $M$ realizations of $\boldsymbol{H}$ are independently generated as $\boldsymbol{H}_{1}, \ldots, \boldsymbol{H}_{M}$, based on which the instantaneous sum rates ${R}_{1}, \ldots, {R}_{M}$ are computed by applying ZF beamforming and optimal power allocation following Section~\ref{subsec:zero-forcing-bf}, respectively. 
        Thus, $\bar{R}_{\text{approx}}$ is given by the empirical expectation of ${R}_{1}, \ldots, {R}_{M}$, i.e., 
        \begin{equation}\label{def:ergodic-sum-rate-monte-carlo-approx}
            \bar{R}_{\text{approx}} = \bar{R}_{\text{MC}} \triangleq \frac{1}{M}\sum_{m = 1}^{M}{{R}_{m}}. 
        \end{equation}
        With a sufficiently large $M$, the approximate ergodic sum rate $\bar{R}_{\text{approx}}$ will closely approach $\bar{R}_{\text{ZF}}$. 
        
        Given $\bar{R}_{\text{approx}}$, its gradient can be written as 
        \begin{equation}\label{def:ergodic-gradient-monte-carlo-approx}
            \nabla_{v}\bar{R}_{\text{approx}} = \nabla_{v}\bar{R}_{\text{MC}} \triangleq \frac{1}{M}\sum_{m = 1}^{M}{\frac{\mathrm{d}{{R}_{m}}}{\mathrm{d}{\boldsymbol{v}}}}, ~v\in\{x, y\}, 
        \end{equation}
        where $\frac{\mathrm{d}{{R}_{m}}}{\mathrm{d}{\boldsymbol{v}}} = [\frac{\mathrm{d}{{R}_{m}}}{\mathrm{d}{{v}_{1}}}, \ldots, \frac{\mathrm{d}{{R}_{m}}}{\mathrm{d}{{v}_{N}}}]^{T}\in\mathbb{R}^{N\times 1}$ is the gradient of ${R}_{m}$ w.r.t. $v\in\{x, y\}$ and can be derived via the chain rule. 
        Particularly, it can be observed from problem~\eqref{prob:zfbf-power-allocation} that $R_{\text{ZF}}$ is determined only by the vector $\boldsymbol{c}$, i.e., we have $R_{\text{ZF}} = \mathcal{R}(\boldsymbol{c})$, where $\mathcal{R}(\cdot)$ is a function that maps the vector $\boldsymbol{c}$ to the instantaneous sum rate after applying ZF beamforming and optimal power allocation. 
        Meanwhile, $\boldsymbol{c}$ is determined by the antenna positions $\boldsymbol{x}$ and $\boldsymbol{y}$. 
        Therefore, we have 
        \begin{equation}\label{eq:inst-sum-rate-gradient-chain-rule}
            \frac{\mathrm{d}R_{\text{ZF}}}{\mathrm{d}{v}_{n}} = \sum_{k = 1}^{K}\frac{\partial R_{\text{ZF}}}{\partial c_{k}}\cdot\frac{\mathrm{d}c_{k}}{\mathrm{d}{v}_{n}}, ~v\in\{x, y\}, ~\forall n. 
        \end{equation}
        Despite the implicit relation between $\boldsymbol{c}$ and $R_{\text{ZF}}$, the derivative of $R_{\text{ZF}}$ w.r.t. $c_{k}$ can be written in closed-form as
        \begin{equation}\label{def:inst-sum-rate-derivative-to-uvec}
            \frac{\partial R_{\text{ZF}}}{\partial c_{k}} = \frac{\partial \mathcal{R}}{\partial c_{k}} = \frac{-p_{k}}{\nu\ln{2}}, ~\forall k, 
        \end{equation}
        with the detailed derivation shown in Appendix~\ref{appendix-subsec:inst-sum-rate-derivative-to-uvec}. 
        Moreover, by defining $\boldsymbol{\kappa}_{l} = [{\kappa}_{l}^{x}, {\kappa}_{l}^{y}]^T$, the derivative of $c_{k}$ w.r.t. ${v}_{n}$ for $v\in\{x, y\}$ is given by (as detailed in Appendix~\ref{appendix-subsec:uvec-derivative-to-position})
        \begin{equation}\label{def:uvec-derivative-to-position}
            \frac{\mathrm{d}c_{k}}{\mathrm{d}{v}_{n}} = \tilde{\boldsymbol{q}}_{n}^H\text{Diag}(\boldsymbol{\xi}_{k})\boldsymbol{\Lambda}^{v}\text{Diag}(\boldsymbol{\xi}_{k}^H)\tilde{\boldsymbol{q}}_{n}, ~\forall k, n, 
        \end{equation}
        where $\boldsymbol{\xi}_{k}\in\mathbb{C}^{L\times 1}$ is the $k$-th column of the matrix $\boldsymbol{\Psi}\boldsymbol{C}_{\boldsymbol{H}}^{-1}\in\mathbb{C}^{L\times K}$, and matrices $\boldsymbol{\Lambda}^{v}$, $v\in\{x, y\}$, are $L\times L$ Hermitian matrices such that their elements in the $i$-th column and $l$-th row are defined as $\Lambda_{il}^{v} = j({\kappa}_{i}^{v} - {\kappa}_{l}^{v})$, $\forall i, l$, respectively. 
        Substituting~\eqref{def:inst-sum-rate-derivative-to-uvec} and~\eqref{def:uvec-derivative-to-position} into~\eqref{eq:inst-sum-rate-gradient-chain-rule} and following the relation $\text{Diag}(\boldsymbol{\xi}_{k})\boldsymbol{\Lambda}^{v}\text{Diag}(\boldsymbol{\xi}_{k}^H) = (\boldsymbol{\xi}_{k}\boldsymbol{\xi}_{k}^H)\odot\boldsymbol{\Lambda}^{v}$, the gradient $\nabla{R_{\text{ZF}}}$ can be written as follows:
        \begin{subequations}\label{def:inst-sum-rate-gradients}
            \allowdisplaybreaks
            \begin{align}
                \frac{\mathrm{d}R_{\text{ZF}}}{\mathrm{d}{v}_{n}} % & = \frac{-1}{\nu\ln{2}}\sum_{k = 1}^{K}{p_{k}\cdot\tilde{\boldsymbol{q}}_{n}^H\text{Diag}(\boldsymbol{\xi}_{k})\boldsymbol{\Lambda}^{v}\text{Diag}(\boldsymbol{\xi}_{k}^H)\tilde{\boldsymbol{q}}_{n}} \\
                & = \frac{-1}{\nu\ln{2}}\tilde{\boldsymbol{q}}_{n}^H\left[
                    \sum_{k = 1}^{K}{\left(
                        p_{k}\cdot\boldsymbol{\xi}_{k}\boldsymbol{\xi}_{k}^H
                    \right)}\odot\boldsymbol{\Lambda}^{v}
                \right]\tilde{\boldsymbol{q}}_{n} \\
                & = \frac{-1}{\nu\ln{2}}\tilde{\boldsymbol{q}}_{n}^{H}\left(
                    \boldsymbol{F}\odot\boldsymbol{\Lambda}^{v}
                \right)\tilde{\boldsymbol{q}}_{n}, ~v\in\{x, y\}, \label{def:inst-sum-rate-gradients-xygrad}
            \end{align}
        \end{subequations}
        where matrix $\boldsymbol{F}$ is given by 
        \begin{equation}\label{def:inst-sum-rate-gradients-Fmat}
            \boldsymbol{F} = \sum_{k = 1}^{K}{p_{k}\cdot\boldsymbol{\xi}_{k}\boldsymbol{\xi}_{k}^H} = \boldsymbol{\Psi}\boldsymbol{C}_{\boldsymbol{H}}^{-1}\text{Diag}(\boldsymbol{p})\boldsymbol{C}_{\boldsymbol{H}}^{-1}\boldsymbol{\Psi}^H. 
        \end{equation}
        Therefore, by calculating parameter $\nu$ and matrix $\boldsymbol{F}$ for the $m$-th channel realization as $\nu_{m}$ and $\boldsymbol{F}_{m}$, respectively, the gradient ${\mathrm{d}{{R}_{m}}}/{\mathrm{d}{\boldsymbol{v}}}$ can be written as 
        \begin{equation}
            \frac{\mathrm{d}{{R}_{m}}}{\mathrm{d}{\boldsymbol{v}}} = \frac{-1}{\nu_{m}\ln{2}}\text{diag}\left(
                \boldsymbol{Q}^{H}\left(
                    \boldsymbol{F}_{m}\odot\boldsymbol{\Lambda}^{v}
                \right)\boldsymbol{Q}
            \right), ~\forall m, 
        \end{equation}
        based on which $\nabla_{v}{\bar{R}_{\text{approx}}}$ can be obtained via~\eqref{def:ergodic-gradient-monte-carlo-approx}. 
        It can be verified that the computational complexity of computing $\nabla_{v}{\bar{R}_{\text{MC}}}$ is given by $\mathcal{O}((K^{3} + NL^{2})M)$.

    \vspace{-6pt}
    \subsection{Asymptotic Approximation}\label{subsec:asymptotic-approximation}
        Despite that the Monte-Carlo approximated ergodic sum rate can be arbitrarily accurate, its relies on a sufficiently large number of random channel realizations, which entails a high computational complexity. 
        To achieve lower complexity, a deterministic approximation of $\bar{R}_{\text{ZF}}$ is derived in this subsection without the need for massive channel realizations. 
        Although $\boldsymbol{c}$ and $\bar{R}_{\text{ZF}}$ are difficult to characterize given the number of channel paths, they become more tractable as ${L}_{k}^{t}\to\infty$, $\forall k$. 
        According to~\cite{ref:couillet-random-matrix-for-commun}, an asymptotic form for the random vector $\boldsymbol{c}$, denoted as $\boldsymbol{c}^{\infty}$, can be obtained with infinitely many paths, which is known as the deterministic equivalent (DE) for $\boldsymbol{c}$. 
        
        % Specifically, define ${b}_{kl} = \tilde{a}_{kl}^2$, $\forall k, l$ as the power response of the $l$-th transmit channel path and $\boldsymbol{b}_{k} = [b_{k1}, \ldots, b_{kL}]^T\in\mathbb{R}_{+}^{L\times 1}$ as the power spectrum vector for user $k$ in the angular domain, $\forall k$. 
        Given $\boldsymbol{b}_{k}$ as the power spectrum vector for user $k$ in the angular domain, $\forall k$, the channel autocorrelation matrix for user $k$ is expressed as $\boldsymbol{G}_{k} = \mathbb{E}[\boldsymbol{h}_{k}\boldsymbol{h}_{k}^H] = \boldsymbol{Q}^H\text{Diag}(\boldsymbol{b}_{k})\boldsymbol{Q}$. 
        By applying the Bai and Silverstein method in~\cite{ref:couillet-random-matrix-for-commun}, $\boldsymbol{c}^{\infty} = [{c}_{1}^{\infty}, \ldots, {c}_{K}^{\infty}]^T\in\mathbb{R}_{+}^{K\times 1}$ can be solved implicitly and ${c}_{k}^{\infty}$ is defined as follows:
        \begin{subequations}\label{def:asymp-cvec}
            \begin{align}
                {c}_{k}^{\infty} & = \left[
                    \text{tr}(\boldsymbol{G}_{k}\boldsymbol{Y}_{k}^{-1})
                \right]^{-1}, \\
                \boldsymbol{Y}_{k} & = \boldsymbol{I}_{N} + \sum_{i\neq k}{\epsilon}_{k, i}\boldsymbol{G}_{i}, 
            \end{align}
        \end{subequations}
        where ${\epsilon}_{k, l}\ge 0$, $\forall l$, are the unique positive solutions to the following equations (with the uniqueness proved in Appendix~\ref{appendix:asymp-ergo-sum-rate}):
        \begin{equation}\label{def:auxiliary-variables-fp-eq}
            \text{tr}\left(
                {\epsilon}_{k, l}\boldsymbol{G}_{l}\boldsymbol{Y}_{k}^{-1}
            \right) = 1, ~l = 1, \ldots, K. 
        \end{equation}
        It can be observed that ${\epsilon}_{k, l}$, $\forall l, k$, and $\boldsymbol{c}^{\infty}$ are deterministic given statistical CSI $\mathcal{D}$ and $\boldsymbol{B}$, regardless of instantaneous channel realizations. 
        Thus, vector $\boldsymbol{c}^{\infty}$ can be considered as a deterministic approximation of $\boldsymbol{c}$ that is asymptotically accurate, as demonstrated in the following proposition. 

        \begin{proposition}\label{prop:uvec-deterministic-equivalent}
            % Define total channel power for user $k$ as $\varrho_{k} = \sum_{l = 1}^{L_{k}^{t}}{{b}_{kl}}$. 
            As $L_{k}^{t}\to\infty$ with $\zeta_{k} = L_{k}^{t}/L_{1}^{t}$, $\forall k$, being fixed, we have ${c}_{k} - {c}_{k}^{\infty}\overset{\mathrm{a.s.}}{\longrightarrow}{0}$, $k = 1, \ldots, K$. 
            % \footnote{This requires ${N}$, $K$, and $L_{k}^{t}$, $\forall k$, to grow at the same rate. Otherwise, Proposition~\ref{prop:uvec-deterministic-equivalent} may not hold true. }
        \end{proposition}
        \begin{proof}
            Please refer to Appendix~\ref{appendix:asymp-ergo-sum-rate}. 
        \end{proof}

        According to the proposition above, for sufficiently large $L_{k}^{t}$, $\forall k$, every realization of $\boldsymbol{c}$ is very close to $\boldsymbol{c}^{\infty}$, which leads to approximately the same instantaneous sum rate $\mathcal{R}(\boldsymbol{c})\approx\mathcal{R}(\boldsymbol{c}^{\infty})$. 
        Therefore, we propose to approximate the ergodic sum rate as 
        \begin{equation}
            \bar{R}_{\text{approx}} = R_{\text{ZF}}^{\infty} \triangleq \mathcal{R}(\boldsymbol{c}^{\infty}). 
        \end{equation}
        Despite the precondition on large values of $L_{k}^{t}$, $\forall k$, in Proposition~\ref{prop:uvec-deterministic-equivalent}, we will show via simulations in Section~\ref{sec:performance-evaluation} that ${R}_{\text{ZF}}^{\infty}$ can also provide a good approximation of the ergodic sum rate for relatively small $L_{k}^{t}$, $\forall k$. 

        To obtain ${c}_{k}^{\infty}$, we need to solve $\boldsymbol{\epsilon}_{k} = [{\epsilon}_{k, 1}, \ldots, {\epsilon}_{k, K}]^T\in\mathbb{R}_{+}^{K\times 1}$ first from equations~\eqref{def:auxiliary-variables-fp-eq}. 
        To this end, define functions for vector $\boldsymbol{\epsilon} = [\epsilon_{1}, \ldots, \epsilon_{K}]^T\in\mathbb{R}_{+}^{K\times 1}$ as follows: 
        \begin{subequations}\label{def:auxiliary-variables-newton-eq}
            \allowdisplaybreaks
            \begin{gather}
                {g}_{k, l}(\boldsymbol{\epsilon}) = \text{tr}\left(
                    \epsilon_{l}\boldsymbol{G}_{l}\left(
                        \boldsymbol{I}_{N} + \textstyle\sum_{i\neq k}\epsilon_{i}\boldsymbol{G}_{i}
                    \right)^{-1}
                \right) - 1, ~\forall l, \\
                \boldsymbol{\mathcal{G}}_{k}(\boldsymbol{\epsilon}) = [{g}_{k, 1}(\boldsymbol{\epsilon}), \ldots, {g}_{k, K}(\boldsymbol{\epsilon})]^T\in\mathbb{R}^{K\times 1}. 
            \end{gather}
        \end{subequations}
        Note that $\boldsymbol{\epsilon}_{k}$ is the unique solution to equation $\boldsymbol{\mathcal{G}}_{k}(\boldsymbol{\epsilon}) = \boldsymbol{0}_{K\times 1}$, which can be solved via the Newton's method~\cite{ref:newton-method}. 
        Specifically, the vector updated in the $i$-th iteration, denoted as $\boldsymbol{\epsilon}_{k}^{(i + 1)}$, is obtained as 
        \begin{equation}\label{def:asymp-newton-update}
            \boldsymbol{\epsilon}_{k}^{(i + 1)} = \boldsymbol{\epsilon}_{k}^{(i)} - \boldsymbol{J}_{\boldsymbol{\mathcal{G}}_{k}}(\boldsymbol{\epsilon}_{k}^{(i)})^{-1}\boldsymbol{\mathcal{G}}_{k}(\boldsymbol{\epsilon}_{k}^{(i)}), ~i = 0, \ldots, {I}_{\text{DE}} - 1, 
        \end{equation}
        where ${I}_{\text{DE}}$ is the number of iterations, $\boldsymbol{\epsilon}_{k}^{(0)} = \boldsymbol{0}_{K\times 1}$ is employed as initialization and $\boldsymbol{J}_{\boldsymbol{\mathcal{G}}_{k}}(\boldsymbol{\epsilon}_{k}^{(i)})$ is the Jacobian matrix of function $\boldsymbol{\mathcal{G}}_{k}$ w.r.t. vector $\boldsymbol{\epsilon}$ at $\boldsymbol{\epsilon}_{k}^{(i)}$. 
        Define matrices $\boldsymbol{B}_{k} = [\boldsymbol{b}_{1}, \ldots, \boldsymbol{b}_{k - 1}, \boldsymbol{0}_{K\times 1}, \boldsymbol{b}_{k + 1}, \ldots, \boldsymbol{b}_{K}]\in\mathbb{R}_{+}^{L\times K}$, $\forall k$. 
        The following equations can be verified:
        \begin{subequations}\label{def:asymp-newton-update-details}
            \allowdisplaybreaks
            \begin{align}
                \boldsymbol{\mathcal{G}}_{k}(\boldsymbol{\epsilon}) & = \text{Diag}(\boldsymbol{\epsilon})\boldsymbol{B}^{T}\text{diag}(\boldsymbol{D}_{k}) - \boldsymbol{1}_{K}, \label{subdef:asymp-newton-fval} \\
                \boldsymbol{J}_{\boldsymbol{\mathcal{G}}_{k}}(\boldsymbol{\epsilon}) & = \text{Diag}(\boldsymbol{B}^{T}\text{diag}(\boldsymbol{D}_{k})) - \text{Diag}(\boldsymbol{\epsilon})\boldsymbol{X}_{k}, \label{subdef:asymp-newton-jacobian} \\
                \boldsymbol{X}_{k} & = \boldsymbol{B}^{T}\left(\boldsymbol{D}_{k}\odot\boldsymbol{D}_{k}^{T}\right)\boldsymbol{B}_{k}, \label{subdef:asymp-newton-Xmat} \\
                \boldsymbol{D}_{k} & = \boldsymbol{Q}\left(
                    \boldsymbol{I}_{N} + \boldsymbol{Q}^{H}\text{Diag}(\boldsymbol{B}_{k}\boldsymbol{\epsilon})\boldsymbol{Q}
                \right)^{-1}\boldsymbol{Q}^{H}. \label{subdef:asymp-newton-Dmat}
            \end{align}
        \end{subequations}
        The proof for~\eqref{def:asymp-newton-update-details} is given in Appendix~\ref{appendix-subsec:asymp-newton-jacobian}. 
        The algorithm to solve $\boldsymbol{\epsilon}_{k}$ given any $k$ is summarized in Algorithm ~\ref{alg:auxiliary-newton-method}. 
        It terminates if the relative error of $\boldsymbol{\epsilon}_{k}^{(i)}$ compared with $\boldsymbol{\epsilon}_{k}^{(i - 1)}$ and the total error of equation $\boldsymbol{\mathcal{G}}_{k}(\boldsymbol{\epsilon}_{k}^{(i)}) = \boldsymbol{0}_{K\times 1}$ are smaller than $10^{-3}$, respectively. 
        After solving $\boldsymbol{\epsilon}_{1}, \ldots, \boldsymbol{\epsilon}_{K}$, the asymptotic vector $\boldsymbol{c}^{\infty}$ is given by equations~\eqref{def:asymp-cvec} and ${R}_{\text{ZF}}^{\infty} = \mathcal{R}(\boldsymbol{c}^{\infty})$ can be obtained accordingly. 

        % \vspace{-10pt}
        \begin{algorithm}[t]
            \begin{minipage}{0.95\linewidth}
            \centering
            \caption{Newton's method to solve $\boldsymbol{\epsilon}_{k}$. }
            \begin{algorithmic}[1]\label{alg:auxiliary-newton-method}
                \REQUIRE Transmit FRM $\boldsymbol{Q}$, users' angular power spectrums $\boldsymbol{b}_{1}$, $\ldots$, $\boldsymbol{b}_{K}$, index $k$. 
                \STATE~Let $\varepsilon_{c}, \varepsilon_{g}\gets 1$, $\boldsymbol{\epsilon}_{k}^{(0)}\gets\boldsymbol{0}_{K\times 1}$, and $i\gets 0$. 
                \WHILE{$\varepsilon_{c} \ge 10^{-3}$ or $\varepsilon_{g}\ge 10^{-3}$}
                    \STATE Compute $\boldsymbol{\epsilon}_{k}^{(i + 1)}$ via~\eqref{def:asymp-newton-update} and~\eqref{def:asymp-newton-update-details} and let $i\gets i + 1$. 
                    \STATE Let $\varepsilon_{c}\gets\|\boldsymbol{\epsilon}_{k}^{(i)} - \boldsymbol{\epsilon}_{k}^{(i - 1)}\|_{2}/\|\boldsymbol{\epsilon}_{k}^{(i)}\|_{2}$ and $\varepsilon_{g}\gets\|\boldsymbol{\mathcal{G}}_{k}(\boldsymbol{\epsilon}_{k}^{(i)})\|_2$. 
                \ENDWHILE
                \RETURN Let $\boldsymbol{\epsilon}_{k}\gets\boldsymbol{\epsilon}_{k}^{(i)}$. 
            \end{algorithmic}
            \end{minipage}
        \end{algorithm}
        % \vspace{-10pt}

        Next, the gradient $\nabla_{v}{\bar{R}_{\text{approx}}} = \nabla_{v}{R_{\text{ZF}}^{\infty}}$ is computed for $v\in\{x, y\}$. 
        By applying $\boldsymbol{c}^{\infty}$ to ZF beamforming, we denote the obtained parameter $\nu$ and power allocation vector $\boldsymbol{p}$ as $\nu^{\infty}$ and $\boldsymbol{p}^{\infty}$, respectively. 
        Then, we have 
        \begin{equation}
            \nabla_{v}{R_{\text{ZF}}^{\infty}} = \sum_{k = 1}^{K}{
                \frac{\partial{R_{\text{ZF}}^{\infty}}}{\partial{{c}_{k}^{\infty}}}
                \cdot
                \frac{\mathrm{d}{{c}_{k}^{\infty}}}{\mathrm{d}\boldsymbol{v}}
            } = \left(
                \frac{\mathrm{d}\boldsymbol{c}^{\infty}}{\mathrm{d}{\boldsymbol{v}}}
            \right)^{T}
            \frac{-\boldsymbol{p}^{\infty}}{\nu^{\infty}\ln{2}}, 
        \end{equation}
        where $\frac{\mathrm{d}\boldsymbol{c}^{\infty}}{\mathrm{d}{\boldsymbol{v}}} = [\frac{\mathrm{d}{{c}_{1}^{\infty}}}{\mathrm{d}{\boldsymbol{v}}}, \ldots, \frac{\mathrm{d}{{c}_{K}^{\infty}}}{\mathrm{d}{\boldsymbol{v}}}]^{T}\in\mathbb{R}^{K\times N}$ is the Jacobian matrix of vector $\boldsymbol{c}^{\infty}$ w.r.t. $\boldsymbol{v}$, $\forall v\in\{x, y\}$. 
        By leveraging equations~\eqref{def:asymp-cvec} and~\eqref{def:auxiliary-variables-fp-eq}, $\frac{\mathrm{d}{{c}_{k}^{\infty}}}{\mathrm{d}{\boldsymbol{v}}}$, $\forall k$, is given by 
        \begin{subequations}\label{def:asymp-cvec-gradient-to-position}
            \allowdisplaybreaks
            \begin{align}
                & \frac{\mathrm{d}{{c}_{k}^{\infty}}}{\mathrm{d}{\boldsymbol{v}}} = - ({c}_{k}^{\infty})^{2}\cdot\left(
                    \boldsymbol{R}_{k}^{v}\tilde{\boldsymbol{\chi}}_{k} + \tilde{\boldsymbol{\mu}}_{k}^{v}
                \right), \\
                & \tilde{\boldsymbol{\chi}}_{k} = \boldsymbol{Z}_{k}^{T}\boldsymbol{b}_{k}, ~~\boldsymbol{R}_{k}^{v} = \boldsymbol{U}_{k}^{v}\text{Diag}(\boldsymbol{\epsilon}_{k})\left[
                    \boldsymbol{J}_{\boldsymbol{\mathcal{G}}_{k}}(\boldsymbol{\epsilon}_{k})^{-1}
                \right]^{T}, \label{subdef:asymp-cgrad-chivec-Rmat} \\
                & \tilde{\boldsymbol{\mu}}_{k}^{v} = 2\text{Re}\left(\boldsymbol{F}_{k}^{v}\boldsymbol{b}_{k}\right), ~~\boldsymbol{U}_{k}^{v} = 2\text{Re}\left(\boldsymbol{F}_{k}^{v}\boldsymbol{B}\right), \label{subdef:asymp-cgrad-muvec-Umat} \\
                & \boldsymbol{F}_{k}^{v} = \left[
                    \left(
                        \boldsymbol{I}_{L} - \tilde{\boldsymbol{D}}_{k}\text{Diag}\left(
                            \boldsymbol{\omega}_{k}
                        \right)
                    \right)
                    \text{Diag}\left(
                        j\boldsymbol{\vartheta}^{v}
                    \right)
                    \boldsymbol{Q}
                \right]^{T}\odot\boldsymbol{E}_{k}, \label{subdef:asymp-cgrad-Fmat}
            \end{align}
        \end{subequations}
        where $\boldsymbol{\vartheta}^{v} = [{\kappa}_{1}^{v}, \ldots, {\kappa}_{L}^{v}]^T\in\mathbb{R}^{L\times 1}$, $v\in\{x, y\}$, and $\tilde{\boldsymbol{D}}_{k} = \boldsymbol{Q}\boldsymbol{Y}_{k}^{-1}\boldsymbol{Q}^{H}$ is the special case of $\boldsymbol{D}_{k}$ with $\boldsymbol{\epsilon} = \boldsymbol{\epsilon}_{k}$, while $\boldsymbol{Z}_{k}$, $\boldsymbol{\omega}_{k}$, and $\boldsymbol{E}_{k}$ are defined as
        \begin{equation}\label{def:asymp-deriv-auxiliary}
            \boldsymbol{Z}_{k} = \left(\tilde{\boldsymbol{D}}_{k}\odot\tilde{\boldsymbol{D}}_{k}^T\right)\boldsymbol{B}_{k}, ~\boldsymbol{\omega}_{k} = \boldsymbol{B}_{k}\boldsymbol{\epsilon}_{k}, ~\boldsymbol{E}_{k} = \boldsymbol{Y}_{k}^{-1}\boldsymbol{Q}^H. 
        \end{equation}
        The detailed derivation for~\eqref{def:asymp-cvec-gradient-to-position} is given in Appendix~\ref{appendix-subsec:asymp-gradients-proof}. 
        It can be verified that the total computational complexity for computing $\nabla_{v}{{R}_{\text{ZF}}^{\infty}}$ is given by $\mathcal{O}((N^{3} + NL^{2})KI_{\text{DE}})$, where $I_{\text{DE}}$ denotes the number of iterations of Algorithm~\ref{alg:auxiliary-newton-method}. 

        Notably, the asymptotic approximation only utilizes the channel autocorrelation matrices for users instead of the exact channel distribution. 
        The transmit FRM $\boldsymbol{Q}$ is adjusted by designing antenna positions according to the users' angular power spectrums $\boldsymbol{b}_{1}, \ldots, \boldsymbol{b}_{K}$. 
        This indicates that the antenna positions are optimized to distinguish the angular power spectrums of users such that the correlation between users' channels are reduced. 
        % Different from the Monte-Carlo approximation whose approximation error can be made arbitrarily small given sufficiently large $M$, the approximation error for the asymptotic approximation can only be reduced by increasing ${N}$, $K$, and $L_{k}^{t}$, $\forall k$, which may suggest a relatively larger error for small system sizes. 
        % However, as will be shown in simulations, the asymptotic gradient is effective for MA position optimization despite small system sizes. 
        Moreover, Monte-Carlo simulations are no longer needed for the asymptotic approximation, making it more computationally efficient in practice. 
        Due to the second-order convergence of the Newton's method, $I_{\text{DE}}$ is usually low and the computational cost for asymptotic approximation method is much lower than that of the Monte-Carlo approximation method, which generally requires large $M$ for better approximation accuracy. 
        % Different from the stochastic gradient where Monte-Carlo simulations are employed for $\boldsymbol{H}$, the asymptotic gradient only utilizes the second moments of the channel, i.e., the channel autocorrelation matrices for users. 

\vspace{-6pt}
\section{Performance Evaluation}\label{sec:performance-evaluation}
    % \vspace{-2pt}

    \subsection{Simulation Setup}\label{subsec:simulation-setup}
        
        %%%%%%%%%%%%%%%%%%%%%%%%%%%% combined version %%%%%%%%%%%%%%%%%%%%%%%%%%%%  
        To validate the proposed solutions in a realistic propagation environment with site-specific angular power distributions, the urban map data at Clementi, Singapore~\cite{ref:openstreetmap} is employed for channel generation given each user location realization. 
        Specifically, the site setup is shown in Fig.~\ref{fig:site-view}, where both 2D and 3D views are given. 
        As shown in Fig.~\ref{subfig:site-view-2d}, a rectangular region of $408.44$ meters (m) length and $219.78$ m width is considered with the BS (represented by the red marker) located at the top of the northnmost building with height $59.97$ m. 
        The BS is equipped with $N = 16$ MAs in an antenna moving region of size ${S}_{x} = {S}_{y} = {S}_{0} = 8\lambda$ orientating to the south. 
        The carrier frequency is set as $f_{c} = 5$ GHz and the corresponding wavelength is $\lambda = 5.98$ cm. 
        The transmit power is $P_{T} = 30$ dBm. 
        {\color{\updatecolor}The performance is evaluated by averaging the ergodic sum rate over multiple realizations of users' reference locations $(\boldsymbol{u}_{1}^{\text{ref}}, \ldots, \boldsymbol{u}_{K}^{\text{ref}})$}. 
        Specifically, users' reference locations are randomly selected within the site of interest and ray-tracing is applied to obtain all the paths from the BS to these locations, including their AoDs and channel power gains, where the maximum number of reflections for all NLoS paths is set as $N_{\text{reflect}} = 3$. 
        {\color{\updatecolor}Based on them, the corresponding statistical CSI is obtained. 
        Then, the ergodic sum rate is calculated as the empirical expectation of $100$ instantaneous sum rates, which is further averaged over $100$ realizations of users' reference locations, such that the averaged communication performance for users within the site of interest can be evaluated. }
        To avoid excessive computational cost for time-consuming ray-tracing in our simulations, ${K}_{\text{c}} = 200$ candidate reference locations are generated within the site beforehand as shown in Fig.~\ref{fig:site-view}, from which users' reference locations are randomly selected. 
        % For each candidate location, the AoDs and power gains of all paths from the BS is solved via ray-tracing and stored before the simulation starts. 
        % Then, the user reference locations are randomly selected from these candidate locations each time, such that the ray-tracing results can be directly obtained. 
        % However, the received power in some parts of the site may be negligible due to blockage and multiple reflections. 
        % Thus, the candidate locations are uniformly generated only for regions with less-reflected paths. 
        % Specifically, the $200$ candidate locations are generated sequentially, where each candidate location is repeatedly generated following the uniform distribution within the whole site with height $1$ m above the ground (or building ceiling, depending on where it stands), until there is at least one path with no more than $N_{\text{reflect}}$ reflections between it and the BS. 
        For instance, the rays generated for the $44$, $188$, and $196$-th candidate locations are shown in Fig.~\ref{subfig:hotspot-rays}, each of which has $12$ paths from the BS. 
        
        \begin{figure}[t!]
            \centering
            {
                \begin{subfigure}[t]{0.48\textwidth}
                    \centering
                    \includegraphics[scale = 0.16]{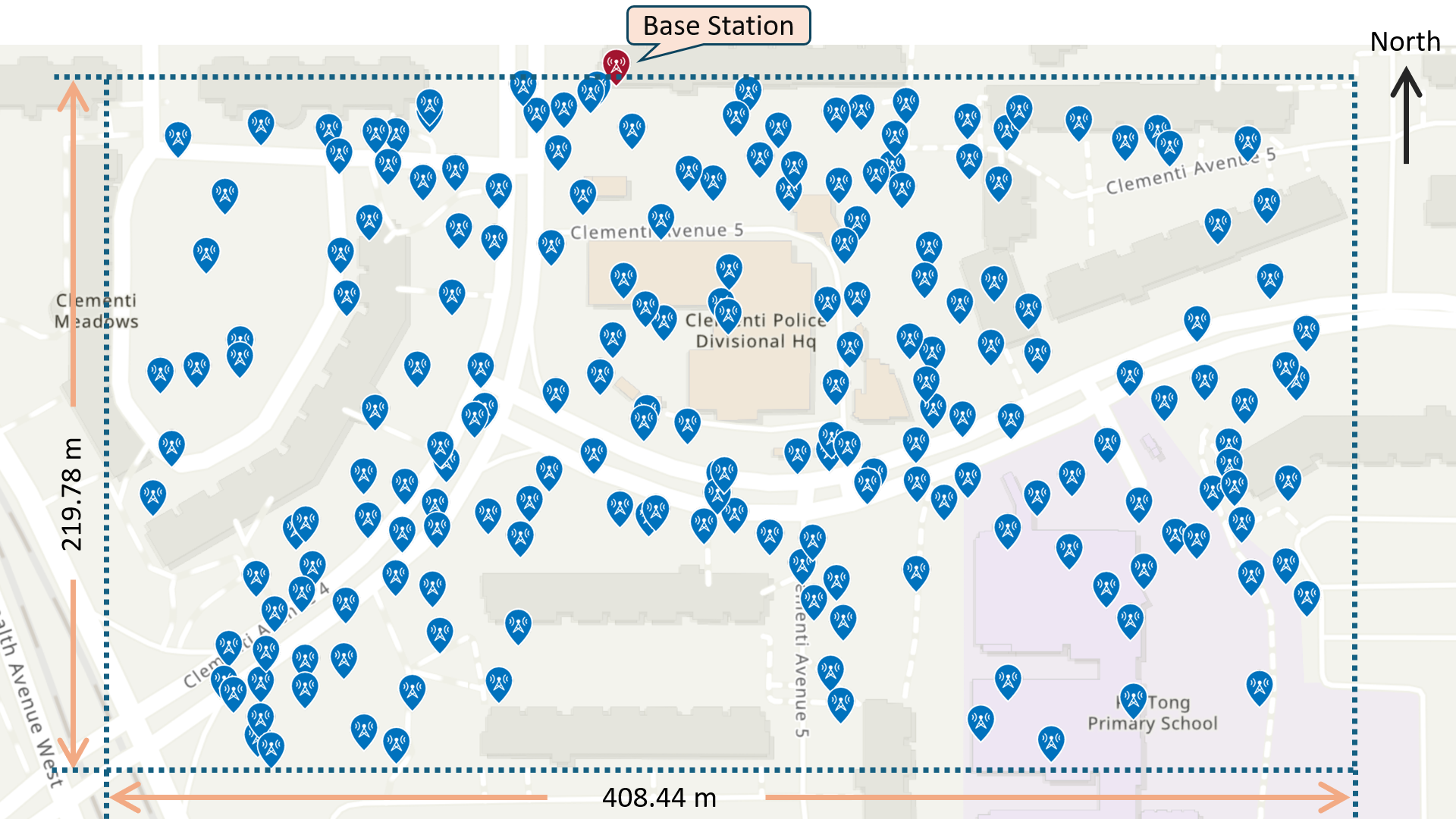}
                    \caption{Site setup (2D view from top). }
                    \label{subfig:site-view-2d}
                    \vspace{2pt}
                \end{subfigure}
                \begin{subfigure}[t]{0.238\textwidth}
                    \centering
                    \includegraphics[scale = 0.17]{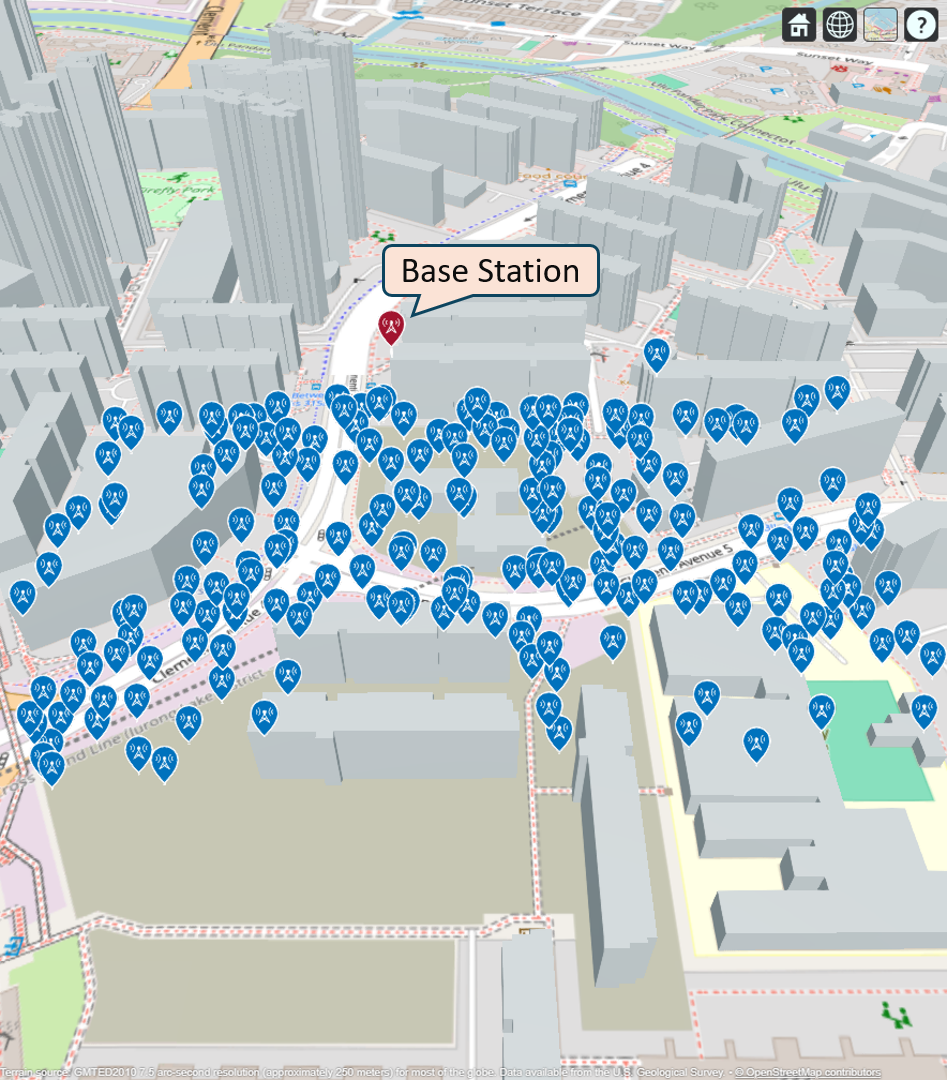}
                    % \vspace{-5pt}
                    \caption{Site setup (3D view). }
                    \label{subfig:site-view-3d}
                \end{subfigure}
                \begin{subfigure}[t]{0.232\textwidth}
                    \centering
                    \includegraphics[scale = 0.17]{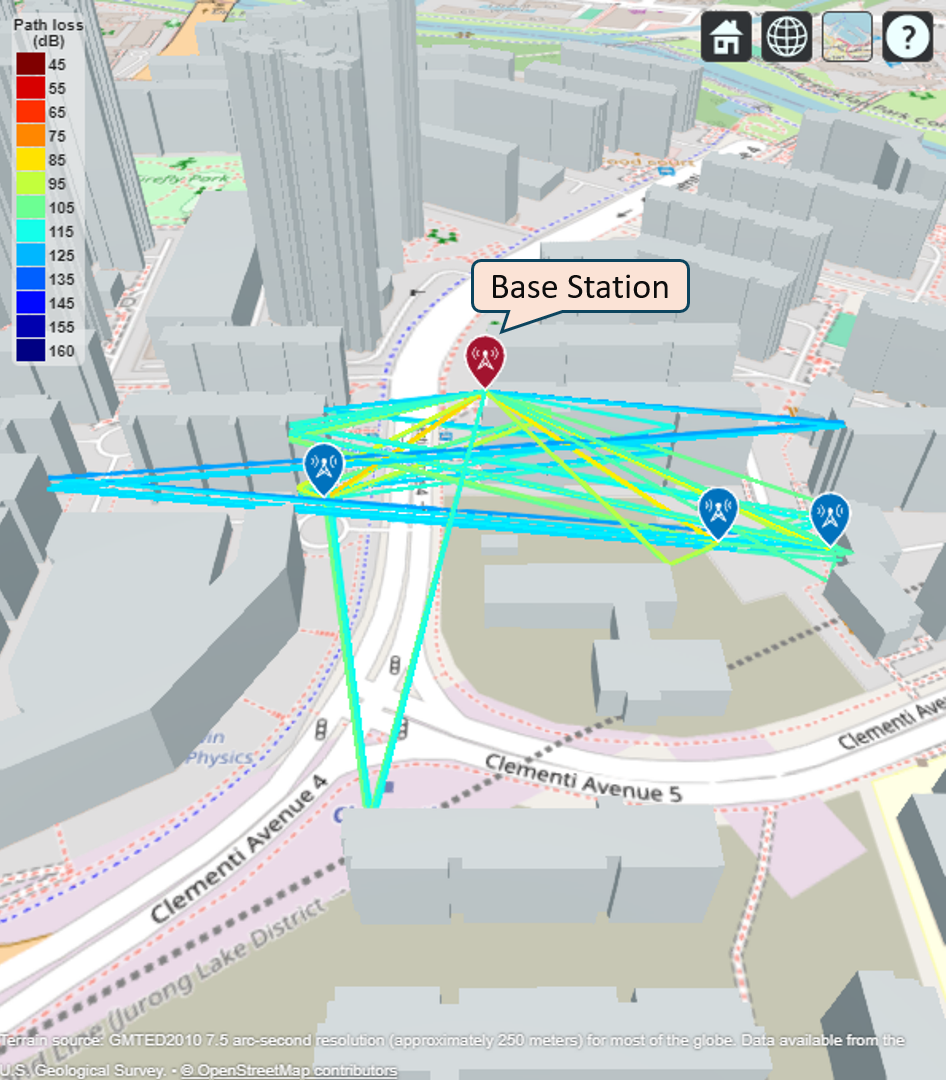}
                    % \vspace{-5pt}
                    \caption{Rays for the $44$, $188$, and $196$-th candidate locations. }
                    \label{subfig:hotspot-rays}
                \end{subfigure}
            }
            \vspace{-2pt}
            \caption{Site environment setup with $200$ randomly generated user candidate locations. }
            \label{fig:site-view}
            \vspace{-2pt}
        \end{figure}
        %%%%%%%%%%%%%%%%%%%%%%%%%%%% combined version ends %%%%%%%%%%%%%%%%%%%%%%%%%%%%

        Next, the acquisition of the statistical CSI from ray-tracing results given the users' reference locations is explained. 
        Note that the each path obtained via ray-tracing corresponds to a unique pair of transmit AoD and receive AoA because only mirror reflections of signals are considered. 
        To account for the complicated scattering environment in practice, the channel power gain of each path is considered as the expected power response corresponding to its transmit AoD, i.e., ${b}_{kl}$ for the $l$-th transmit channel path of user $k$. 
        Then, the statistical channel model can be established following Section~\ref{subsec:statistical-channel-model}. 
        Additionally, to study the influence of Rician factor on the performance, two parameters $\eta_{\text{LoS}}$ and $\eta_{\text{NLoS}}$ are applied to scale the LoS and NLoS channels, respectively. 
        Particularly, they are selected such that the ratio of users' expected LoS and NLoS power equals to the Rician factor ${\beta}$ while the total expected channel power remains constant. 
        By denoting $\bar{P}_{\text{LoS}}$ and $\bar{P}_{\text{NLoS}}$ as the users' expected LoS and NLoS channel power within the site without scaling, respectively, the expected total channel power is given by $\bar{P}_{\text{total}} = \bar{P}_{\text{LoS}} + \bar{P}_{\text{NLoS}}$, where $\eta_{\text{LoS}}$ and $\eta_{\text{NLoS}}$ are written as  
        \begin{equation}\label{def:rician-parameters}
            \eta_{\text{LoS}} = \left(
                \frac{\bar{P}_{\text{total}}}{\bar{P}_{\text{LoS}}}\frac{\beta}{1 + \beta}
            \right)^{\frac{1}{2}}, 
            ~\eta_{\text{NLoS}} = \left(
                \frac{\bar{P}_{\text{total}}}{\bar{P}_{\text{NLoS}}}\frac{1}{1 + \beta}
            \right)^{\frac{1}{2}}. 
        \end{equation}
        As such, we have $\eta_{\text{LoS}}^{2}\bar{P}_{\text{LoS}}/\eta_{\text{NLoS}}^{2}\bar{P}_{\text{NLoS}} = \beta$ while the total channel power $\eta_{\text{LoS}}^{2}\bar{P}_{\text{LoS}} + \eta_{\text{NLoS}}^{2}\bar{P}_{\text{NLoS}} = \bar{P}_{\text{total}}$ is unchanged. 
        
        Unless otherwise noted, we set $K = 12$, and $\beta = 10$ dB. 
        Additionally, we set noise power $\sigma^2 = -90$ dBm, $\Delta = \lambda/2$, and $M = 30$ for the Monte-Carlo approximation\footnote{Note that the $M$ channel realizations are employed to compute $\bar{R}_{\text{MC}}$ in~\eqref{def:ergodic-sum-rate-monte-carlo-approx}, which should be independent of the $100$ channel realizations used for evaluating the ergodic sum rate after solving the antenna positions. }. 
        Parameters for Algorithm~\ref{alg:log-barrier-gradient-ascent-framework} and~\ref{alg:auxiliary-newton-method} are set as $\mu_{0} = 1$, $\varepsilon_{r} = 0.01\lambda$, ${\alpha}_{0} = 0.15\lambda$, $\eta = 0.2$, $I = 20$, and $\rho = 0.4$. 
        % All results are averaged over at least $100$ user reference locations realizations, for each of which the ergodic sum rate is averaged over $100$ instantaneous channel realizations. 

    \vspace{-6pt}
    \subsection{Proposed and Benchmarks Schemes}\label{subsec:perf-benchmarks}
        In the following, the proposed algorithms using approximate gradients given by Sections~\ref{subsec:monte-carlo-approx} and~\ref{subsec:asymptotic-approximation} are labeled as `MA-statistical-MC' and `MA-statistical-DE', respectively, where `MC' referes to Monte-Carlo while `DE' referes to deterministic equivalents (based on asymptotic analysis). 
        To validate the effectiveness of the proposed algorithms, three benchmarks are considered for performance comparison, which are listed as follows: 
        i) \textbf{MA-instantaneous}: 
        The MA positions are optimized based on instantaneous CSI. 
        The proposed log-barrier penalized gradient ascent algorithm is employed to optimize antenna positions but gradient $\nabla_{v}\bar{R}_{\text{approx}}$ is replaced by $\nabla_{v}{R}_{ZF}$, which can be obtained via equation~\eqref{def:inst-sum-rate-gradients}. 
        The resulting ergodic sum rate serves as an upper bound on that of the proposed MA optimization method based on statistical CSI. 
        ii) \textbf{UPA-dense}: The conventional dense UPA is employed, where the antennas forms a $4\times 4$ array with inter-antenna separation $\lambda/2$. 
        iii) \textbf{UPA-sparse}: The antennas are sparsely placed into a $4\times 4$ array with inter-antenna separations ${S}_{x}/4$ and ${S}_{y}/4$ along the $x$ and $y$ axes, respectively. 

        Additionally, for all schemes except MA-instantaneous, the asymptotic approximation for ergodic sum rate, i.e., ${R}_{\text{ZF}}^{\infty}$, can be computed given the statistical CSI and the optimized antenna positions. 
        To verify the effectiveness of this approximation, ${R}_{\text{ZF}}^{\infty}$ is also shown in simulation results, labeled with the same name as the scheme used for antenna position optimization but followed by `asymptotic'. 
        For example, `MA-statistical-MC, asymptotic' represents the asymptotic approximate ergodic sum rate given the antenna positions optimized via `MA-statistical-MC'.

    \vspace{-6pt}
    \subsection{Algorithm Convergence}\label{subsec:perf-convergence-and-asymp-errors}
    % 2 figures required: optimized MA positions, ergodic sum rate vs iterations. 
    % 2 figures optional: ergodic sum rate compared with de sum rate, ergodic gradient vs de gradient. 

        \begin{figure}[t!]
            \centering
            {
                \begin{subfigure}[t]{0.235\textwidth}
                    \centering
                    % \hspace{-3pt}
                    \includegraphics[scale = 0.37]{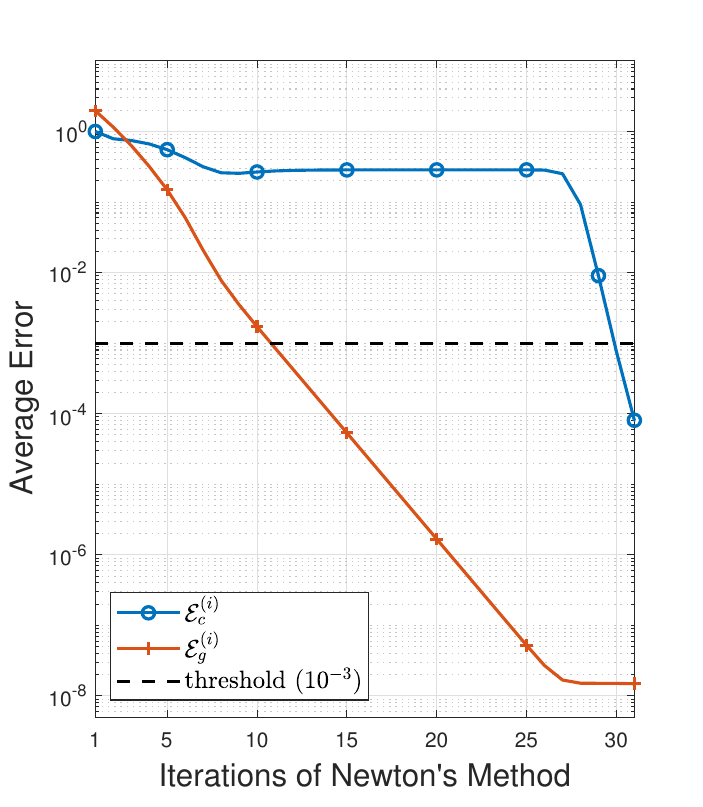}
                    % \vspace{-5pt}
                    \caption{Averaged $\mathcal{E}_{c}^{(i)}$ and $\mathcal{E}_{g}^{(i)}$ during iterations of Algorithm~\ref{alg:auxiliary-newton-method}. }
                    \label{subfig:asymp-error-evolve}
                \end{subfigure}
                \begin{subfigure}[t]{0.235\textwidth}
                    \centering
                    % \hspace{-3pt}
                    \includegraphics[scale = 0.37]{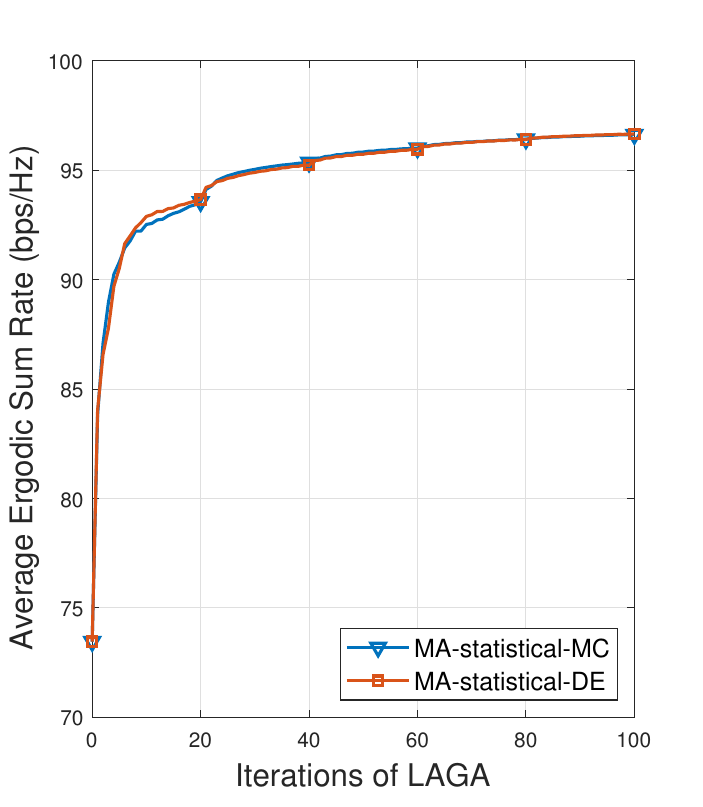}
                    % \includegraphics[scale = 0.31]{figures/conv/rate-evolve.eps}
                    % \vspace{-5pt}
                    \caption{Ergodic sum rate during iterations of Algorithm~\ref{alg:log-barrier-gradient-ascent-framework}. }
                    \label{subfig:rate-evolve}
                \end{subfigure}
            }
            \vspace{-2pt}
            \caption{Convergence of proposed algorithms. }
            % \caption{Visualization of antenna positions optimization. }
            \label{fig:visualiztaion}
            \vspace{-12pt}
        \end{figure}

        \begin{figure}[t]
            \centering
            \includegraphics[scale = 0.4]{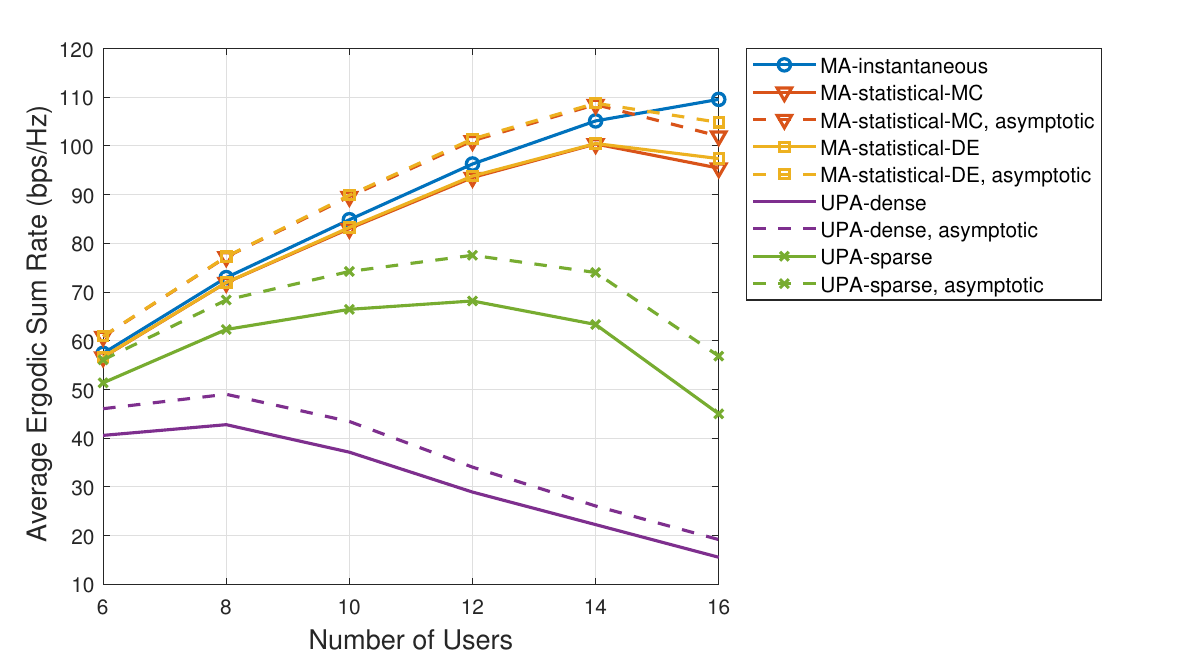}
            \caption{Ergodic sum rate versus the number of users which are uniformly distributed in the site of interest. }
            \label{fig:uniform-usernum-test}
        \end{figure}

        In Fig.~\ref{fig:visualiztaion}, the convergence of proposed algorithms is shown. 
        {\color{\updatecolor}For the $i$-th iteration of Algorithm~\ref{alg:auxiliary-newton-method}, we define averaged errors $\mathcal{E}_{c}^{(i)} = \frac{1}{K}\sum_{k = 1}^{K}{\|\boldsymbol{\epsilon}_{k}^{(i)} - \boldsymbol{\epsilon}_{k}^{(i - 1)}\|_{2}/\|\boldsymbol{\epsilon}_{k}^{(i)}\|_{2}}$ and $\mathcal{E}_{g}^{(i)} = \frac{1}{K}\sum_{k = 1}^{K}{\|\boldsymbol{\mathcal{G}}_{k}(\boldsymbol{\epsilon}_{k}^{(i)})\|_2}$, representing the average relative error of $\boldsymbol{\epsilon}_{k}^{(i)}$ compared with $\boldsymbol{\epsilon}_{k}^{(i - 1)}$ over all users and the average error between $\boldsymbol{\mathcal{G}}_{k}(\boldsymbol{\epsilon}_{k}^{(i)})$ and $\boldsymbol{0}_{K\times 1}$, respectively. 
        Given UPA-sparse and after being further averaged over $500$ realizations of users' reference locations, the averaged $\mathcal{E}_{c}^{(i)}$ and $\mathcal{E}_{g}^{(i)}$ are plotted in Fig.~\ref{subfig:asymp-error-evolve}. 
        Both of them keep decreasing during the iterations, which indicates that vector $\boldsymbol{\epsilon}_{k}^{(i)}$ converges to the unique solution of equation $\boldsymbol{\mathcal{G}}_{k}(\boldsymbol{\epsilon}) = \boldsymbol{0}_{K\times 1}$ for any $k$, i.e., $\boldsymbol{\epsilon}_{k}$. }
        Notably, Algorithm~\ref{alg:auxiliary-newton-method} terminates after about $30$ iterations, enabling efficient implementation in practice. 
        Meanwhile, the average ergodic sum rates during iterations of Algorithm~\ref{alg:log-barrier-gradient-ascent-framework} are shown in Fig.~\ref{subfig:rate-evolve} for both MA-statistical-MC and MA-statistical-DE. 
        It can be observed that both of them gradually increase and converge to almost the same level, which validates the effectiveness of the proposed LAGA algorithm in improving ergodic sum rate based on statistical CSI. 
        Additionally, $\mu$ is shrinked by $\rho$ for every $20$ iterations. 
        For each value of $\mu$, the gradient ascent converges to a locally optimal solution for problem~\eqref{prob:ma-optm-log-barrier-relaxed}. 
        After the value of $\mu$ is updated, the antenna positions will be adjusted for the next locally optimal solution, giving rise to a suddenly larger growth rate of the ergodic sum rate, as shown in Fig.~\ref{subfig:rate-evolve}.

        \begin{figure*}[t!]
            \centering
            % \vspace{-12pt}
            {
                \begin{subfigure}[t]{0.32\textwidth}
                    \centering
                    \includegraphics[scale = 0.38]{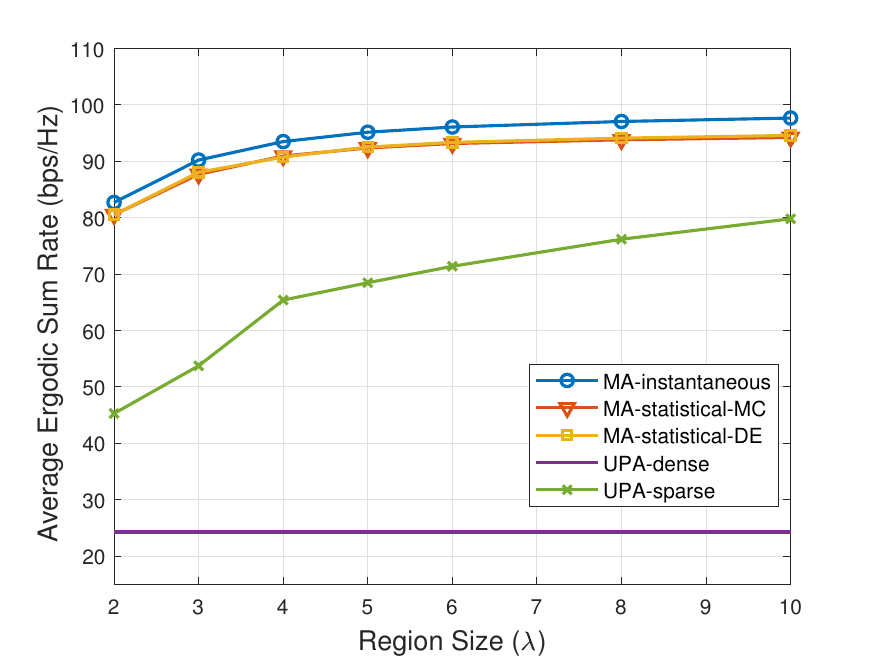}
                    \caption{Ergodic sum rate versus region size $S_0$. }
                    \label{subfig:uniform-aperture-test}
                \end{subfigure}
                % \hspace{2pt}
                \begin{subfigure}[t]{0.32\textwidth}
                    \centering
                    \includegraphics[scale = 0.38]{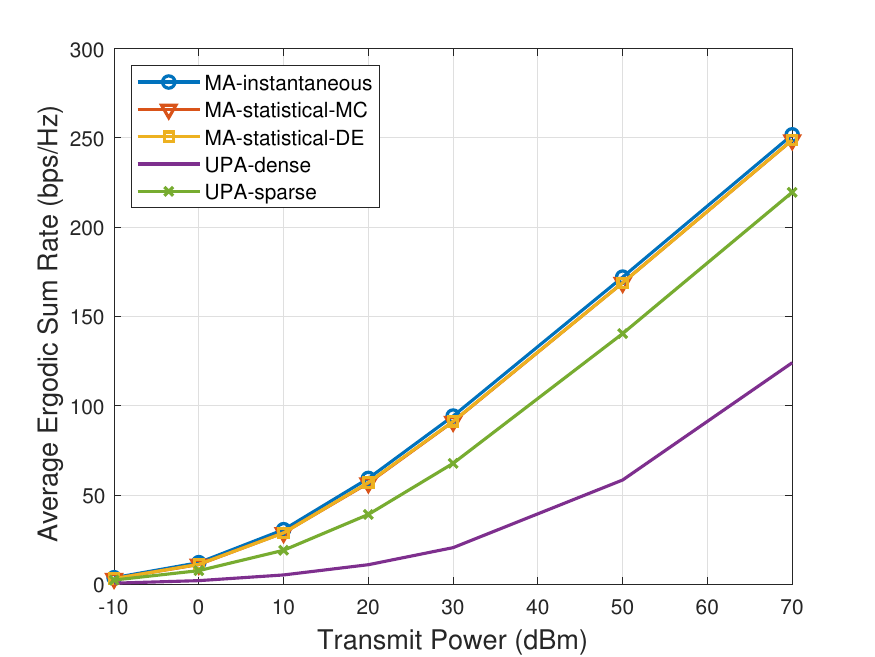}
                    \caption{Ergodic sum rate versus $P_{T}$. }
                    \label{subfig:uniform-txpwr-test}
                \end{subfigure}
                % \hspace{2pt}
                \begin{subfigure}[t]{0.32\textwidth}
                    \centering
                    \includegraphics[scale = 0.38]{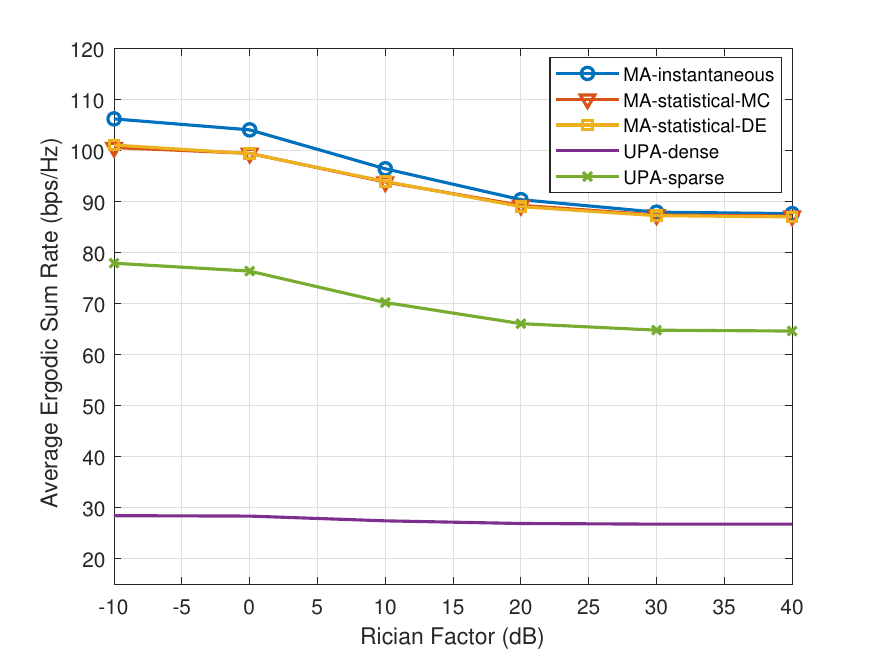}
                    \caption{Ergodic sum rate versus $\beta$. }
                    \label{subfig:uniform-rician-test}
                \end{subfigure}
            }
            \vspace{-2pt}
            \caption{Performance with users uniformly distributed in the site. }
            \label{fig:uniform-perf-test}
            \vspace{-6pt}
        \end{figure*}
        
        \begin{figure*}[t!]
            \centering
            % \vspace{-6pt}
            {
                \begin{subfigure}[t]{0.32\textwidth}
                    \centering
                    % \hspace{-3pt}
                    % \includegraphics[scale = 0.38]{figures/clustered/acc-ver/aper-acc-32-02.eps}
                    \includegraphics[scale = 0.38]{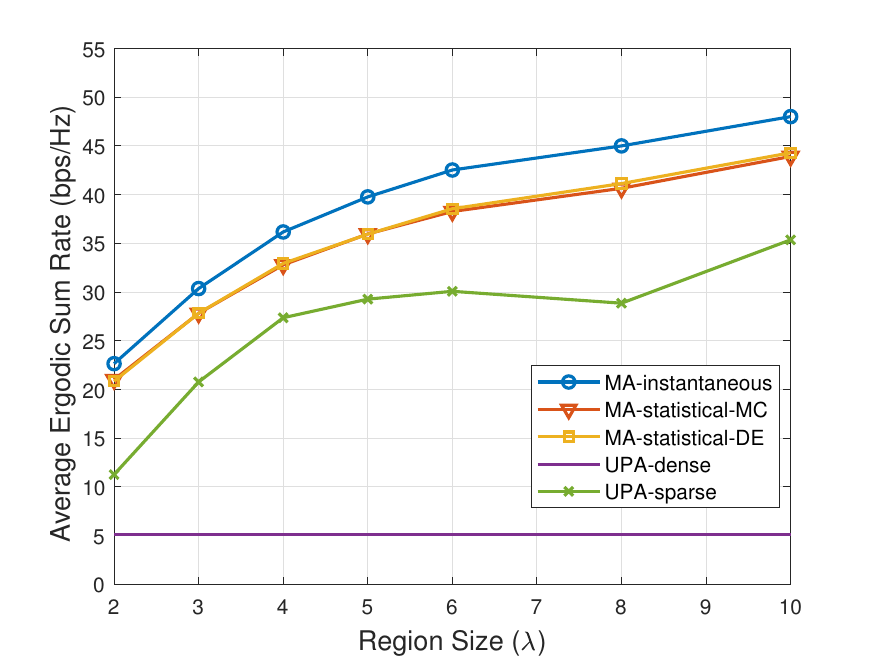}
                    % \vspace{-5pt}
                    \caption{Ergodic sum rate versus region size $S_0$. }
                    \label{subfig:hotspot-aperture-test}
                \end{subfigure}
                % \hspace{2pt}
                \begin{subfigure}[t]{0.32\textwidth}
                    \centering
                    % \hspace{-3pt}
                    % \includegraphics[scale = 0.38]{figures/clustered/acc-ver/txpwr-test.eps}
                    \includegraphics[scale = 0.38]{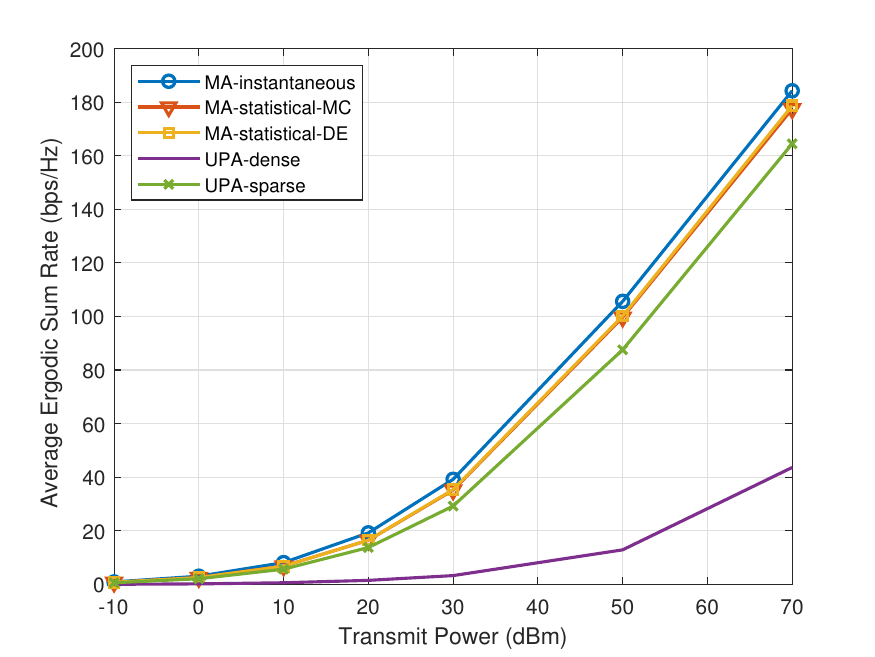}
                    % \vspace{-5pt}
                    \caption{Ergodic sum rate versus $P_{T}$. }
                    \label{subfig:hotspot-txpwr-test}
                \end{subfigure}
                % \hspace{2pt}
                \begin{subfigure}[t]{0.32\textwidth}
                    \centering
                    % \hspace{-10pt}
                    % \includegraphics[scale = 0.38]{figures/clustered/acc-ver/rician-test.eps}
                    \includegraphics[scale = 0.38]{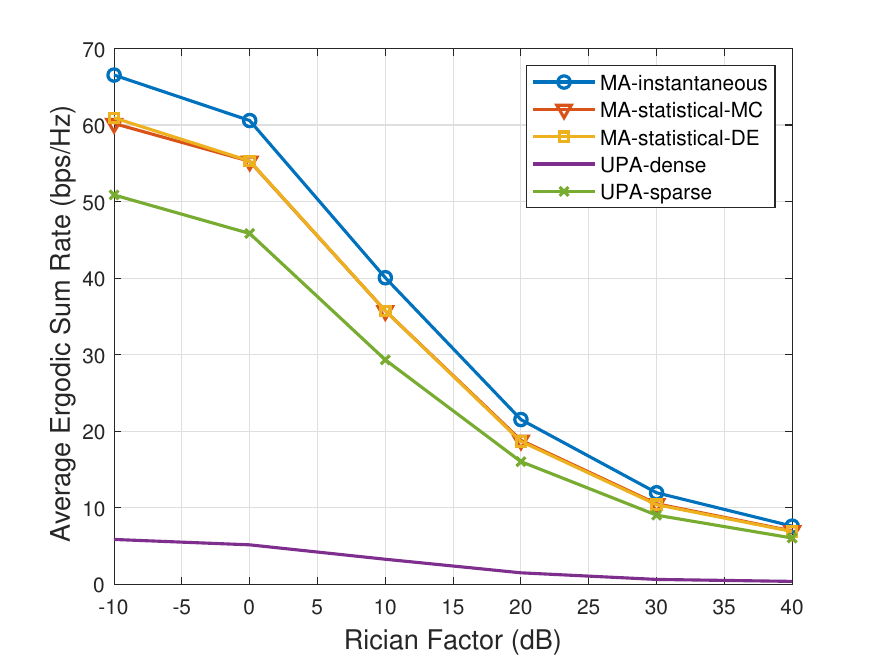}
                    % \vspace{-5pt}
                    \caption{Ergodic sum rate versus $\beta$. }
                    \label{subfig:hotspot-rician-test}
                \end{subfigure}
            }
            % \vspace{-2pt}
            \caption{Performance with users clustered in three hotspots (the $46$-th, $188$-th, and $196$-th candidate locations). }
            \label{fig:hotspot-perf-test}
            \vspace{-12pt}
        \end{figure*}

    \vspace{-6pt}
    \subsection{Performance Evaluation with Uniformly Distributed Users}\label{subsec:perf-uniform-users}
    % 5 figures: ergodic sum rate vs Rician factor, path number, array size, user number, transmit power. 
        In this subsection, the performance improvements in ergodic sum rate achieved by the proposed algorithms are demenstrated with users randomly selected from all the candidate locations with the same probability, which can be interpreted as that the users are uniformly distributed within the site of interest. 
        
        The average ergodic sum rate is shown in Fig.~\ref{fig:uniform-usernum-test} with different number of users, where the approximate ergodic sum rates $R_{\text{ZF}}^{\infty}$ given the antenna positions obtained by the proposed algorithms as well as UPA-dense and UPA-sparse are also shown. 
        It can be observed that UPA-sparse outperforms UPA-dense significantly, which is resulted from the fact that users can be distinguished more easily for UPA-sparse due to its higher angular beam resolution compared with UPA-dense~\cite{ref:zlp-ma-near-field-statistical}. 
        On top of that, the average ergodic sum rates obtained by MA-statistical-MC and MA-statistical-DE are approximately the same, which are much higher than that of UPA-sparse and are even close to that of MA-instantaneous. 
        Such results validate the effectiveness of the proposed algorithms and indicate that designing the antenna positions based on statistical CSI may be sufficient to reap almost the same performance gain as that obtained by real-time MA optimization based on instantaneous CSI. 
        Besides, as the user number $K$ grows larger, the performance gains of the proposed algorithms over UPA-dense and UPA-sparse becomes more significant. 
        % ergodic sum rate obtained by MA-instantaneous keeps increasing while those of all other approaches decrease when $K$ becomes too large. 
        % This is because the correlations between users get larger when $K$ is larger. 
        When $K$ is small, the channel correlations between users are low and the ergodic sum rate increases almost linearly with $K$ for all schemes. 
        As $K$ increases, the channel correlations between users become higher for FPA systems, hindering further improvements for the ergodic sum rate. 
        However, by designing antenna positions according to the instantaneous/statistical CSI, such correlations are reduced and thus higher performance gain can be achieved by MA systems. 
        % As a result, the ergodic sum rates of the proposed algorithms keep increasing until $K = 16$, while UPA-sparse and UPA-dense have been suffering from high user channel correlations for smaller $K$. 
        For example, for $K = 14$ and $K = 16$, the ergodic sum rates are improved via adjusting antenna positions based on statistical CSI by $59.3\%$ and $120\%$ over UPA-sparse and both more than $300\%$ over UPA-dense, respectively.

        Next, the effects of MA's region size $S_{0}$, transmit power $P_{T}$, and Rician factor $\beta$ on the average ergodic sum rate are shown in Figs.~\ref{subfig:uniform-aperture-test},~\ref{subfig:uniform-txpwr-test}, and~\ref{subfig:uniform-rician-test}, respectively. 
        % In particular, the antennas are formed into $N_{x}\times N_{y}$ arrays for UPA-dense and UPA-sparse in Fig.~\ref{subfig:uniform-arraysize-test} with $N_{x} = N_{y} = \sqrt{{N}}$. 
        In all figures, the performances of MA-statistical-MC and MA-statistical-DE coincide with each other with almost negligible errors. 
        As the region size $S_{0}$ and transmit power $P_{T}$ grow larger, the performances of all schemes improve, where the ergodic sum rates obtained by two proposed algorithms are always close to that of MA-instantaneous and outperform UPA-sparse and UPA-dense at the same time. 
        % When ${N}$ is sufficiently large, the channels for different user locations tend to be approximately orthogonal, which reduces the channel correlations between users. 
        % Thus, the improvements brough by MAs shrink gradually compared to UPA-sparse, as shown in Fig.~\ref{subfig:uniform-arraysize-test}. 
        From Fig.~\ref{subfig:uniform-aperture-test}, it can be observed that the ergodic sum rate performance gradually converges as $S_{0}$ increases and $S_{0} = 4\lambda$ is almost sufficient to capture most of the improvements by MAs. 
        Meanwhile, since the transmit power does not affect the channel correlations between users, the improvements of MA schemes against UPA-sparse is almost constant in Fig.~\ref{subfig:uniform-txpwr-test} for $P_{T}\ge 30$ dBm. 
        In contrast, the influence of the Rician factor $\beta$ on the performance as shown in Fig.~\ref{subfig:uniform-rician-test} is more interesting. 
        In general, the average ergodic sum rates are lower for a larger Rician factor, i.e., a more LoS-dominant environment. 
        Provided that the total channel power is fixed, most power is concentrated in the LoS paths given a large $\beta$, which results in difficulties in decorrelating users with highly-correlated LoS paths by ZF beamforming. 
        However, if the channel power is more equally assigned to all paths, e.g., $\beta$ is smaller, users can be easily distinguished by leveraging their uncorrelated NLoS paths even if their LoS paths are close to each other, leading to better performance.

    \vspace{-6pt}
    \subsection{Performance Evaluation with Clustered Users}\label{subsec:perf-clustered-users}
    % scattered vs clustered users: which is better?
    % case 1: all users scattered uniformly all over. 
    % case 2: all users in one cluster unblocked. 
    % case 3: all users in one cluster but blocked. 
    % data-source: 5GHz-r0.5-new
    %     with 12 paths: #46, 168, 188, 196
    %     max power dta:  38,  41,  34,  28
    %     ====> #196 chosen as the cluster. 
        To investigate the impact of user distribution on the system performance, the average ergodic sum rate is computed with users clustered in three hotspots at the $46$-th, $188$-th, and $196$-th candidate locations, which are shown in Fig.~\ref{subfig:hotspot-rays}. 
        Specifically, each user is randomly generated within one of these three candidate locations\footnote{{\color{\updatecolor}Since only $3$ locations are available, it can be verified that there are $\mathcal{J}(K) = (K + 1)(K/2 + 1)$ possible realizations of users' reference locations in total, ignoring the order of users. Thus, for small $K$, the average ergodic sum rates are evaluated with fewer than $100$ reference locations realizations. For example, only $\mathcal{J}(12) = 91$ reference locations are used for $K = 12$. }}. 
        % Note that the total number of paths $L$ should be no smaller than the user number $K = 12$ for ZF beamforming to work. 
        % In the worst case, all users are placed into the same candidate location, which requires at least $K$ paths for this candidate location. 
        % Due to this reason, the $46$-th, $188$-th, and $196$-th candidate locations are chosen as the hotspots in this subsection because each of them has $12$ paths. 

        The relationships between the average ergodic sum rate and the region size, transmit power, and Rician factor are shown in Figs.~\ref{subfig:hotspot-aperture-test},~\ref{subfig:hotspot-txpwr-test}, and~\ref{subfig:hotspot-rician-test}, respectively. 
        The first observation is that the performance improvements against UPA-sparse brought by MAs are reduced drastically compared to systems with uniform user distribution. %, where the users are uniformly distributed in the site. 
        Such changes are caused by the clustered user distribution, which incurs very similar angular power spectrums for different users and thus high user channel correlations, especially for those in the same hotspot. 
        % Consequently, there is not much space for user channel correlation reduction except for exploiting the concentrated elevation AoD distribution that has been utilized by UPA-sparse, which limits further performance improvements by the MA-based schemes. 
        Consequently, it is more difficult to reduce channel correlation between different users, which limits further performance improvements by MAs. 
        % In comparison, the AoDs of channel paths are more diverse and the power distributions for users are different given uniform user distibution, making it easier to suppress users' channel correlations. 
        {\color{\updatecolor}However, it can be noticed in Fig.~\ref{subfig:hotspot-aperture-test} that the performance of UPA-sparse decreases at $S_{0} = 8\lambda$. 
        This can be explained by the sidelobes caused by UPA-sparse~\cite{ref:zlp-ma-near-field-statistical} that happen to cover two of the hotspots, introducing strong interference for users. 
        Nevertheless, such interference can be mitigated by MAs using the proposed algorithms, as shown in Fig.~\ref{subfig:hotspot-aperture-test}, which also validates the advantages of the proposed schemes. }
        In addition, it is also worth noting that the ergodic sum rates for all schemes except UPA-dense decreases rapidly with $\beta$ in Fig.~\ref{subfig:hotspot-rician-test}. 
        This is resulted from the overwhelming dominance of the LoS paths as $\beta$ grows larger, which induces even higher correlations between users since there are only $3$ significantly different LoS paths for these $12$ users from the three hotspots. 
        
        For clearer illustration of the relationship between the ergodic sum rate and the user distribution, we further consider the case where some users are clustered in one hotspot while the rest of them are uniformly distributed within the site. 
        Specifically, define $\tau\in[0, 1]$ as the portion of users clustered around the $196$-th candidate location, which is also referred to as the cluster rate. 
        Thus, each user has a probability of $\tau$ to be placed around the $196$-th candidate location and a probability of $1 - \tau$ to be randomly placed at any one of the other candidate locations. 
        By adjusting $\tau$ from $0\%$ to $100\%$, the corresponding average ergodic sum rates are plotted in Fig.~\ref{fig:focused-clusterrate-test}. 
        As can be observed, the ergodic sum rates for all schemes keep decreasing as $\tau$ increases. 
        It is worth noting that the ergodic sum rate improvements by the MA-aided schemes over UPA-sparse are also shrinking as $\tau$ grows larger. 
        With larger portion of users clustered in the same hotspot, the paths and their power gains for different users gradually coincides with each other, inducing highly overlaped angular power spectrums for users. 
        Combined with the discussions above, we conclude that higher performance improvements can be obtained via designing MA positions for an environment with more diverse angular power spectrums for users.

        \begin{figure}[t]
            \centering
            \includegraphics[scale = 0.38]{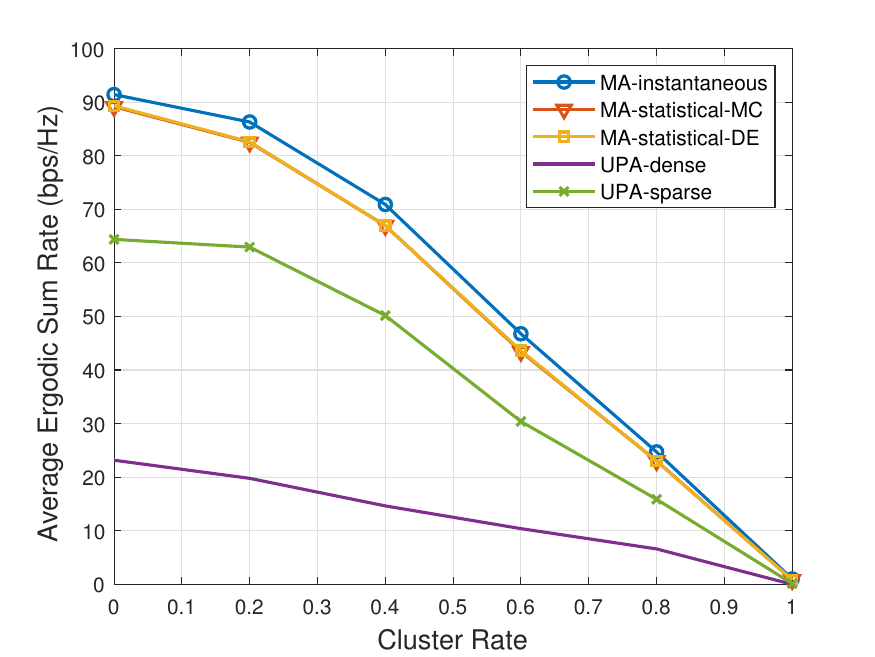}
            % \vspace{-5pt}
            \caption{Ergodic sum rate versus the cluster rate of users in one hotspot (the $196$-th candidate location). }
            \label{fig:focused-clusterrate-test}
        \end{figure}

% \vspace{-6pt}
\section{Conclusion}\label{sec:conclusion}
    
    In this paper, we investigated antenna position optimization for MA-aided MU-MIMO systems using statistical CSI, aiming to maximize the ergodic sum rate of all users. 
    By optimizing antenna positions over a large timescale, the proposed approach is able to exploit the channel statistics associated with the distibutions of users and scatterers, which thus achieves significantly lower energy consumption and time overhead for antenna movement compared to real-time MA configurations based on instantaneous CSI. 
    To capture channel variations induced by users' local movements, we developed a field response based statistical channel model tailored for MA-aided multiuser communication systems. 
    Different from the conventional Rician fading channel model, our proposed model incorporates the effects of dominant and local scatterers within the site of interest and accounts for multiple channel paths between the BS and each user, offering enhanced generality. 
    Building on this model, a two-timescale optimization problem was formulated, where the precoding matrix and antenna positions were designed based on instantaneous and statistical CSI, respectively. 
    To solve this problem, ZF beamforming was adopted, and a log-barrier penalized gradient ascent algorithm was proposed to optimize antenna positions for ergodic sum rate maximization. 
    Additionally, we presented two methods to calculate the gradient of the ergodic sum rate w.r.t. antenna positions, one using Monte Carlo approximation method and the other leveraging asymptotic analysis. 
    Simulation results demonstrated the significant performance gains of the proposed scheme using statistical CSI, achieving ergodic sum rates comparable to those obtained through instantaneous CSI-based optimization. 
    Furthermore, the results highlighted that performance improvements are more pronounced when users' channels exhibit diverse power distributions in the angular domain, which facilitates the capability of MAs to reduce channel correlations between users. 
    Finally, while the proposed statistical channel model and optimization framework are designed for scenarios where users move within localized regions, they can be extended to accommodate random user movements across larger areas, enabling site-specific designs, which will be explored in future work.

\appendices

\section{Derivations of (24) and (25)}\label{appendix:inst-sum-rate-gradient}
    \subsection{Derivation of (24)}\label{appendix-subsec:inst-sum-rate-derivative-to-uvec}
        Due to the implicit expression of ${R}_{\text{ZF}}$, we derive its derivative w.r.t. ${c}_{k}$ by definition. 
        Specifically, consider $\delta\in\mathbb{R}$ and vector $\boldsymbol{c}' = [c_1, \ldots, c_{k - 1}, {c}_{k} + \delta, c_{k + 1}, \ldots, c_{K}]^T\in\mathbb{R}^{K\times 1}$ given $k$. 
        Then, we have 
        \begin{equation}\label{def:inst-sum-rate-derivative-to-uvec-lim}
            \frac{\mathrm{d}R_{\text{ZF}}}{\mathrm{d}{c}_{k}} = \lim_{\delta\to 0}\frac{\mathcal{R}(\boldsymbol{c}') - \mathcal{R}(\boldsymbol{c})}{\delta}. 
        \end{equation}
        To enable analysis on $\mathcal{R}(\boldsymbol{c}')$, the power allocation parameter $\nu$ plays an essential role. 
        Although the expression of $\nu$ is not explicitly given, its properties can be obtained by utilizing equation (13). 
        First, it is obvious that $\nu > 0$ and is the unique solution satisfying equation (13). 
        Given $\boldsymbol{c}$, define sets 
        \begin{subequations}\label{def:power-alloc-status-sets}
            \begin{align}
                \mathcal{I}_{+}(\boldsymbol{c}) & = \{i | 1\le i\le{K}, \nu - \sigma^2 c_{i} > 0\}, \label{def:power-alloc-positive-sets} \\
                \mathcal{I}_{0}(\boldsymbol{c}) & = \{i | 1\le i\le{K}, \nu - \sigma^2 c_{i} = 0\}, \label{def:power-alloc-zero-sets} \\
                \mathcal{I}_{-}(\boldsymbol{c}) & = \{i | 1\le i\le{K}, \nu - \sigma^2 c_{i} < 0\}, \label{def:power-alloc-negative-sets}
            \end{align}
        \end{subequations}
        and ${K}_{+} = |\mathcal{I}_{+}(\boldsymbol{c})|$, ${K}_{0} = |\mathcal{I}_{0}(\boldsymbol{c})|$, and ${K}_{-} = |\mathcal{I}_{-}(\boldsymbol{c})|$ as their cardinality, respectively. 
        Note that ${K}_{+} > 0$ while ${K}_{0}, {K}_{-}\ge 0$. 
        Then, equation (13) can be rewritten as 
        \begin{equation}\label{eq:nu-expression}
            \nu = \frac{P_{T}}{{K}_{+}} + \frac{\sigma^2}{{K}_{+}}\sum_{i\in\mathcal{I}_{+}(\boldsymbol{c})}{c_{i}}. 
        \end{equation}
        Now that $\boldsymbol{c}'$ is given, the resulting parameter $\nu'$ can be obtained by considering three cases separately, as follows. 

        Firstly, we consider the case of $k\in\mathcal{I}_{+}(\boldsymbol{c})$. 
        If $\delta > 0$, it is easy to verify that $\nu' >\nu$ and, as $\delta\to 0^{+}$, we have $\mathcal{I}_{+}(\boldsymbol{c}') = \mathcal{I}_{+}(\boldsymbol{c})\cup\mathcal{I}_{0}(\boldsymbol{c})$, $\mathcal{I}_{0}(\boldsymbol{c}') = \emptyset$, and $\mathcal{I}_{-}(\boldsymbol{c}') = \mathcal{I}_{-}(\boldsymbol{c})$. 
        % \begin{subequations}
        %     \begin{gather}
        %         \mathcal{I}_{+}(\boldsymbol{c}') = \mathcal{I}_{+}(\boldsymbol{c})\cup\mathcal{I}_{0}(\boldsymbol{c}), \\
        %         \mathcal{I}_{0}(\boldsymbol{c}') = \emptyset, ~\mathcal{I}_{-}(\boldsymbol{c}') = \mathcal{I}_{-}(\boldsymbol{c}). 
        %     \end{gather}
        % \end{subequations}
        Thus, we have ${K}_{+}' = |\mathcal{I}_{+}(\boldsymbol{c}')| = {K}_{+} + {K}_{0}$, ${K}_{0}' = |\mathcal{I}_{0}(\boldsymbol{c}')| = 0$, and ${K}_{-}' = |\mathcal{I}_{-}(\boldsymbol{c}')| = {K}_{-}$. 
        Under such conditions, it can be verified that $\nu'$ is given by 
        % \begin{subequations}
        %     \begin{align}
        %         \nu' & = \frac{1}{{K}_{+} + {K}_{0}}\left(
        %             P_{T} + \sigma^2\delta + \sigma^2\sum_{
        %                 i\in\mathcal{I}_{+}(\boldsymbol{c}')
        %             }{c_{i}}
        %         \right) \\
        %         & = \frac{1}{{K}_{+} + {K}_{0}}\left(
        %             {K}_{+}\nu + \sigma^2\delta + \sigma^2\sum_{i\in\mathcal{I}_{0}(\boldsymbol{c})}{u_i}
        %         \right) \\
        %         & = \nu + \frac{\sigma^2\delta}{{K}_{+} + {K}_{0}}, 
        %     \end{align}
        % \end{subequations}
        \begin{equation}
            \nu' = \nu + \sigma^2\delta / {K}_{+}', 
        \end{equation}
        % \begin{subequations}
        %     \begin{align}
        %         \nu' & = \frac{1}{{K}_{+} + {K}_{0}}\left(
        %             P_{T} + \sigma^2\delta + \sigma^2\sum_{
        %                 i\in\mathcal{I}_{+}(\boldsymbol{c}')
        %             }{c_{i}}
        %         \right) \\
        %         & = \nu + \sigma^2\delta / {K}_{+}', 
        %     \end{align}
        % \end{subequations}
        where the last equation results from~\eqref{eq:nu-expression} and the fact that $\sigma^2 c_{i} = \nu$ for $i\in\mathcal{I}_{0}(\boldsymbol{c})$. 
        Then, the power allocation vector $\boldsymbol{p}' = [p_{1}', \ldots, p_{K}']^T\in\mathbb{R}_{+}^{K\times 1}$ is given by 
        \begin{equation}
            p_{i}' = \left\{
                \begin{array}{ll}
                    \nu'/c_{i} - \sigma^2 & i\in\mathcal{I}_{+}(\boldsymbol{c}')\setminus\{k\}, \\
                    \nu'/({c}_{k} + \delta) - \sigma^2 & i = k, \\
                    0 & i\in\mathcal{I}_{-}(\boldsymbol{c}'),
                \end{array}
            \right. 
        \end{equation}
        where $\mathcal{I}_{+}(\boldsymbol{c}')\setminus\{k\}$ denotes the set excluding $k$ from $\mathcal{I}_{+}(\boldsymbol{c}')$. 
        Therefore, the instantaneous sum rate given $\boldsymbol{c}'$ is expressed as 
        \begin{subequations}
            \allowdisplaybreaks
            \begin{align}
                & \mathcal{R}(\boldsymbol{c}') = \sum_{
                    i\in\mathcal{I}_{+}(\boldsymbol{u'})
                }{\log_{2}\left(
                    1 + \frac{p_{i}'}{\sigma^2}
                \right)} \\
                & ~~~~ = \log_{2}\left(
                        \frac{\nu'/\sigma^2}{{c}_{k} + \delta}
                    \right) + \sum_{
                        i\in\mathcal{I}_{+}(\boldsymbol{c}')\setminus\{k\}
                    }{\log_{2}\left(
                        \frac{\nu'}{\sigma^2 c_{i}}
                    \right)}. 
            \end{align}
        \end{subequations}
        With $\mathcal{R}(\boldsymbol{c}) = \sum_{i\in\mathcal{I}(\boldsymbol{c})}{\log_{2}(\nu/(\sigma^2 {c}_{i}))}$, we have 
        \begin{subequations}\label{eq:inst-sum-rate-difference}
            \allowdisplaybreaks
            \begin{align}
                \begin{split}
                    \Delta\mathcal{R} & = \begin{aligned}[t]
                        & \mathcal{R}(\boldsymbol{c}') - \mathcal{R}(\boldsymbol{c}) = \log_{2}\left(
                            \frac{\nu'}{\sigma^2({c}_{k} + \delta)}
                        \right) \\
                        & - \log_{2}\left(
                            \frac{\nu}{\sigma^2 {c}_{k}}
                        \right) + \sum_{
                            i\in\mathcal{I}_{+}(\boldsymbol{c}')\setminus\{k\}
                        }{\log_{2}\left(
                            \frac{\nu'}{\nu}
                        \right)}. 
                    \end{aligned}
                \end{split} \\
                \begin{split}
                    & = \log_{2}\left(
                        \frac{{c}_{k}}{{c}_{k} + \delta}
                    \right) + {K}_{+}'\log_{2}\left(\frac{\nu'}{\nu}\right). 
                \end{split}
            \end{align}
        \end{subequations}
        As such, the limit in~\eqref{def:inst-sum-rate-derivative-to-uvec-lim} can be evaluated from the right hand side, i.e., $\delta\to 0^{+}$, as follows:
        \begin{subequations}\label{eq:inst-sum-rate-derivative-to-uvec-lim-rhs}
            \allowdisplaybreaks
            \begin{align}
                \begin{split}
                    \lim_{\delta\to 0^{+}}\frac{\Delta\mathcal{R}}{\delta} & = \lim_{\delta\to 0^{+}}\frac{1}{\delta}\log_{2}\left(
                        \frac{{c}_{k}}{{c}_{k} + \delta}
                    \right) + \frac{{K}_{+}'}{\delta}{\log_{2}\left(
                        \frac{\nu'}{\nu}
                    \right)}
                \end{split} \\
                \begin{split}
                    & = \frac{1}{\ln{2}}\left(
                        -\frac{1}{{c}_{k}} + \frac{\sigma^2}{\nu}
                    \right) = \frac{-p_{k}}{\nu\ln{2}}, 
                \end{split}
            \end{align}
        \end{subequations}
        where the last equation holds because $p_{k} = \nu/{c}_{k} - \sigma^2 > 0$ by assumption. 
        On the other hand, if $\delta < 0$, we have $\mathcal{I}_{+}(\boldsymbol{c}') = \mathcal{I}_{+}(\boldsymbol{c})$, $\mathcal{I}_{0}(\boldsymbol{c}') = \mathcal{I}_{0}(\boldsymbol{c})$, and $\mathcal{I}_{-}(\boldsymbol{c}') = \mathcal{I}_{-}(\boldsymbol{c})$ as $\delta\to 0^{-}$. 
        Therefore, it can be verified similarly following the above process that $\nu' = \nu + \sigma^2\delta/{K}_{+}$ and the limit in~\eqref{def:inst-sum-rate-derivative-to-uvec-lim} can be evaluated from the left hand side, i.e., $\delta\to 0^{-}$, as follows:
        \begin{equation}\label{eq:inst-sum-rate-derivative-to-uvec-lim-lhs}
            \lim_{\delta\to 0^{-}}\frac{\Delta\mathcal{R}}{\delta} = \lim_{\delta\to 0^{-}}\frac{1}{\delta}\log_{2}\left(
                \frac{{c}_{k}}{{c}_{k} + \delta}
            \right) + \frac{{K}_{+}}{\delta}{\log_{2}\left(
                \frac{\nu'}{\nu}
            \right)}, 
        \end{equation}
        which yields the same result as~\eqref{eq:inst-sum-rate-derivative-to-uvec-lim-rhs}, i.e., $-p_{k}/(\nu\ln{2})$. 
        Combining~\eqref{eq:inst-sum-rate-derivative-to-uvec-lim-rhs} and~\eqref{eq:inst-sum-rate-derivative-to-uvec-lim-lhs}, we come to the conclusion that $\lim_{\delta\to 0}{\Delta\mathcal{R}/\delta} = -p_{k}/(\nu\ln{2})$ for $k\in\mathcal{I}_{+}(\boldsymbol{c})$. 

        For the second case, assume $k\in\mathcal{I}_{0}(\boldsymbol{c})$, indicating $\nu = \sigma^2 {c}_{k}$ and $p_{k} = 0$. 
        If $\delta > 0$, we have $\nu < \sigma^2({c}_{k} + \delta)$ and thus $\nu' = \nu$, $p_{k}' = p_{k} = 0$, and $p_{i}' = p_{i}, \forall i\neq{k}$, which, in turn, lead to $\mathcal{R}(\boldsymbol{c}') = \mathcal{R}(\boldsymbol{c})$ and $\lim_{\delta\to 0^{+}}{\Delta\mathcal{R}/\delta} = 0$. 
        Meanwhile, if $\delta < 0$, it can be verified that this case is identical to the case where $k\in\mathcal{I}_{+}(\boldsymbol{c})$ and $\delta > 0$. 
        Therefore, we have % $\lim_{\delta\to 0^{-}}{\Delta\mathcal{R}/\delta} = -p_{k}/(\nu\ln{2}) = 0$. 
        \begin{equation}
            \lim_{\delta\to 0^{-}}{\frac{\Delta\mathcal{R}}{\delta}} = \frac{-p_{k}}{\nu\ln{2}} = 0. 
        \end{equation}
        Hence, the derivative $\lim_{\delta\to 0}{\Delta\mathcal{R}/\delta}$ is always $0$ for $k\in\mathcal{I}_{0}(\boldsymbol{c})$, which satisfies the result $-p_{k}/(\nu\ln{2})$ in the first case. 

        Finally, the case of $k\in\mathcal{I}_{-}(\boldsymbol{c})$ is considered. 
        Obviously, as $\delta\to 0$, we always have $\nu' = \nu < \sigma^2({c}_{k} + \delta)$ and thus $p_{k}' = p_{k} = 0$ and $p_{i}' = p_{i}$, $\forall i\neq{k}$. 
        Hence, the sum rate $\mathcal{R}(\boldsymbol{c}') = \mathcal{R}(\boldsymbol{c})$ remains unchanged and we have $\lim_{\delta\to 0}{\Delta\mathcal{R}/\delta} = 0 = -p_{k}/(\nu\ln{2})$, which also coincides with the first case. 

        Combining the three cases discussed above, the derivatives given in (24) are obtained.

    \vspace{-6pt}
    \subsection{Derivation of (25)}\label{appendix-subsec:uvec-derivative-to-position}
        Given vector $\boldsymbol{c} = \text{diag}((\boldsymbol{H}^H\boldsymbol{H})^{-1}) = \text{diag}(\boldsymbol{C}_{\boldsymbol{H}}^{-1})$, we have ${c}_{k} = [\boldsymbol{C}_{\boldsymbol{H}}^{-1}]_{kk}$ and 
        \begin{subequations}\label{eq:uvec-derivative-to-position-expanded}
            \begin{align}
                \frac{\mathrm{d}{{c}_{k}}}{\mathrm{d}{{v}_{n}}} & = \left[
                    \frac{\mathrm{d}{\boldsymbol{C}_{\boldsymbol{H}}^{-1}}}{\mathrm{d}{{v}_{n}}}
                \right]_{kk}
                    = \left[
                    - \boldsymbol{C}_{\boldsymbol{H}}^{-1}
                    \frac{\mathrm{d}{\boldsymbol{C}_{\boldsymbol{H}}}}{\mathrm{d}{{v}_{n}}}\boldsymbol{C}_{\boldsymbol{H}}^{-1}
                \right]_{kk} \\
                & = \left[
                    - \boldsymbol{C}_{\boldsymbol{H}}^{-1}\boldsymbol{\Psi}^{H}
                    \frac{\mathrm{d}{\boldsymbol{Q}\boldsymbol{Q}^{H}}}{\mathrm{d}{{v}_{n}}}
                    \boldsymbol{\Psi}\boldsymbol{C}_{\boldsymbol{H}}^{-1}
                \right]_{kk} \\
                & = -\boldsymbol{\xi}_{k}^{H}\frac{\mathrm{d}{\boldsymbol{Q}\boldsymbol{Q}^{H}}}{\mathrm{d}{{v}_{n}}}\boldsymbol{\xi}_{k}, 
            \end{align}
        \end{subequations}
        where $\boldsymbol{\xi}_{k}$ is the $k$-th column of matrix $\boldsymbol{\Psi}\boldsymbol{C}_{\boldsymbol{H}}^{-1}$ as defined in (25). 
        Note that for any $i, l$, we have $[\boldsymbol{Q}\boldsymbol{Q}^H]_{il} = \sum_{r = 1}^{{N}}{\exp(-j(\boldsymbol{\kappa}_{i} - \boldsymbol{\kappa}_{l})^T\boldsymbol{r}_{r})}$. 
        Taking the derivative w.r.t. ${v}_{n}$, the following result is obtained:  
        \begin{equation}\label{eq:QRmat-derivative-to-position}
            \left[
                \frac{\mathrm{d}{\boldsymbol{Q}\boldsymbol{Q}^{H}}}{\mathrm{d}{{v}_{n}}}
            \right]_{il} = -j({\kappa}_{i}^{v} - {\kappa}_{l}^{v})\exp\left(
                -j(\boldsymbol{\kappa}_{i} - \boldsymbol{\kappa}_{l})^T\boldsymbol{r}_{n}
            \right). 
        \end{equation}
        By defining matrix $\boldsymbol{\Lambda}^{v}$, it is easy to verify that the right-hand side of equation~\eqref{eq:QRmat-derivative-to-position} can be written as the element of matrix $-\text{diag}(\tilde{\boldsymbol{q}}_{n})\boldsymbol{\Lambda}^{v}\text{diag}(\tilde{\boldsymbol{q}}_{n}^{H})$ in the $i$-th row and $l$-th column, i.e., $[-\text{diag}(\tilde{\boldsymbol{q}}_{n})\boldsymbol{\Lambda}^{v}\text{diag}(\tilde{\boldsymbol{q}}_{n}^{H})]_{il}$, and thus we have $\frac{\mathrm{d}{\boldsymbol{Q}\boldsymbol{Q}^{H}}}{\mathrm{d}{{v}_{n}}} = -\text{diag}(\tilde{\boldsymbol{q}}_{n})\boldsymbol{\Lambda}^{v}\text{diag}(\tilde{\boldsymbol{q}}_{n}^{H})$. 
        Hence, the derivative of ${c}_{k}$ w.r.t. ${v}_{n}$ can be expressed as 
        \begin{equation}
            \frac{\mathrm{d}{{c}_{k}}}{\mathrm{d}{{v}_{n}}} = \boldsymbol{\xi}_{k}^{H}\text{diag}(\tilde{\boldsymbol{q}}_{n})\boldsymbol{\Lambda}^{v}\text{diag}(\tilde{\boldsymbol{q}}_{n}^{H})\boldsymbol{\xi}_{k}^{H}, ~\forall k, n, 
        \end{equation}
        which is equivalent to (25).

\section{Proof of Proposition 1}\label{appendix:asymp-ergo-sum-rate}
    % In this section, Proposition~\ref{prop:uvec-deterministic-equivalent} is proved. 
    Consider matrix $\boldsymbol{\Phi} \triangleq \frac{1}{K}\boldsymbol{H}^H\boldsymbol{H} + \gamma\boldsymbol{I}_{K}$ with $\gamma > 0$ and $\boldsymbol{\phi} = \text{diag}(\boldsymbol{\Phi}^{-1}) = [\phi_1, \ldots, \phi_{K}]^T\in\mathbb{R}^{K\times 1}$. 
    Since $\boldsymbol{C}_{\boldsymbol{H}} = \boldsymbol{H}^H\boldsymbol{H}$ is invertible, we have $\boldsymbol{C}_{\boldsymbol{H}}^{-1} = \frac{1}{K}\lim_{\gamma\to 0^{+}}{\boldsymbol{\Phi}^{-1}}$ and $~\boldsymbol{c} = \frac{1}{K}\lim_{\gamma\to 0^{+}}{\boldsymbol{\phi}}$. 
    The asymptotic forms for $\boldsymbol{\phi}$ given $\gamma > 0$ is first derived, denoted as $\boldsymbol{\phi}^{\infty}$, based on which $\boldsymbol{c}^{\infty}$ can be easily obtained by letting $\gamma\to 0^{+}$. 

    Let $\boldsymbol{H} = \left[\boldsymbol{h}_1, \boldsymbol{H}_{\sim{1}}\right]$, where $\boldsymbol{h}_1$ is the first column of matrix $\boldsymbol{H}$ and $\boldsymbol{H}_{\sim{1}} = [\boldsymbol{h}_2, \ldots, \boldsymbol{h}_{K}]\in\mathbb{C}^{{N}\times{(K - 1)}}$. 
    Then, matrix $\boldsymbol{\Phi}^{-1}$ can be rewritten as 
    \begin{equation}
        \boldsymbol{\Phi}^{-1} = \left[
            \begin{array}{cc}
                \frac{1}{K}\boldsymbol{h}_1^H\boldsymbol{h}_1 + \gamma & \frac{1}{K}\boldsymbol{h}_1^H\boldsymbol{H}_{\sim{1}} \\
                \frac{1}{K}\boldsymbol{H}_{\sim{1}}^H\boldsymbol{h}_1 & \frac{1}{K}\boldsymbol{H}_{\sim{1}}^H\boldsymbol{H}_{\sim{1}} + \gamma\boldsymbol{I}_{K - 1}
            \end{array}
        \right]^{-1}. 
    \end{equation}
    By applying the Schur components, the first element of vector $\boldsymbol{\phi}$ can be written as 
    \begin{equation}
            \phi_{1} = \bigg[
                \frac{1}{K}\boldsymbol{h}_1^H\boldsymbol{h}_1 + \gamma - \frac{1}{M^{2}}\boldsymbol{h}_1^H\boldsymbol{H}_{\sim{1}}\boldsymbol{\Phi}_{\sim{1}}^{-1}\boldsymbol{H}_{\sim{1}}^H\boldsymbol{h}_1
            \bigg]^{-1}, 
    \end{equation}
    where $\boldsymbol{\Phi}_{\sim{1}}$ is defined as $\boldsymbol{\Phi}_{\sim{1}} \triangleq \frac{1}{K}\boldsymbol{H}_{\sim{1}}^H\boldsymbol{H}_{\sim{1}} + \gamma\boldsymbol{I}_{K - 1}$. 
    Based on the Woodbury matrix identity~\cite{ref:woodbury-identity}, we have $\boldsymbol{I}_{K} - \frac{1}{K}\boldsymbol{H}_{\sim{1}}\boldsymbol{\Phi}_{\sim{1}}^{-1}\boldsymbol{H}_{\sim{1}}^H = \gamma(\gamma\boldsymbol{I}_{K} + \frac{1}{K}\boldsymbol{H}_{\sim{1}}\boldsymbol{H}_{\sim{1}}^H)^{-1}$ and thus  
    \begin{equation}\label{subeq:omegavec-reform-RH}
        \phi_{1} = \frac{1}{\gamma}\bigg[1 + \frac{1}{K}\boldsymbol{h}_1^H\Big(\gamma\boldsymbol{I}_{K} + \frac{1}{K}\boldsymbol{H}_{\sim{1}}\boldsymbol{H}_{\sim{1}}^H\Big)^{-1}\boldsymbol{h}_1\bigg]^{-1}. 
    \end{equation}
    It is noteworthy that if $\gamma = 0$,~\eqref{subeq:omegavec-reform-RH} would not hold true because the $N\times N$ matrix $\boldsymbol{H}_{\sim{1}}\boldsymbol{H}_{\sim{1}}^H$ has a rank of only $K - 1 < {N}$ and is thus singular, which manifests the benefit of considering $\boldsymbol{\phi}$ instead of $\boldsymbol{c}$ directly. 
    Then, all the elements of vector $\boldsymbol{\phi}$ can be expressed in a similar way. 
    Define $\boldsymbol{R} = \gamma\boldsymbol{I}_{{N}} + \frac{1}{K}\boldsymbol{H}\boldsymbol{H}^H$ and $\boldsymbol{R}_{\sim{k}} = \gamma\boldsymbol{I}_{{N}} + \frac{1}{K}\boldsymbol{H}_{\sim{k}}\boldsymbol{H}_{\sim{k}}^H$, $\forall k$, where matrix $\boldsymbol{H}_{\sim{k}}$ is given by 
    \begin{equation}
        \boldsymbol{H}_{\sim{k}} = \big[\boldsymbol{h}_{1}, \ldots, \boldsymbol{h}_{k - 1}, \boldsymbol{h}_{k + 1}, \ldots, \boldsymbol{h}_{K}\big]\in\mathbb{C}^{{N}\times{(K - 1)}}, 
    \end{equation}
    which is a sub-matrix of $\boldsymbol{H}$ by removing its $k$-th column. 
    Then, for any $k$, we have 
    \begin{equation}
        \phi_{k} = \frac{1}{\gamma}\bigg[1 + \frac{1}{K}\boldsymbol{h}_{k}^H\boldsymbol{R}_{\sim{k}}^{-1}\boldsymbol{h}_{k}\bigg]^{-1} = \frac{1}{\gamma}\bigg[1 + \frac{1}{K}\tilde{\boldsymbol{\psi}}_{k}^H\boldsymbol{V}_{k}\tilde{\boldsymbol{\psi}}_{k}\bigg]^{-1}, 
    \end{equation}
    where $\tilde{\boldsymbol{\psi}}_{k}\in\mathbb{C}^{L\times 1}$ is the $k$-th column of matrix $\boldsymbol{\Psi}$ and $\boldsymbol{V}_{k} = \boldsymbol{Q}\boldsymbol{R}_{\sim{k}}^{-1}\boldsymbol{Q}^H$. 
    Note that since random vectors $\tilde{\boldsymbol{\psi}}_{k}$, $\forall k$, are independent of each other, vector $\tilde{\boldsymbol{\psi}}_{k}$ is independent of the random matrix $\boldsymbol{V}_{k}$. 
    Besides, elements of $\tilde{\boldsymbol{\psi}}_{k}$ are independent of each other. 
    Thus, according to Theorem 3.4 in~\cite{ref:couillet-random-matrix-for-commun}, we have $\frac{1}{K}\tilde{\boldsymbol{\psi}}_{k}^H\boldsymbol{V}_{k}\tilde{\boldsymbol{\psi}}_{k} - \frac{1}{K}\text{tr}\left(\text{Diag}(\boldsymbol{b}_{k})\boldsymbol{V}_{k}\right)\to 0$ almost surely as $L_{k}^{t}\to\infty$, $\forall k$, with $\zeta_{k}$ being fixed for any $k$. 
    Thus, we have $\phi_{k} - \frac{1}{\gamma}\bigg[1 + \frac{1}{K}\text{tr}\left(\text{Diag}(\boldsymbol{b}_{k})\boldsymbol{V}_{k}\right)\bigg]^{-1} \to 0$ almost surely. 
    By noticing $\text{tr}(\text{Diag}(\boldsymbol{b}_{k})\boldsymbol{V}_{k}) = \text{tr}(\boldsymbol{G}_{k}\boldsymbol{R}_{\sim{k}}^{-1})$ with $\boldsymbol{G}_{k} = \boldsymbol{Q}^H\text{diag}(\boldsymbol{b}_{k})\boldsymbol{Q}$, this can be equivalently written as 
    \begin{equation}\label{eq:omegavec-DE-asymp}
        \phi_{k} - \frac{1}{\gamma}\left[
            1 + \frac{1}{K}\text{tr}\left(
                \boldsymbol{G}_{k}\boldsymbol{R}_{\sim{k}}^{-1}
            \right)
        \right]^{-1} \overset{\mathrm{a.s.}}{\longrightarrow} 0. 
    \end{equation}

    Next, the DE for random variable $e_{k} = \frac{1}{K}\text{tr}(\boldsymbol{G}_{k}\boldsymbol{R}_{\sim{k}}^{-1})$ is derived, which is denoted as $e_{k}^{\infty}$. 
    To this end, we further define $\varepsilon_{k,l} = \frac{1}{K}\text{tr}(\boldsymbol{G}_{l}\boldsymbol{R}_{\sim{k}}^{-1})$, $\forall k, l$, and $\varepsilon_{k,l}^{\infty}$ as their DEs, respectively, such that $\varepsilon_{k,l} - \varepsilon_{k,l}^{\infty}\to 0$ almost surely. 
    Given $\gamma$, define matrix 
    \begin{equation}\label{def:Upsilon-matrix-expression}
        \boldsymbol{\Upsilon}_{k} = \gamma\boldsymbol{I}_{{N}} + \frac{1}{K}\sum_{
            1\le l\le{K}, l\neq{k}
        }{
            \frac{
                \boldsymbol{G}_{l}
            }{
                1 + \varepsilon_{k,l}^{\infty}
            }
        }. 
    \end{equation}
    Then, the DE for $e_{k}$ is given by $e_{k}^{\infty} = \frac{1}{K}\text{tr}(\boldsymbol{G}_{k}\boldsymbol{\Upsilon}_{k}^{-1})$ and we have $e_{k} - e_{k}^{\infty}\to 0$ almost surely. 
    Specifically, it is easy to verify
    \begin{subequations}\label{eq:Bmat-DE-error}
        \allowdisplaybreaks
        \begin{align}
            & \boldsymbol{R}_{\sim{k}}^{-1} - \boldsymbol{\Upsilon}_{k}^{-1} = \boldsymbol{R}_{\sim{k}}^{-1}\left(
                \boldsymbol{\Upsilon}_{k} - \boldsymbol{R}_{\sim{k}}
            \right)\boldsymbol{\Upsilon}_{k}^{-1} \\
            % & ~~~~ = \boldsymbol{R}_{\sim{k}}^{-1}\left[
            %     \left(
            %         \boldsymbol{\Upsilon}_{k} - \gamma\boldsymbol{I}_{{N}}
            %     \right) - \frac{1}{K}\boldsymbol{H}_{\sim{k}}\boldsymbol{H}_{\sim{k}}^H
            % \right]\boldsymbol{\Upsilon}_{k}^{-1} \\
            & ~~~~ = \boldsymbol{T}_{k} - \boldsymbol{R}_{\sim{k}}^{-1}\left(
                \sum_{
                    1\le l\le{K}, l\neq{k}
                }{\frac{1}{K}\boldsymbol{h}_{l}\boldsymbol{h}_{l}^H}
            \right)\boldsymbol{\Upsilon}_{k}^{-1}, 
        \end{align}
    \end{subequations}
    where $\boldsymbol{T}_{k} = \boldsymbol{R}_{\sim{k}}^{-1}(\boldsymbol{\Upsilon}_{k} - \gamma\boldsymbol{I}_{{N}})\boldsymbol{\Upsilon}_{k}^{-1}$. 
    Define matrix $\boldsymbol{H}_{\sim{kl}}\in\mathbb{C}^{{{N}}\times{(K - 2)}}$ as a sub-matrix of $\boldsymbol{H}$ with the $k$-th and $i$-th columns removed and $\boldsymbol{R}_{\sim{kl}} = \gamma\boldsymbol{I}_{{N}} + \frac{1}{K}\boldsymbol{H}_{\sim{kl}}\boldsymbol{H}_{\sim{kl}}^H$, $\forall l\neq{k}$. 
    Then, we have $\boldsymbol{R}_{\sim{k}} = \boldsymbol{R}_{\sim{kl}} + \frac{1}{K}\boldsymbol{h}_{l}\boldsymbol{h}_{l}^H$ and equation~\eqref{eq:Bmat-DE-error} can be rewritten as equation~\eqref{eq:Bmat-DE-error-frac} shown at the bottom of the page, where the second equation is resulted from the Woodbury matrix identity~\cite{ref:woodbury-identity}. 
    \begin{figure*}[bp]
        \begin{equation}\label{eq:Bmat-DE-error-frac}
            \boldsymbol{R}_{\sim{k}}^{-1} - \boldsymbol{\Upsilon}_{k}^{-1} = \boldsymbol{T}_{k} - \sum_{
                1\le l\le{K}, l\neq{k}
            }{
                \left(
                    \boldsymbol{R}_{\sim{kl}} + \frac{1}{K}\boldsymbol{h}_{l}\boldsymbol{h}_{l}^H
                \right)^{-1}\frac{1}{K}\boldsymbol{h}_{l}\boldsymbol{h}_{l}^H\boldsymbol{\Upsilon}_{k}^{-1}
            } = \boldsymbol{T}_{k} - \sum_{
                1\le l\le{K}, l\neq{k}
            }\frac{
                \frac{1}{K}\boldsymbol{R}_{\sim{kl}}^{-1}\boldsymbol{h}_{l}\boldsymbol{h}_{l}^H\boldsymbol{\Upsilon}_{k}^{-1}
            }{
                1 + \frac{1}{K}\boldsymbol{h}_{l}^H\boldsymbol{R}_{\sim{kl}}^{-1}\boldsymbol{h}_{l}
            }. 
        \end{equation}
        % \vspace{-15pt}
    \end{figure*}
    Then, by applying an arbitrary deterministic matrix $\boldsymbol{G}\in\mathbb{C}^{{N}\times{{N}}}$ to the left-hand side of equation~\eqref{eq:Bmat-DE-error-frac}, we have 
    % the error between $e_{k}$ and $e_{k}^{\infty}$ is given by 
    \begin{equation}\label{eq:trace-asymp-error}
        \begin{aligned}
            & \frac{1}{K}\text{tr}\left(
                \boldsymbol{G}(\boldsymbol{R}_{\sim{k}}^{-1} - \boldsymbol{\Upsilon}_{k}^{-1})
            \right) = \frac{1}{K}\text{tr}(\boldsymbol{G}\boldsymbol{T}_{k}) \\
            & ~~~~~~~~~~~~~~ - \frac{1}{K}\sum_{
                1\le l\le{K}, l\neq{k}
            }{
                \frac{
                    \frac{1}{K}\boldsymbol{h}_{l}^H\boldsymbol{\Upsilon}_{k}^{-1}\boldsymbol{G}\boldsymbol{R}_{\sim{kl}}^{-1}\boldsymbol{h}_{l}
                }{
                    1 + \frac{1}{K}\boldsymbol{h}_{l}^H\boldsymbol{R}_{\sim{kl}}^{-1}\boldsymbol{h}_{l}
                }
            }, 
        \end{aligned}
    \end{equation}
    with $\boldsymbol{G} = \boldsymbol{G}_{k}$ for $e_{k} - e_{k}^{\infty}$. 
    Substituting the definition of $\boldsymbol{\Upsilon}_{k}$ and $\boldsymbol{T}_{k}$ into equation~\eqref{eq:trace-asymp-error}, the first term on the right-hand side of~\eqref{eq:trace-asymp-error} can be equivalently written as 
    \begin{equation}\label{eq:trace-asymp-expression}
        \frac{1}{K}\text{tr}(\boldsymbol{G}\boldsymbol{T}_{k}) = \frac{1}{K}\sum_{
                1\le l\le{K}, l\neq{k}
        }{
            \frac{
                \mathrm{tr}(\boldsymbol{G}\boldsymbol{R}_{\sim_{k}}^{-1}\boldsymbol{G}_{l}\boldsymbol{\Upsilon}_{k}^{-1})
            }{
                1 + \varepsilon_{k, l}^{\infty}
            }
        }. 
    \end{equation}
    Note that $\boldsymbol{h}_{l}$ is independent of matrices $\boldsymbol{\Upsilon}_{k}^{-1}\boldsymbol{G}\boldsymbol{R}_{\sim{kl}}^{-1}$ and $\boldsymbol{R}_{\sim{kl}}^{-1}$, $\forall l\neq{k}$. 
    Thus, by employing Theorem $3.4$ and $3.9$ in~\cite{ref:couillet-random-matrix-for-commun}, the asymptotic forms for both the denominator and numerator of each term in the summation in~\eqref{eq:trace-asymp-error} can be obtained as $\frac{1}{K}\text{tr}(\boldsymbol{G}_{l}\boldsymbol{\Upsilon}_{k}^{-1}\boldsymbol{G}\boldsymbol{R}_{\sim{k}}^{-1})$ and $1 + \frac{1}{K}\text{tr}(\boldsymbol{G}_{l}\boldsymbol{R}_{\sim{k}}^{-1})$, respectively, which coincide with those in equation~\eqref{eq:trace-asymp-expression}. 
    It can be verified following Section $6.2$ of~\cite{ref:couillet-random-matrix-for-commun} that for any $\boldsymbol{G}$, we have $\frac{1}{K}\text{tr}\left(\boldsymbol{G}\left(\boldsymbol{R}_{\sim{k}}^{-1} - \boldsymbol{\Upsilon}_{k}^{-1}\right)\right)\overset{\text{a.s.}}{\longrightarrow} 0$. 
    With $\boldsymbol{G} = \boldsymbol{G}_{k}$, we have $e_{k} - e_{k}^{\infty}\to 0$ almost surely. 
    Meanwhile, with $\boldsymbol{G} = \boldsymbol{G}_{l}$, $l\neq{k}$, we find $\varepsilon_{k, l} - \frac{1}{K}\text{tr}(\boldsymbol{G}_{l}\boldsymbol{\Upsilon}_{k}^{-1})\to 0$ almost surely, which forms a system of equations for $\varepsilon_{k, l}^{\infty}$, $\forall l\neq{k}$, given by 
    \begin{equation}\label{def:varepsilon-ml-fixed-point-eq}
        \varepsilon_{k, l}^{\infty} = \frac{1}{K}\text{tr}(\boldsymbol{G}_{l}\boldsymbol{\Upsilon}_{k}^{-1}), ~\forall k, l. 
    \end{equation}
    % \begin{equation}\label{def:varepsilon-ml-fixed-point-eq}
    %     \varepsilon_{k, l}^{\infty} = \frac{1}{K}\text{tr}\left(
    %         \boldsymbol{G}_{l}\left(
    %             \gamma\boldsymbol{I}_{{N}} + \sum_{
    %                 \substack{
    %                     1\le i\le{K}, \\
    %                     i\neq{k}
    %                 }
    %             }{
    %                 \frac{
    %                     \frac{1}{K}\boldsymbol{G}_{i}
    %                 }{
    %                     1 + \varepsilon_{k, i}^{\infty}
    %                 }
    %             }
    %         \right)^{-1}
    %     \right). 
    % \end{equation}
    As proved in Section $6.2$ of~\cite{ref:couillet-random-matrix-for-commun}, the solutions to the fixed-point equations in~\eqref{def:varepsilon-ml-fixed-point-eq} are unique, which can be solved by the Newton's method. 
    As such, matrix $\boldsymbol{\Upsilon}_{k}$ as well as $e_{k}^{\infty} = \frac{1}{K}\text{tr}(\boldsymbol{G}_{k}\boldsymbol{\Upsilon}_{k}^{-1})$, $\forall k$, can be obtained given the solved $\varepsilon_{k, l}^{\infty}$, $\forall l\neq{k}$, and $\boldsymbol{\phi}^{\infty}$ can be solved as $\phi_{k}^{\infty} = 1/(\gamma + \gamma{e}_{k}^{\infty})$. 
    % \begin{equation}
    %     \phi_{k}^{\infty} = \frac{1}{\gamma + \gamma{e}_{k}^{\infty}}. 
    % \end{equation}

    Finally, ${c}_{k}^{\infty}$ can be obtained by approaching the limit $\gamma\to 0^{+}$. 
    Note that the fixed-point equations in~\eqref{def:varepsilon-ml-fixed-point-eq} can be reformulated as follows:
    \begin{equation}
        \frac{1}{K}\frac{1}{\gamma + \gamma\varepsilon_{k, l}^{\infty}} = \left[
            \gamma{K} + \gamma\text{tr}\left(
                \boldsymbol{G}_{l}\boldsymbol{\Upsilon}_{k}^{-1}
            \right)
        \right]^{-1}, ~\forall k, l. 
    \end{equation}
    % $\epsilon_{k, l} \triangleq \frac{1}{K}\lim_{\gamma\to 0^{+}}\frac{1}{\gamma + \gamma\varepsilon_{k, l}^{\infty}}$
    By taking the limit $\gamma\to 0^{+}$ for both sides and defining the limit of the left side as $\epsilon_{k, l}$, it can be verified that we have $\lim_{\gamma\to 0^{+}}\gamma\boldsymbol{\Upsilon}_{k}^{-1} = \boldsymbol{Y}_{k}^{-1}$ and 
    \begin{equation}
        \epsilon_{k, l} = \left[
            \text{tr}\left(
                \boldsymbol{G}_{l}\boldsymbol{Y}_{k}^{-1}
            \right)
        \right]^{-1}, ~\forall k, l. 
    \end{equation}
    Thus, the system of equations in (30) is obtained. 
    Hence, ${c}_{k}^{\infty}$ is given by 
    \begin{equation}
        c_{k}^{\infty} = \frac{1}{K}\lim_{\gamma\to 0^{+}}\phi_{k}^{\infty} = \left[
            \text{tr}\left(
                \boldsymbol{G}_{k}\boldsymbol{Y}_{k}^{-1}
            \right)
        \right]^{-1}, ~\forall k, 
    \end{equation}
    which completes the proof of Proposition 1.

\section{Derivations of~\eqref{def:asymp-newton-update-details} and~\eqref{def:asymp-cvec-gradient-to-position}}
\label{appendix:asymp-derivative-expression}
    \subsection{Derivation of~\eqref{def:asymp-newton-update-details}}
    \label{appendix-subsec:asymp-newton-jacobian}
        First, we prove equation~\eqref{subdef:asymp-newton-fval}. 
        Given channel autocorrelation matrices $\boldsymbol{G}_{k} = \boldsymbol{Q}^{H}\mathrm{Diag}(\boldsymbol{b}_{k})\boldsymbol{Q}$, $\forall k$, define matrices $\boldsymbol{\mathcal{Y}}_{k}(\boldsymbol{\epsilon}) = \boldsymbol{I}_{N} + \sum_{i\neq k}{\epsilon_{i}\boldsymbol{G}_{i}}\in\mathbb{C}^{N\times N}$, $\forall k$. 
        It can be verified that 
        \begin{subequations}
            \begin{align}
                \boldsymbol{\mathcal{Y}}_{k}(\boldsymbol{\epsilon}) & = \boldsymbol{I}_{N} + \boldsymbol{Q}^{H}\sum_{i\neq k}{\mathrm{Diag}(\epsilon_{i}\boldsymbol{b}_{i})}\boldsymbol{Q} \\
                & = \boldsymbol{I}_{N} + \boldsymbol{Q}^{H}\mathrm{Diag}(\boldsymbol{B}_{k}^{T}\boldsymbol{\epsilon})\boldsymbol{Q}.
            \end{align}
        \end{subequations}
        Then, functions $g_{k, l}(\boldsymbol{\epsilon})$ can be equivalently rewritten as 
        \begin{equation}
            g_{k, l}(\boldsymbol{\epsilon}) = \mathrm{tr}\left(
                \epsilon_{l}\boldsymbol{G}_{l}\boldsymbol{\mathcal{Y}}_{k}(\boldsymbol{\epsilon})^{-1}
            \right) - 1 = \epsilon_{l}\boldsymbol{b}_{l}^{T}\mathrm{diag}(\boldsymbol{D}_{k}) - 1, 
        \end{equation}
        where $\boldsymbol{D}_{k} = \boldsymbol{Q}\boldsymbol{\mathcal{Y}}_{k}(\boldsymbol{\epsilon})^{-1}\boldsymbol{Q}^{H}$ as defined in equation~\eqref{subdef:asymp-newton-Dmat}. 
        Therefore, $\boldsymbol{\mathcal{G}}_{k}(\boldsymbol{\epsilon})$ is given by 
        \begin{equation}
            \boldsymbol{\mathcal{G}}_{k}(\boldsymbol{\epsilon}) = \mathrm{Diag}(\boldsymbol{\epsilon})\boldsymbol{B}^{T}\mathrm{diag}(\boldsymbol{D}_{k}) - \boldsymbol{1}_{K}, 
        \end{equation}
        which verifies equation~\eqref{subdef:asymp-newton-fval}. 

        Next, the Jacobian matrix $\boldsymbol{J}_{\boldsymbol{\mathcal{G}}_{k}}(\boldsymbol{\epsilon})$ is derived. 
        To begin with, $\boldsymbol{J}_{\boldsymbol{\mathcal{G}}_{k}}(\boldsymbol{\epsilon})$ is defined as 
        \begin{equation}
            \boldsymbol{J}_{\boldsymbol{\mathcal{G}}_{k}}(\boldsymbol{\epsilon}) = \left[
                \begin{array}{ccc}
                    \frac{\partial{g_{k, 1}(\boldsymbol{\epsilon})}}{\partial\epsilon_{1}} & \ldots & \frac{\partial{g_{k, 1}(\boldsymbol{\epsilon})}}{\partial\epsilon_{K}} \\
                    \vdots & \ddots & \vdots \\
                    \frac{\partial{g_{k, K}(\boldsymbol{\epsilon})}}{\partial\epsilon_{1}} & \ldots & \frac{\partial{g_{k, K}(\boldsymbol{\epsilon})}}{\partial\epsilon_{K}}
                \end{array}
            \right], 
        \end{equation}
        where $\frac{\partial{g_{k, l}(\boldsymbol{\epsilon})}}{\partial\epsilon_{i}}$ denotes the partical derivative of function $g_{k, l}(\boldsymbol{\epsilon})$ w.r.t. $\epsilon_{i}$. 
        Specifically, we have 
        \begin{equation}\label{def:newton-func-partial-deriv-to-epsilon}
            \frac{\partial{g_{k, l}(\boldsymbol{\epsilon})}}{\partial\epsilon_{i}} = \mathrm{tr}\left(
                \boldsymbol{G}_{l}\boldsymbol{\mathcal{Y}}_{k}(\boldsymbol{\epsilon})^{-1}
            \right)\frac{
                \partial\epsilon_{l}
            }{\partial\epsilon_{i}} + \epsilon_{l}\mathrm{tr}\left(
                \boldsymbol{G}_{l}\frac{
                    \partial\boldsymbol{\mathcal{Y}}_{k}(\boldsymbol{\epsilon})^{-1}
                }{\partial\epsilon_{i}}
            \right). 
        \end{equation}
        For convenience, denote the first and second terms in~\eqref{def:newton-func-partial-deriv-to-epsilon} as $\mathcal{P}_{kl, i}^{(1)}$ and $\mathcal{P}_{kl, i}^{(2)}$, respectively. 

        Consider first $i = k$. 
        Obviously, $\mathcal{P}_{kl, i}^{(2)}$ is reduced to $0$ since $\boldsymbol{\mathcal{Y}}_{k}(\boldsymbol{\epsilon})$ is independent of $\epsilon_{k}$. 
        Besides, $\mathcal{P}_{kl, i}^{(1)}$ is not zero if and only if $l = i = k$. 
        Thus, we have 
        \begin{subequations}
            \begin{align}
                \frac{\partial{g_{k, l}(\boldsymbol{\epsilon})}}{\partial\epsilon_{k}} & = \mathbbm{1}(l - k)\mathrm{tr}\left(
                    \boldsymbol{G}_{l}\boldsymbol{\mathcal{Y}}_{k}(\boldsymbol{\epsilon})^{-1}
                \right) \\
                & = \mathbbm{1}(l - k)\boldsymbol{b}_{l}^{T}\mathrm{diag}(\boldsymbol{D}_{k}), ~\forall l, 
            \end{align}
        \end{subequations}
        where $\mathbbm{1}(a)$ is the indicator function that equals to $1$ if and only if $a = 0$. 
        Otherwise, it equals to $0$. 

        Next, consider $i\neq k$. 
        Similarly, $\mathcal{P}_{kl, i}^{(1)}$ can be represented by $\mathbbm{1}(l - i)\boldsymbol{b}_{l}^{T}\mathrm{diag}(\boldsymbol{D}_{k})$. 
        Meanwhile, the derivative of matrix $\boldsymbol{\mathcal{Y}}_{k}(\boldsymbol{\epsilon})$ w.r.t. $\epsilon_{i}$ is given by 
        \begin{subequations}\label{eq:Ymat-derivative-to-epsilon}
            \begin{align}
                \frac{
                    \partial\boldsymbol{\mathcal{Y}}_{k}(\boldsymbol{\epsilon})^{-1}
                }{\partial\epsilon_{i}} & = -\boldsymbol{\mathcal{Y}}_{k}(\boldsymbol{\epsilon})^{-1}\frac{
                    \partial\boldsymbol{\mathcal{Y}}_{k}(\boldsymbol{\epsilon})
                }{\partial\epsilon_{i}}\boldsymbol{\mathcal{Y}}_{k}(\boldsymbol{\epsilon})^{-1} \\
                & = -\boldsymbol{\mathcal{Y}}_{k}(\boldsymbol{\epsilon})^{-1}\boldsymbol{G}_{i}\boldsymbol{\mathcal{Y}}_{k}(\boldsymbol{\epsilon})^{-1}. 
            \end{align}
        \end{subequations}
        As such, $\mathcal{P}_{kl, i}^{(2)}$ can be written as 
        \begin{subequations}\label{eq:newton-derivative-P2}
            \begin{align}
                \mathcal{P}_{kl, i}^{(2)} & = - \epsilon_{l}\mathrm{tr}\left(
                    \boldsymbol{G}_{l}\boldsymbol{\mathcal{Y}}_{k}(\boldsymbol{\epsilon})^{-1}\boldsymbol{G}_{i}\boldsymbol{\mathcal{Y}}_{k}(\boldsymbol{\epsilon})^{-1}
                \right) \\
                & = - \epsilon_{l}\mathrm{tr}\left(
                    \mathrm{Diag}(\boldsymbol{b}_{l})\boldsymbol{D}_{k}\mathrm{Diag}(\boldsymbol{b}_{i})\boldsymbol{D}_{k}
                \right) \\
                & = - \epsilon_{l}\boldsymbol{b}_{i}^{T}\left(
                    \boldsymbol{D}_{k}\odot\boldsymbol{D}_{k}^{T}
                \right)\boldsymbol{b}_{l}. 
            \end{align}
        \end{subequations}

        By defining $K\times 1$ vectors $\boldsymbol{\mathcal{P}}_{kl}^{(1)} = [\mathcal{P}_{kl, 1}^{(1)}, \ldots, \mathcal{P}_{kl, K}^{(1)}]^{T}$ and $\boldsymbol{\mathcal{P}}_{kl}^{(2)} = [\mathcal{P}_{kl, 1}^{(2)}, \ldots, \mathcal{P}_{kl, K}^{(2)}]^{T}$, we have $\partial{g_{k,l}(\boldsymbol{\epsilon})}/\partial{\boldsymbol{\epsilon}} = \boldsymbol{\mathcal{P}}_{kl}^{(1)} + \boldsymbol{\mathcal{P}}_{kl}^{(2)}$. 
        Then, based on the two cases discussed above, it can be verified that $\boldsymbol{\mathcal{P}}_{kl}^{(1)}$ and $\boldsymbol{\mathcal{P}}_{kl}^{(2)}$ can be equivalently written as 
        \begin{equation}
            \boldsymbol{\mathcal{P}}_{kl}^{(1)} = \boldsymbol{B}^{T}\mathrm{diag}(\boldsymbol{D}_{k})\odot\boldsymbol{e}_{l}, ~\boldsymbol{\mathcal{P}}_{kl}^{(2)} = -\epsilon_{l}\cdot\boldsymbol{\chi}_{l}^{k}, 
        \end{equation}
        where $\boldsymbol{e}_{i}\in\mathbb{R}^{K\times 1}$ is the vector that its $i$-th element equal to $1$ while all other elements is $0$, while $\boldsymbol{\chi}_{l}^{k} = \boldsymbol{B}_{k}^{T}(\boldsymbol{D}_{k}\odot\boldsymbol{D}_{k}^{T})\boldsymbol{b}_{l}\in\mathbb{R}^{K\times 1}$. 
        Note that the identity matrix of size $K\times K$ can be written as $\boldsymbol{I}_{K} = [\boldsymbol{e}_{1}, \ldots, \boldsymbol{e}_{K}]$. 
        Thus, we have 
        \begin{subequations}
            \begin{align}
                & \boldsymbol{J}_{\boldsymbol{\mathcal{G}}_{k}}(\boldsymbol{\epsilon}) = \left[
                    \frac{
                        \partial{g_{k,1}(\boldsymbol{\epsilon})}
                    }{
                        \partial{\boldsymbol{\epsilon}}
                    }, \ldots, \frac{
                        \partial{g_{k,K}(\boldsymbol{\epsilon})}
                    }{
                        \partial{\boldsymbol{\epsilon}}
                    }
                \right]^{T} \\
                & ~~~~ = \left[
                    \boldsymbol{\mathcal{P}}_{k1}^{(1)}, \ldots, \boldsymbol{\mathcal{P}}_{kK}^{(1)}
                \right]^{T} + \left[
                    \boldsymbol{\mathcal{P}}_{k1}^{(2)}, \ldots, \boldsymbol{\mathcal{P}}_{kK}^{(2)}
                \right]^{T} \\
                & ~~~~ = \mathrm{Diag}\left(
                    \boldsymbol{B}^{T}\mathrm{diag}(\boldsymbol{D}_{k})
                \right) - \mathrm{Diag}(\boldsymbol{\epsilon})\boldsymbol{X}_{k}, 
            \end{align}
        \end{subequations}
        where $\boldsymbol{X}_{k} = [\boldsymbol{\chi}_{1}^{k}, \ldots, \boldsymbol{\chi}_{K}^{k}]^{T} = \boldsymbol{B}^{T}(\boldsymbol{D}_{k}\odot\boldsymbol{D}_{k}^{T})\boldsymbol{B}_{k}$ as defined in~\eqref{subdef:asymp-newton-Xmat}, which completes the derivation for~\eqref{def:asymp-newton-update-details}.

    \subsection{Derivation of~\eqref{def:asymp-cvec-gradient-to-position}}
    \label{appendix-subsec:asymp-gradients-proof}
        % In this section, we prove the gradients of the asymptotic vector $\boldsymbol{c}^{\infty}$ w.r.t. antenna positions $\boldsymbol{x}$ and $\boldsymbol{y}$. 
        According to the definition of $c_{k}^{\infty}$ in\eqref{def:asymp-cvec-gradient-to-position}, we have 
        \begin{equation}\label{eq:asymp-cvec-derivative-to-position}
            \frac{
                \mathrm{d}{c_{k}^{\infty}}
            }{\mathrm{d}{\boldsymbol{v}}} = - \mathrm{tr}\left(
                \boldsymbol{G}_{k}\boldsymbol{Y}_{k}^{-1}
            \right)^{-2}\frac{\mathrm{d}}{\mathrm{d}{\boldsymbol{v}}}{
                \mathrm{tr}\left(
                    \boldsymbol{G}_{k}\boldsymbol{Y}_{k}^{-1}
                \right)
            }, 
        \end{equation}
        where $\boldsymbol{Y}_{k} = \boldsymbol{\mathcal{Y}}_{k}(\boldsymbol{\epsilon}_{k})$. 
        To derive the gradients given in~\eqref{def:asymp-cvec-gradient-to-position}, we first consider 
        \begin{equation}
            \mu_{kl, n}^{v} = \frac{\partial{
                \mathrm{tr}\left(\boldsymbol{G}_{k}\boldsymbol{Y}_{k}^{-1}\right)
            }}{\partial{{v}_{n}}}, ~\forall k, l, n, ~\forall v\in\{x, y\}, 
        \end{equation}
        and vectors $\boldsymbol{\mu}_{kl}^{v} = [\mu_{kl, 1}^{v}, \ldots, \mu_{kl, N}^{v}]\in\mathbb{R}^{N\times 1}$, $\forall k, l$. 
        Then, equation~\eqref{eq:asymp-cvec-derivative-to-position} can be rewritten as 
        \begin{equation}\label{eq:asymp-cvec-derivative-to-position-expanded}
            \frac{
                \mathrm{d}{c_{k}^{\infty}}
            }{\mathrm{d}{\boldsymbol{v}}} = - \left(
                c_{k}^{\infty}
            \right)^{2}\left[
                \boldsymbol{\mu}_{kk}^{v} + \sum_{l\neq k}{
                    \frac{\partial{
                        \mathrm{tr}\left(
                            \boldsymbol{G}_{k}\boldsymbol{Y}_{k}^{-1}
                        \right)
                    }}{\partial{\epsilon_{k, l}}}\frac{\mathrm{d}{
                        \epsilon_{k, l}
                    }}{\mathrm{d}{
                        \boldsymbol{v}
                    }}
                }
            \right]. 
        \end{equation}
        Moreover, by taking the total derivatives of both sides in equation $\boldsymbol{\mathcal{G}}_{k}(\boldsymbol{\epsilon}_{k}) = \boldsymbol{0}_{K\times 1}$, we have 
        \begin{equation}
            \boldsymbol{0}_{K\times 1} = \frac{
                \mathrm{d}{\boldsymbol{\mathcal{G}}_{k}(\boldsymbol{\epsilon}_{k})}
            }{\mathrm{d}{\boldsymbol{v}}}
            = \frac{
                \partial{\boldsymbol{\mathcal{G}}_{k}(\boldsymbol{\epsilon}_{k})}
            }{\partial{\boldsymbol{v}}} + \boldsymbol{J}_{\boldsymbol{\mathcal{G}}_{k}}(\boldsymbol{\epsilon}_{k})
            \cdot\frac{
                \mathrm{d}{\boldsymbol{\epsilon}_{k}}
            }{\mathrm{d}{\boldsymbol{v}}}, 
        \end{equation}
        where $\boldsymbol{J}_{\boldsymbol{\mathcal{G}}_{k}}(\boldsymbol{\epsilon})$ is the Jacobian matrix given in~\eqref{subdef:asymp-newton-jacobian}. 
        Thus, the derivative $\mathrm{d}{\boldsymbol{\epsilon}_{k}}/\mathrm{d}{\boldsymbol{v}}\in\mathbb{R}^{K\times N}$ is given by 
        \begin{equation}
            \frac{
                \mathrm{d}{\boldsymbol{\epsilon}_{k}}
            }{\mathrm{d}{\boldsymbol{v}}} = - \boldsymbol{J}_{\boldsymbol{\mathcal{G}}_{k}}(\boldsymbol{\epsilon}_{k})^{-1}\frac{
                \partial{\boldsymbol{\mathcal{G}}_{k}(\boldsymbol{\epsilon}_{k})}
            }{\partial{\boldsymbol{v}}},  
        \end{equation}
        from which $\mathrm{d}{\epsilon_{k, l}}/\mathrm{d}{\boldsymbol{v}}$ can be obtained. 
        Note that $\partial{\boldsymbol{\mathcal{G}}_{k}(\boldsymbol{\epsilon}_{k})}/\partial{\boldsymbol{v}}$ is the Jacobian matrix of $\boldsymbol{\mathcal{G}}_{k}(\boldsymbol{\epsilon_{k}})$ w.r.t. $\boldsymbol{v}$ and is given by 
        \begin{equation}
            \frac{
                \partial{\boldsymbol{\mathcal{G}}_{k}(\boldsymbol{\epsilon}_{k})}
            }{\partial{\boldsymbol{v}}} = \mathrm{Diag}(\boldsymbol{\epsilon}_{k})\cdot\left(\boldsymbol{U}_{k}^{v}\right)^{T}\in\mathbb{R}^{K\times N}, 
        \end{equation}
        where $\boldsymbol{U}_{k}^{v} = [\boldsymbol{\mu}_{k1}^{v}, \ldots, \boldsymbol{\mu}_{kK}^{v}]\in\mathbb{R}^{N\times K}$. 
        Therefore, to obtain the expression for gradients given in~\eqref{def:asymp-cvec-gradient-to-position}, it is essential to first derive $\boldsymbol{\mu}_{kl}^{v}$, $\forall k, l$, which is detailed in the following.

        \subsubsection{Derivation of $\boldsymbol{\mu}_{kl}^{v}$}
        \label{subsubsec:asymp-Fvec}
            For any $k$, $l$, and $n$, we have 
            \begin{equation}
                \mu_{kl, n}^{v} = \mathrm{tr}\left(
                    \frac{\partial{\boldsymbol{G}_{l}}}{\partial{{v}_{n}}}\boldsymbol{Y}_{k}^{-1}
                \right) + \mathrm{tr}\left(
                    \boldsymbol{G}_{l}\frac{\partial{\boldsymbol{Y}_{k}^{-1}}}{\partial{{v}_{n}}}
                \right), 
            \end{equation}
            where the first and second terms are denoted as $\mathcal{T}_{kl,n}^{(1)} = \mathrm{tr}\left(\frac{\partial{\boldsymbol{G}_{l}}}{\partial{{v}_{n}}}\boldsymbol{Y}_{k}^{-1}\right)$ and $\mathcal{T}_{kl,n}^{(2)} = \mathrm{tr}\left(\boldsymbol{G}_{l}\frac{\partial{\boldsymbol{Y}_{k}^{-1}}}{\partial{{v}_{n}}}\right)$, respectively. 
            Moreover, define $\boldsymbol{\mathcal{T}}_{kl}^{(1)} = \big[\mathcal{T}_{kl,1}^{(1)}, \ldots, \mathcal{T}_{kl,N}^{(1)}\big]^{T}\in\mathbb{R}^{N\times 1}$ and $\boldsymbol{\mathcal{T}}_{kl}^{(2)} = \big[\mathcal{T}_{kl,1}^{(2)}, \ldots, \mathcal{T}_{kl,N}^{(2)}\big]^{T}\in\mathbb{R}^{N\times 1}$, which are handled separately. 

            % The derivative of matrix $\boldsymbol{G}_{l}$ w.r.t. ${v}_{n}$ is given in the following lemma. 
            \begin{lemma}\label{lemma:corr-matrix-derivative-to-antenna-position}
                The derivative of matrix $\boldsymbol{G}_{l} = \boldsymbol{Q}^{H}\mathrm{Diag}(\boldsymbol{b}_{l})\boldsymbol{Q}$ w.r.t. ${v}_{n}$ can be written as  
                \begin{equation}
                    \frac{\mathrm{d}{\boldsymbol{G}_{l}}}{\mathrm{d}{{v}_{n}}} = \boldsymbol{\Omega}_{l, n}^{v} + (\boldsymbol{\Omega}_{l, n}^{v})^{H}, ~\forall l, n. 
                \end{equation}
                where $\boldsymbol{\Omega}_{l, n}^{v} = \big[\boldsymbol{0}_{N\times(n - 1)}, \boldsymbol{s}_{l, n}^{v}, \boldsymbol{0}_{N\times(N - n)}\big]\in\mathbb{C}^{N\times N}$. 
                The vector $\boldsymbol{s}_{l, n}^{v}\in\mathbb{C}^{N\times 1}$ is given by 
                \begin{equation}
                    \boldsymbol{s}_{l, n}^{v} = \boldsymbol{Q}^{H}\mathrm{Diag}(\boldsymbol{b}_{l})\mathrm{Diag}(j\boldsymbol{\vartheta}^{v})\tilde{\boldsymbol{q}}_{n}, 
                \end{equation}
                where $\boldsymbol{\vartheta}^{v} = [\kappa_{1}^{v}, \ldots, \kappa_{L}^{v}]^{T}$ as defined in Section~\ref{subsec:asymptotic-approximation}. 
            \end{lemma}
            \begin{proof}\label{proof-lemma:corr-matrix-derivative-to-antenna-position}
                Note that 
                \begin{equation}
                    \left[
                        \boldsymbol{G}_{l}
                    \right]_{mm'} = \sum_{i = 1}^{L}{
                        {b}_{li}\exp(j\boldsymbol{\kappa}_{i}^{T}(\boldsymbol{r}_{m'} - \boldsymbol{r}_{m}))
                    }, ~\forall m, m',  
                \end{equation}
                which leads to $\big[\frac{\mathrm{d}{\boldsymbol{G}_{l}}}{\mathrm{d}{{v}_{n}}}\big]_{mm'} = 0$ for $m, m'\neq n$. 
                Besides, it can be observed that $[\boldsymbol{G}_{l}]_{nn} = \boldsymbol{1}_{L}^{T}\boldsymbol{b}_{l}$ and thus $\big[\frac{\mathrm{d}{\boldsymbol{G}_{l}}}{\mathrm{d}{{v}_{n}}}\big]_{nn} = 0$. 
                For $m\neq m' = n$, we have 
                \begin{subequations}
                    \begin{align}
                        \left[
                            \frac{\mathrm{d}{\boldsymbol{G}_{l}}}{\mathrm{d}{{v}_{n}}}
                        \right]_{mn} & = \sum_{i = 1}^{L}{
                            j\kappa_{i}^{v}{b}_{li}\exp(j\boldsymbol{\kappa}_{i}^{T}(\boldsymbol{r}_{n} - \boldsymbol{r}_{m}))
                        } \\
                        & = \tilde{\boldsymbol{q}}_{m}^{H}\mathrm{Diag}(\boldsymbol{b}_{l})\mathrm{Diag}(j\boldsymbol{\vartheta}^{v})\tilde{\boldsymbol{q}}_{n}. 
                    \end{align}
                \end{subequations}
                Since $\boldsymbol{G}_{l}$ is an Hermitian matrix, we have $\big[\frac{\mathrm{d}{\boldsymbol{G}_{l}}}{\mathrm{d}{{v}_{n}}}\big]_{mn} = \big[\frac{\mathrm{d}{\boldsymbol{G}_{l}}}{\mathrm{d}{{v}_{n}}}\big]_{nm}^{*}$. 
                Then, by writing $\big[\frac{\mathrm{d}{\boldsymbol{G}_{l}}}{\mathrm{d}{{v}_{n}}}\big]_{nn} = 0 = [\boldsymbol{s}_{l, n}^{v}]_{n} - [\boldsymbol{s}_{l, n}^{v}]_{n}^{*}$, the lemma can be easily verified. 
            \end{proof}

            Following Lemma~\ref{lemma:corr-matrix-derivative-to-antenna-position}, $\mathcal{T}_{kl,n}^{(1)}$ can be written as 
            \begin{subequations}
                \begin{align}
                    \mathcal{T}_{kl,n}^{(1)} & = \mathrm{tr}\left(
                        \boldsymbol{\Omega}_{l,n}^{v}\boldsymbol{Y}_{k}^{-1}
                    \right) + \mathrm{tr}\left(
                        (\boldsymbol{\Omega}_{l,n}^{v})^{H}\boldsymbol{Y}_{k}^{-1}
                    \right) \\
                    & = \boldsymbol{\zeta}_{k,n}^{H}\boldsymbol{s}_{l, n}^{v} + (\boldsymbol{s}_{l, n}^{v})^{H}\boldsymbol{\zeta}_{k,n} \\
                    & = 2\mathrm{Re}\left(
                        \boldsymbol{\zeta}_{k,n}^{H}\boldsymbol{s}_{l, n}^{v}
                    \right), 
                \end{align}
            \end{subequations}
            where $\boldsymbol{\zeta}_{k,n}\in\mathbb{R}^{N\times 1}$ is the vector at the $n$-th column of matrix $\boldsymbol{Y}_{k}^{-1}$. 
            Thus, it can be verified that vector $\boldsymbol{\mathcal{T}}_{kl}^{(1)}$ is given by
            \begin{equation}
                \boldsymbol{\mathcal{T}}_{kl}^{(1)} = 2\mathrm{Re}\left[
                    \mathrm{diag}(\boldsymbol{Y}_{k}^{-1}\boldsymbol{S}_{l}^{v})
                \right], 
            \end{equation}
            where matrix $\boldsymbol{S}_{l}^{v} = [\boldsymbol{s}_{l, 1}^{v}, \ldots, \boldsymbol{s}_{l, N}^{v}]\in\mathbb{C}^{N\times N}$ is defined as 
            \begin{equation}
                \boldsymbol{S}_{l}^{v} = \boldsymbol{Q}^{H}\mathrm{Diag}(\boldsymbol{b}_{l})\mathrm{Diag}(j\boldsymbol{\vartheta}^{v})\boldsymbol{Q}, ~\forall l. 
            \end{equation}

            Next, the expression for vector $\boldsymbol{\mathcal{T}}_{kl}^{(2)}$ is derived. 
            According to the definition of $\boldsymbol{Y}_{k}$, we have 
            \begin{subequations}\label{eq:asymp-T2}
                \allowdisplaybreaks
                \begin{align}
                    \mathcal{T}_{kl, n}^{(2)} & = - \mathrm{tr}\left(
                        \boldsymbol{G}_{l}\boldsymbol{Y}_{k}^{-1}\frac{
                            \partial{
                                \boldsymbol{Y}_{k}
                            }
                        }{
                            \partial{{v}_{n}}
                        }\boldsymbol{Y}_{k}^{-1}
                    \right) \\
                    & = - \sum_{i\neq k}{
                        \epsilon_{k, i}\mathrm{tr}\left(
                            \boldsymbol{G}_{l}\boldsymbol{Y}_{k}^{-1}\frac{
                                \partial{
                                    \boldsymbol{G}_{i}
                                }
                            }{
                                \partial{{v}_{n}}
                            }\boldsymbol{Y}_{k}^{-1}
                        \right)
                    }. \label{subeq:asymp-T2-sum}
                \end{align}
            \end{subequations}
            Then, by applying Lemma~\ref{lemma:corr-matrix-derivative-to-antenna-position} again and similar to the derivation of $\boldsymbol{\mathcal{T}}_{kl}^{(1)}$, we have 
            \begin{subequations}
                \allowdisplaybreaks
                \begin{align}
                    \begin{split}
                        \boldsymbol{\mathcal{T}}_{kl}^{(2)} & = - \sum_{i\neq k}{
                            \epsilon_{k, i}\cdot2\mathrm{Re}\left[
                                \mathrm{diag}\left(
                                    \boldsymbol{Y}_{k}^{-1}\boldsymbol{G}_{l}\boldsymbol{Y}_{k}^{-1}\boldsymbol{S}_{i}^{v}
                                \right)
                            \right]
                        }
                    \end{split} \\
                    \begin{split}
                        & \begin{aligned}[t]
                            & = - \sum_{i\neq k}
                                2\mathrm{Re}\Big[
                                    \mathrm{diag}\Big(
                                        \boldsymbol{Y}_{k}^{-1}\boldsymbol{G}_{l}\boldsymbol{Y}_{k}^{-1}
                                        \boldsymbol{Q}^{H} \\
                            & ~~~~~~~~ ~~~~~~~~ ~~ \cdot\mathrm{Diag}(
                                    \epsilon_{k, i}\boldsymbol{b}_{i}
                                )\mathrm{Diag}(j\boldsymbol{\vartheta}^{v})\boldsymbol{Q}
                            \Big)
                            \Big]
                        \end{aligned}
                    \end{split} \\
                    \begin{split}
                        & \begin{aligned}[t]
                            & = - 2\mathrm{Re}\Big[
                                \mathrm{diag}\Big(
                                    \boldsymbol{Y}_{k}^{-1}\boldsymbol{Q}^{H}\mathrm{Diag}(\boldsymbol{b}_{l}) \\
                            & ~~~~~~~~ ~~~~~~ \cdot\tilde{\boldsymbol{D}}_{k}
                                    \mathrm{Diag}(\boldsymbol{B}_{k}\boldsymbol{\epsilon}_{k})\mathrm{Diag}(j\boldsymbol{\vartheta}^{v})\boldsymbol{Q}
                                \Big)
                            \Big], 
                        \end{aligned}
                    \end{split}
                \end{align}
            \end{subequations}
            where $\boldsymbol{B}_{k} = [\boldsymbol{b}_{1}, \ldots, \boldsymbol{b}_{k - 1}, \boldsymbol{0}_{L\times 1}, \boldsymbol{b}_{k + 1}, \ldots, \boldsymbol{b}_{K}]\in\mathbb{R}_{+}^{L\times K}$ and $\tilde{\boldsymbol{D}}_{k} = \boldsymbol{Q}\boldsymbol{Y}_{k}^{-1}\boldsymbol{Q}^{H}$, as defined in Section~\ref{subsec:asymptotic-approximation}. 

            Combining $\boldsymbol{\mathcal{T}}_{kl}^{(1)}$ and $\boldsymbol{\mathcal{T}}_{kl}^{(2)}$ together, vector $\boldsymbol{\mu}_{kl}^{v}$ can be written as 
            \begin{subequations}
                \allowdisplaybreaks
                \begin{align}
                    \begin{split}
                        \boldsymbol{\mu}_{kl}^{v} & = \boldsymbol{\mathcal{T}}_{kl}^{(1)} + \boldsymbol{\mathcal{T}}_{kl}^{(2)}
                    \end{split} \\
                    \begin{split}
                        & \begin{aligned}[t]
                            & = 2\mathrm{Re} \Big[
                                \mathrm{diag}\Big(
                                    \boldsymbol{Y}_{k}^{-1}\boldsymbol{Q}^{H}\mathrm{Diag}(\boldsymbol{b}_{l})\Big(\boldsymbol{I}_{L} - \\
                            & ~~~~~~~~ ~~~~ \tilde{\boldsymbol{D}}_{k}
                                    \mathrm{Diag}(\boldsymbol{B}_{k}^{T}\boldsymbol{\epsilon}_{k})
                                    \Big)\mathrm{Diag}(j\boldsymbol{\vartheta}^{v})\boldsymbol{Q}
                                \Big)
                            \Big]. 
                        \end{aligned}
                    \end{split} \\
                    \begin{split}
                        & = 2\mathrm{Re}\left(
                            \boldsymbol{F}_{k}^{v}\boldsymbol{b}_{l}
                        \right), ~\forall k, l, ~\forall v\in\{x, y\}, 
                    \end{split}
                \end{align}
            \end{subequations}
            where matrix $\boldsymbol{F}_{k}^{v}$ is defined as in equation~\eqref{subdef:asymp-cgrad-Fmat}, i.e., 
            \begin{equation}
                \boldsymbol{F}_{k}^{v} = \left[
                    \left(
                        \boldsymbol{I}_{L} - \tilde{\boldsymbol{D}}_{k}\mathrm{Diag}(\boldsymbol{B}_{k}^{T}\boldsymbol{\epsilon}_{k})
                    \right)
                    \mathrm{Diag}(j\boldsymbol{\vartheta}^{v})\boldsymbol{Q}
                \right]^{T}
                \odot\boldsymbol{E}_{k}, 
            \end{equation}
            with $\boldsymbol{E}_{k} = \boldsymbol{Y}_{k}^{-1}\boldsymbol{Q}^{H}$ as defined in equations~\eqref{def:asymp-deriv-auxiliary}.

        \subsubsection{Derivation of $\mathrm{d}{c_{k}^{\infty}} / \mathrm{d}{\boldsymbol{v}}$}
            Given $\boldsymbol{\mu}_{kl}^{v}$, we have $\boldsymbol{\mu}_{kk}^{v} = 2\mathrm{Re}(\boldsymbol{F}_{k}^{v}\boldsymbol{b}_{k})$ and $\boldsymbol{U}_{k}^{v} = 2\mathrm{Re}(\boldsymbol{F}_{k}^{v}\boldsymbol{B})$. 
            % Then, $\partial{\boldsymbol{\mathcal{G}}_{k}(\boldsymbol{\epsilon}_{k})}/\partial{\boldsymbol{v}}$ can be expressed as  
            % \begin{equation}
            %     \frac{
            %         \partial{\boldsymbol{\mathcal{G}}_{k}(\boldsymbol{\epsilon}_{k})}
            %     }{\partial{\boldsymbol{v}}} = \mathrm{Diag}(\boldsymbol{\epsilon}_{k})\left[
            %         2\mathrm{Re}\left(
            %             \boldsymbol{F}_{k}^{v}\boldsymbol{B}
            %         \right)
            %     \right]^{T}. 
            % \end{equation}
            Based on results from Appendix~\ref{appendix-subsec:asymp-newton-jacobian}, the term $\partial{\mathrm{tr}(\boldsymbol{G}_{k}\boldsymbol{Y}_{k}^{-1})}/\partial{\epsilon_{k, l}}$ in equation~\eqref{eq:asymp-cvec-derivative-to-position-expanded} is given by 
            \begin{equation}
                \frac{
                    \partial{\mathrm{tr}(\boldsymbol{G}_{k}\boldsymbol{Y}_{k}^{-1})}
                }{
                    \partial{\epsilon_{k, l}}
                } = \left\{
                    \begin{aligned}
                        & -\boldsymbol{b}_{l}^{T}\left(
                            \tilde{\boldsymbol{D}}_{k}\odot\tilde{\boldsymbol{D}}_{k}^{T}
                        \right)\boldsymbol{b}_{k}, & ~l\neq k, \\
                        & 0, & ~l = k. 
                    \end{aligned}
                \right.
            \end{equation}
            Therefore, we have $\partial{\mathrm{tr}(\boldsymbol{G}_{k}\boldsymbol{Y}_{k}^{-1})}/\partial{\boldsymbol{\epsilon}_{k}} = -\tilde{\boldsymbol{\chi}}_{k}$, where 
            \begin{equation}
                \tilde{\boldsymbol{\chi}}_{k} = \boldsymbol{B}_{k}^{T}\left(
                    \tilde{\boldsymbol{D}}_{k}\odot\tilde{\boldsymbol{D}}_{k}^{T}
                \right)\boldsymbol{b}_{k} = \boldsymbol{Z}_{k}^{T}\boldsymbol{b}_{k}, ~\forall k, 
            \end{equation}
            where $\boldsymbol{Z}_{k} = (\tilde{\boldsymbol{D}}_{k}\odot\tilde{\boldsymbol{D}}_{k}^{T})\boldsymbol{B}_{k}$ as defined in equations~\eqref{def:asymp-deriv-auxiliary}. 
            Thus, it can be verified that we have 
            the second term in the square brackets in~\eqref{eq:asymp-cvec-derivative-to-position-expanded} can be equivalently written as 
            \begin{subequations}
                \begin{align}
                    & \sum_{l\neq k}{
                        \frac{\partial{
                            \mathrm{tr}\left(
                                \boldsymbol{G}_{k}\boldsymbol{Y}_{k}^{-1}
                            \right)
                        }}{\partial{\epsilon_{k, l}}}\frac{\mathrm{d}{
                            \epsilon_{k, l}
                        }}{\mathrm{d}{
                            \boldsymbol{v}
                        }}
                    } = \left[
                        \frac{\mathrm{d}{
                            \boldsymbol{\epsilon}_{k}
                        }}{\mathrm{d}{
                            \boldsymbol{v}
                        }}
                    \right]^{T}
                    \frac{\partial{
                        \mathrm{tr}\left(
                            \boldsymbol{G}_{k}\boldsymbol{Y}_{k}^{-1}
                        \right)
                    }}{\partial{\boldsymbol{\epsilon}_{k}}} \\
                    & ~~~~ ~~~~ = \boldsymbol{U}_{k}^{v}\mathrm{Diag}(\boldsymbol{\epsilon}_{k})\left[
                        \boldsymbol{J}_{\boldsymbol{\mathcal{G}}_{k}}(\boldsymbol{\epsilon}_{k})^{-1}
                    \right]^{T}\tilde{\boldsymbol{\chi}}_{k}. 
                \end{align}
            \end{subequations}
            By integrating all the results above, the derivative in equation~\eqref{eq:asymp-cvec-derivative-to-position-expanded} can be expressed by 
            \begin{subequations}
                \begin{align}
                    & \frac{
                        \mathrm{d}{c_{k}^{\infty}}
                    }{\partial{\boldsymbol{v}}} = - \left(c_{k}^{\infty}\right)^{2}\left(
                        \boldsymbol{R}_{k}^{v}\tilde{\boldsymbol{\chi}}_{k} + \tilde{\boldsymbol{\mu}}_{k}^{v}
                    \right), \\
                    & \boldsymbol{R}_{k}^{v} = \boldsymbol{U}_{k}^{v}\mathrm{Diag}(\boldsymbol{\epsilon}_{k})\left[
                        \boldsymbol{J}_{\boldsymbol{\mathcal{G}}_{k}}(\boldsymbol{\epsilon}_{k})^{-1}
                    \right]^{T}, \\
                    & \tilde{\boldsymbol{\mu}}_{k}^{v} = 2\mathrm{Re}(\boldsymbol{F}_{k}^{v}\boldsymbol{b}_{k}), ~\boldsymbol{U}_{k}^{v} = 2\mathrm{Re}(\boldsymbol{F}_{k}^{v}\boldsymbol{B}), 
                \end{align}
            \end{subequations}
            which is equivalent to~\eqref{def:asymp-cvec-gradient-to-position} in Section~\ref{subsec:asymptotic-approximation}.

\vspace{-6pt}

\bibliographystyle{IEEEtran} % use IEEEtran.bst style
\bibliography{IEEEabrv, reference}

\end{document}